\documentclass[11pt]{article}
\usepackage{graphicx}
\usepackage{color}
\usepackage{pdfsync}
\usepackage{times}
\usepackage{amssymb}
\usepackage{amsmath}
\usepackage{amsthm}
\usepackage{epsfig}
\usepackage{cite}
\usepackage{verbatim}
\usepackage{enumerate}
\usepackage{subfigure}
\usepackage{multicol}
\usepackage{amsmath}
\usepackage{graphicx}
\usepackage{setspace}
\usepackage{algpseudocode}
\usepackage{algorithm}
\usepackage{cases}
\usepackage{tikz}
\usepackage{url}
\usepackage{blkarray}
\makeatletter 
\newbox\BA@first@box
\makeatother

\newtheorem{theorem}{\bf Theorem}
\newtheorem{property}{\bf Property}

\newtheorem{lemma}{\bf Lemma}

\newtheorem{assumption}{\bf Assumption}
\newtheorem{definition}{\bf Definition}
\newtheorem{proposition}{\bf Proposition}
\newtheorem{remark}{\bf Remark}
\newtheorem{claim}{\bf Claim}
\newtheorem*{claim*}{\bf Claim}

\newcommand{\Head}{\text{Head}}
\newcommand{\Tail}{\text{Tail}}
\newcommand{\In}{\text{In}}
\newcommand{\Out}{\text{Out}}
\newcommand{\Ord}{\text{Ord}}
\newcommand{\Span}{\textrm{colspan}}
\newcommand{\rank}{\text{rank}}
\newcommand{\Grank}{\text{Grank}}
\newcommand{\Granki}[1]{\Grank\left(\mathbf{H}_1^{T_1^{(#1)}}, \mathbf{H}_2^{T_1^{(#1)}}, \mathbf{G}_2^{T_2^{(#1)}}\right)}
\newcommand{\mA}{\mathbf{A}}
\newcommand{\mB}{\mathbf{B}}
\newcommand{\mH}[1]{\mathbf{H}_{#1}}
\newcommand{\mG}[1]{\mathbf{G}_{#1}}
\newcommand{\HT}[2]{\mathbf{H}_{#1}^{T_1^{(#2)}}}
\newcommand{\GT}[2]{\mathbf{G}_{#1}^{T_2^{(#2)}}}
\newcommand{\HU}[1]{\mathbf{H}_{#1}^{U_1^{(i)}}}

\newcommand{\HI}[1]{\mathbf{H}_{#1}^{I_1^{(i)}}}
\newcommand{\GI}[1]{\mathbf{G}_{#1}^{I_2^{(i)}}}

\newcommand{\GU}[1]{\mathbf{G}_{#1}^{U_2^{(i)}}}

\newcommand{\HO}[1]{\mathbf{H}_{#1}^{O_1^{(i)}}}
\newcommand{\GA}[1]{\mathbf{G}_{#1}^{A_2^{(i)}}}

\newcommand{\GB}[1]{\mathbf{G}_{#1}^{B_2^{(i)}}}

\newcommand{\HOI}[2]{\mathbf{H}_{#1}^{O^{(i)}_{1,#2}}}

\newcommand{\GAI}[2]{\mathbf{G}_{#1}^{A^{(i)}_{2,#2}}}

\newcommand{\FIj}[2]{\mathbf{F}^{I_{#1}^{(i)}, \left\{ e_{i,#2} \right\}}}

\newcommand{\HTall}{
    \begin{bmatrix}
        \HT{1}{i} \\ \HT{2}{i} 
    \end{bmatrix}
}

\newcommand{\HIall}{
    \begin{bmatrix}
        \HI{1} \\ \HI{2} 
    \end{bmatrix}
}

\newcommand{\HUall}{
    \begin{bmatrix}
        \HU{1} \\ \HU{2} 
    \end{bmatrix}
}

\newcommand{\HOall}{
    \begin{bmatrix}
        \HO{1} \\ \HO{2} 
    \end{bmatrix}
}

\newcommand{\HUIall}{
    \begin{bmatrix}
        \HU{1} & \HI{1}  \\
        \HU{2} & \HI{2} 
    \end{bmatrix}
}

\newcommand{\HUOall}{
    \begin{bmatrix}
        \HU{1} & \HO{1}  \\
        \HU{2} & \HO{2} 
    \end{bmatrix}
}

\newcommand{\HGall}{
    \begin{bmatrix}
        \HT{2}{i} & \GT{2}{i}  \\
    \end{bmatrix}
}

\title{Alignment based Network Coding for Two-Unicast-Z Networks }
\author{
Weifei Zeng\\
RLE, MIT\\
{Cambridge, MA, USA}\\
	{\tt weifei@mit.edu}
	\and
	Viveck R. Cadambe \\
	Department of Electrical Engineering\\
	The Pennsylvania State University \\
	{\tt viveck@engr.psu.edu}
\and
Muriel M\'{e}dard \\
RLE, MIT\\
{Cambridge, MA, USA}\\
	{\tt medard@mit.edu}
\thanks{This material is based upon work supported by the Air Force Office of Scientific Research (AFOSR) under award
No. FA9550-13-1-0023.}
\thanks{This work was also supported by Viveck R. Cadambe's startup grant provided by the Department of Electrical
Engineering at the Pennsylvania State University.}
}

\begin{document}
\date{}
\maketitle
\fontfamily{cmr}
\selectfont

\begin{abstract}
    In this paper, we study the wireline \emph{two-unicast-Z} communication network over directed acyclic graphs. The
    two-unicast-$Z$ network is a two-unicast network where the destination intending to decode the second message has
    apriori side information of the first message. We make three contributions in this paper:
	\begin{enumerate}
    \item We describe a new linear network coding algorithm for two-unicast-Z networks over directed acyclic graphs. Our
        approach includes the idea of interference alignment as one of its key ingredients. 
        {
        For graphs of a bounded
        degree, our algorithm has linear complexity in terms of the number of vertices, and
        polynomial complexity in terms of the number of edges. } %based on the intuition that maximizing the linear coding capacity of the two-unicast-Z network requires an
\item We prove that our algorithm achieves the rate-pair $(1,1)$ whenever it is feasible in the network. Our proof
    serves as an alternative, albeit restricted to two-unicast-Z networks over directed acyclic graphs, to an earlier
    result of Wang et. al. which studied necessary and sufficient conditions for feasibility of the rate pair $(1,1)$ in
    two-unicast networks.  
\item We provide a new proof of the classical max-flow min-cut theorem for directed acyclic graphs.  
	\end{enumerate}

\end{abstract}
\newpage
\section{Introduction}
Since the advent of network coding \cite{Ahlswede_Cai_etal}, characterizing the capacity region of networks of
orthogonal noiseless capacitated links, often termed the network coding capacity, has been an active area of research.
Inspired by the success of \emph{linear} network coding for \emph{multicast} networks \cite{Koetter_Medard, Li_Yeung_Cai:linear}, a
significant body of work has been devoted to understanding the design and performance limits of linear network codes even for non-multicast communication scenarios.
Previous approaches to linear network code design for non-multicast settings have at least one of two
drawbacks: the network code design is restricted to a limited set of network topologies, or the approach has
a prohibitive computational complexity.  The goal of our paper is develop ideas and algorithms that fill this gap in literature. 
The main contribution of our paper is the development of a low-complexity linear network coding algorithm
for \emph{two-unicast-Z} networks - a network communication setting with two independent message sources and two corresponding destination nodes, where
one destination has a priori knowledge of the undesired message (See Fig. \ref{fig:twounicast-Z}). We begin with a
brief survey of previous, related literature.

\subsection{Related Work}

The simplest linear network code is in fact the technique of routing, which is used to show the max-flow min-cut theorem. In addition to the single-source single-destination setting of the max-flow min-cut theorem, routing has been shown to be optimal for several classes of networks with multiple source
messages in \cite{TCHu, Schrijver, Meng_routing}\footnote{In fact, in \cite{Li_Li_Conjecture}, routing has been
conjectured to achieve network capacity for multiple unicast networks over undirected directed acyclic graphs.}.  The
technique of random linear network coding, which is optimal for multicast networks \cite{Koetter_Medard, Ho_etal_Award},
is also shown to achieve capacity for certain non-multicast networks \cite{Koetter_Medard,Erez_Feder,
Yeung_NetworkCoding}\footnote{Linear network coding techniques which do not necessarily involve choosing co-efficients
    randomly have also been studied for multicast and certain non-multicast settings in \cite{Li_Yeung_Cai:linear,
    Yeung_NetworkCoding}}. The ideas of random linear network coding and routing have been combined to develop network
    coding algorithms for arbitrary networks in \cite{Kim_evolutionary}. Nonetheless, it is well known that the
    techniques of random linear network coding and routing are, in general, sub-optimal linear network codes for
    networks with multiple sources.

To contrast the optimistic results of \cite{Koetter_Medard, Li_Yeung_Cai:linear, Erez_Feder, TCHu, Li_Li_Conjecture,
Meng_routing}, pessimistic results related to the performance limits of linear network codes have been shown in
\cite{dougherty_insufficiency_2005, Blasiak_Kleinberg_Lubetzky, kamath_twounicastishard}. Through deep
connections between linear network coding and matroidal theory, it has been shown that linear network codes are
sub-optimal in general \cite{dougherty_insufficiency_2005, dougherty_networks_2007}. In fact, most recently
\cite{kamath_twounicastishard}, it has been shown that even the best linear code cannot, in general, achieve the
capacity of the \emph{two-unicast} network. 
%where there are two independent message sources and two corresponding
%destinations, each intending to decode to the message generated by the corresponding source. 

Literature has also studied the computational limits of construction and performance evaluation
of linear codes. It is shown in \cite{lehman2004complexity} that characterizing the set of rates achievable by
\emph{scalar} linear codes\footnote{We use the term \emph{linear coding} for network codes where all the codewords come
from a vector space, and encoding functions are linear operators over this vector space. The term \emph{scalar} linear
network coding is used when the dimensionality of this vector space is equal to $1$. It is known that vector linear
network coding strictly outperforms scalar linear network coding in general \cite{Medard_Effros}.}  is NP-complete.
Furthermore, for a given field size, the computational complexity of determining the best (scalar or vector) linear code
is associated with challenging open problems related to polynomial solvability
\cite{Dougherty_Freiling_Zeger_complexity} and graph coloring \cite{Langberg_Sprinston}. Nonetheless, 
characterization of the computational complexity of determining the set of all rates achievable by linear network coding remains an open problem. 
{
One main challenge in resolving this question is to develop ideas and algorithms for constructing linear network
codes for non-multicast settings.} The study and development of linear network codes for non-multicast settings is also
important from an engineering standpoint, especially in current times, because its utility has been recently
demonstrated in settings that arise from modeling distributed storage systems \cite{Cadambe_Jafar_Maleki}, wireless
networks \cite{Jafar_Topological}, and content disribution systems \cite{Niesen_Mohammad}. We review the principal
approaches to developing linear network codes outside the realm of routing and random linear network coding. 
%linear code is closely to solvability of a set of polynomial equations \cite{Dougherty_Freiling_Zeger_complexity}, and
%to determing the representability of matroids \cite{dougherty_networks_2007}.

\textbf{Network Codes over the binary field:}
Initial approaches to developing codes for non-multicast settings restricted their attention to codes over the binary
field. Since the binary field provides a small set of choices in terms of the linear combinations that can be obtained
by an encoding node, it is possible to search over the set of all coding solutions relatively efficiently. This idea was
exploited to develop a linear programming based network coding algorithms in \cite{traskov_koetter, Ho_constructive}.
Reference \cite{CCWang:TwoUnicast} presents a noteworthy result that demonstrates the power of carefully designed
network codes over the binary field. Through a careful understanding of the network communication graph, 
\cite{CCWang:TwoUnicast} characterized the feasibility of rate tuple $(1,1)$ in two unicast networks. The result of
\cite{CCWang:TwoUnicast} can be interpreted as follows: the rate tuple $(1,1)$ is achievable if and only if (i) the
min-cut between each source and its respective destination is at least equal to $1$, and (ii) the \emph{generalized
network sharing bound} - a network capacity outer bound formulated in \cite{Kamath_Tse_Ananthram} - is at least $2$.
Furthermore, the rate $(1,1)$ is feasible if and only if it is feasible through linear network coding over the binary
field. In general, however, coding over the binary field does not suffice even for the two unicast network
\cite{Kamath_Tse_Ananthram}.
%The role of neutralization has also been identified in uncovering the sum-capacity of a class of networks in
%\cite{Zeng_TwoUnicastZ}.

\textbf{Interference Alignment:} \emph{Interference alignment,} which is a technique discovered in the context of
interference management for wireless networks, has recently emerged as a promising tool for linear network code design
even for wireline networks. Specifically, an important goal in the design of coding co-efficients for non-multicast
networks is to ensure that an unwanted message does not corrupt a desired message at a destination. The unwanted message
can be interpreted as an interferer at the destination, and thus, tools from interference management in wireless
networks can be inherited into the network coding setting. 

Interference alignment was first explicitly identified as a tool for multiple unicast network coding in
\cite{Das_etal_ISIT2010}, which describes a class of networks where the asymptotic interference alignment scheme of
\cite{Cadambe_Jafar_int} is applicable. This class of networks has been further studied and generalized in
\cite{Meng_Das_TransIT, Bavirisetti_etal}.  The power of interference alignment for network coding was demonstrated in
references \cite{Wu_Dimakis,Cadambe_Jafar_Maleki, Cadambe_Huang_Jafar_Li, Dimitris_Dimakis_Cadambe}; these references 
developed alignment-based erasure codes to solve open problems related to minimizing repair bandwidth in distributed data storage
systems. Reference \cite{Maleki_Cadambe_Jafar} used interference alignment to characterize the capacity of classes of the \emph{index coding} problem \cite{Birk_Kol_INFOCOM}, a sub-class of the
class of general network coding capacity problems. The index coding problem is especially important because references \cite{Rouayheb_Sprinston_Georghiades, Effros_Equivalence}
have shown an equivalence between the general network coding capacity problem and the index coding problem. %network code for an arbitrary network can be mapped to an equivalent linear index code for a carefully constructed the
%general network coding capacity problem. 

%A first step to answer to this question relies on understanding the design and performance of linear network codes for
%networks. The classical techniques of routing and random linear network coding are known to be insufficient in general
%in finding optimal network coding solutions. In general, wireline networks require linear network codes whose coding
%coefficients depend on the structure of the network.  
\textbf{Index Coding Based Approaches:} The equivalence of the general network coding problem and the index coding
problem, which is established in \cite{Rouayheb_Sprinston_Georghiades, Effros_Equivalence}, opens another door to the
development of linear network codes for general networks. Specifically, for a given network coding setting,  the
approach of \cite{Effros_Equivalence} can be used first to obtain an equivalent index coding setting; then the
approaches of \cite{blasiak, shanmugam_coloring, Kim_indexcoding, kim_timesharing} can be used to develop linear
\emph{index} codes. 
{
However, these 
%of \cite{blasiak, shanmugam_coloring, Kim_indexcoding, kim_timesharing} 
index code design approaches have high computational complexity, since they require solving challenging graph coloring related
problems or linear programs whose number of constraints is exponential in terms of the number of users. 
}
Because of the nature of the mapping of \cite{Effros_Equivalence}, 
this means that 
these approaches require solving
linear programs where the number of constraints is exponential number in terms of the number of edges of the network in
consideration.  
Another common approach to obtaining index coding solutions is given in \cite{bar2011index},  which
connects the rate achievable via linear index coding to a graph functional known as \emph{minrank}. While in principle,
the minrank characterizes the rate achievable by the best possible linear index code, the min-rank of a matrix over a
given field size is difficult to evaluate, and NP-Hard in general \cite{Minrank_NPhard}; 
{
furthermore, there is no
systematic approach to characterizing the field size.
}

In summary, while recent ideas of interference alignment and connections to index coding broaden the scope of linear
network coding, these approaches inherit the main drawbacks of linear network
coding. Specifically, network coding approaches outside the realm of routing or random linear network coding are either
carefully hand-crafted for a restricted set of network topologies, or their enormous computational complexity inhibits
their utility. The motivation of our paper is to partially fill this gap in literature by devising algorithms for linear
network coding. We review our contributions next. 
\begin{figure}
	\begin{center}
		\includegraphics[height=2.0in, width=4in]{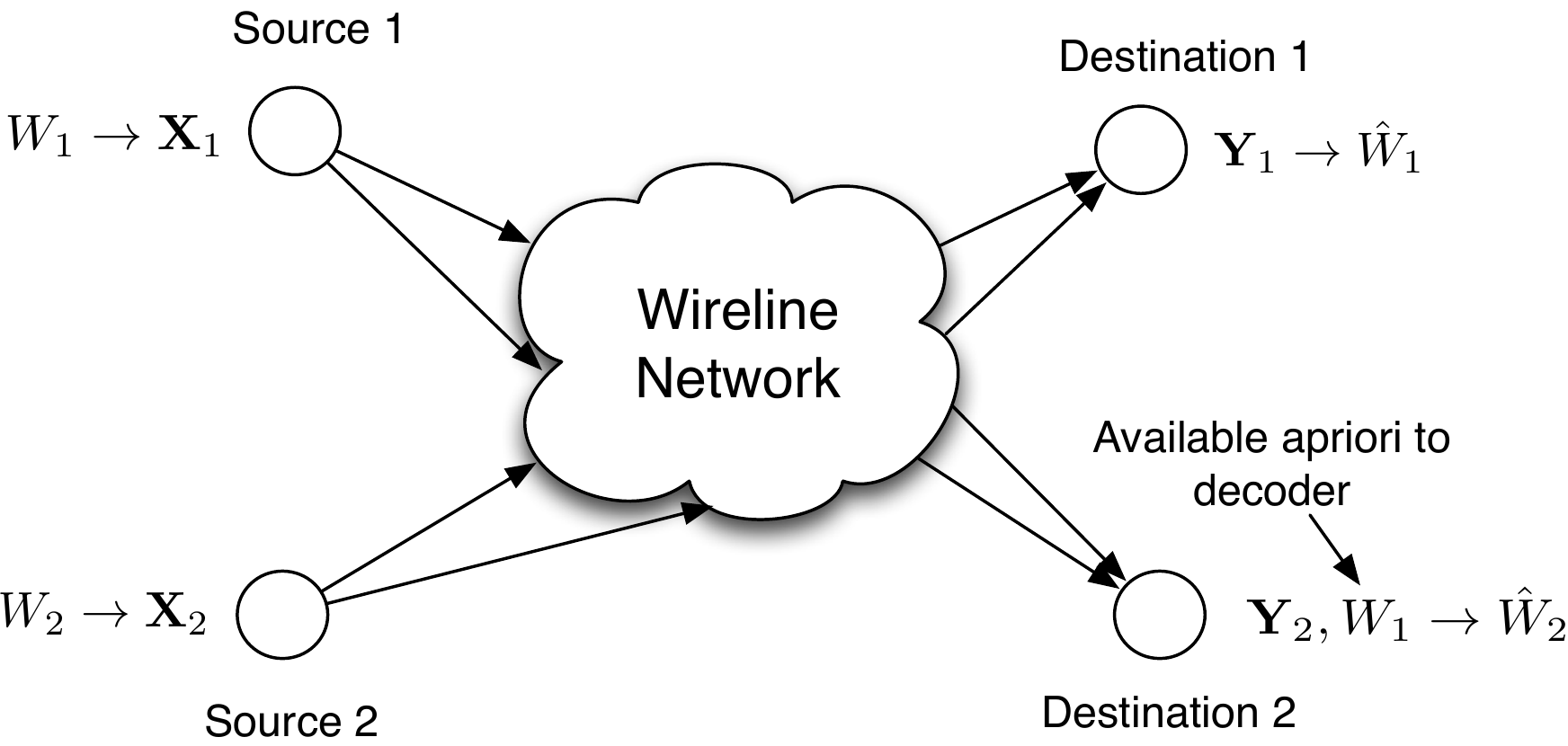}
\end{center}
\caption{The Two-unicast-Z Network}
\label{fig:twounicast-Z}
\end{figure}

%realm of systematic coding algorithms, one outlet is a  a single approach  This gap is noted in \cite{}, which
%identifies interference alignment based solutions that outperform the linear programming based approaches to the index
%coding problem. %identified where  networks coding schemes for multiple unicast  settings was boosted by \cite{} Since
%the routing solution to information transfer can be interpreted as a linear network coding scheme, perhaps the most
%classical demonstration of the power of linear network coding is the classical max-flow min-cut theorem. The  more
%interesting demonstration of the power of linear codes comes from \cite{}.  In the latter reference, it was shown that
%random linear network coding achieves the min-cut for multicast networks where a single source is demanded by more than
%one destination. In contrast to the power of random linear network codes, reference \cite{} also showed that the search
%for good linear codes for networks where there is more than one source in general is connected to wireline networks
%require linear network codes whose coding coefficients depend on the structure of the network.

\subsection{Contributions}
The goal of this paper is to devise systematic algorithms for linear network coding that incorporate ideas from
interference alignment. In this paper, 
{
we focus 
%our attention 
on \emph{two-unicast-Z} networks over directed acyclic
graphs.}
The two-unicast-Z network communication problem consists of two sources $s_1, s_2$, two destinations $t_1, t_2$
and two independent messages $W_1, W_2$. Message $W_i$ is generated by source $s_i$ and is intended to be decoded by
destination $t_i$. In the two-unicast-Z setting, destination $t_2$ has apriori side information of the message $W_1$.
Our nomenclature is inspired from the $Z$-interference channel \cite{Costa} in wireless communications, where, like our
network, only one destination faces interference\footnote{It is perhaps tempting to consider a network where $s_1$ is
not connected to $t_2$ and use the nomenclature of the two-unicast-Z for such a network. However, it is worth noting
that routing is optimal for such a network and each message achieves a rate equal to the min-cut between the respective
source and destination. Therefore, the study of such a network is not particularly interesting from a network capacity
viewpoint.}. In fact, with linear coding, relation between the sources and destinations in the two-unicast-Z setting is
the same as a $Z$-interference channel. We note that the two-unicast-Z network can be interpreted as a two unicast
network with an infinitely capacitated link between source $1$ and destination $2$. Therefore two-unicast-Z networks
form a subclass of two-unicast networks.

Despite the simplicity of its formulation, little is known about the two-unicast-Z network. For instance, while the generalized
network sharing (GNS) bound - a network capacity outerbound formulated in \cite{Kamath_Tse_Ananthram} - is loose in
general for two-unicast networks, we are not aware of any two-unicast-Z network where the GNS 
bound is loose. Similarly, while linear network coding is insufficient for two-unicast, it is not known whether linear
network coding suffices for two-unicast-$Z$ networks. In this paper, we use two-unicast-Z networks over directed acyclic
graphs as a framework to explore our ideas of linear network coding. We make three main contributions in this paper:
%end-to-end   

%Our nomenclature  other words,  two-unicast problem where one of the two messages is available at the destination that
%does  at two source nodes each of which is intended at a  two destinations 
%\subsection{Summary of Contributions}
\begin{enumerate}[(1)]
    \item In Section \ref{sec:algorithm}, we describe an algorithm that obtains linear network codes for two-unicast-Z
        networks over directed acyclic graphs. Our approach is based on designing coding co-efficients to maximize the
        capacity of the implied end-to-end Z-interference channels in the network \cite{Zeng_Cadambe_Medard}.  
        {
        For graphs whose degree is
        bounded by some parameter, 
        %the complexity of our algorithm grows linearly in the number of vertices of the graph and polynomially in the
        %number of edges. 
        the complexity of our algorithm is linear in the number of vertices and polynomial in the number of edges. 
        } We provide
        a high level intuitive description of our algorithm in Section \ref{sec:intuition}.
    \item In Section \ref{sec:mfmc}, we provide an alternate proof of the classical max-flow min-cut theorem for
        directed acyclic graphs. Our proof therefore adds to the previous literature that has uncovered several proofs
        of the theorem \cite{ west2001introduction,Shannon_Elias_Feinstein, Ahlswede_Cai_etal}. Our proof is a direct,
        linear coding based proof, and relies on tools from elementary linear algebra. This is in contrast to previous
        proofs which rely on graph theoretic results (Menger's theorem) or linear programming. This proof, in fact,
        inspires our algorithm for two-unicast-Z networks. We provide an intuitive description of the proof in Section
        \ref{sec:intuition}.
    \item In Section \ref{sec:one-one}, we prove that our algorithm achieves a rate of $(1,1)$ whenever it is feasible.
        Note that the necessary and sufficient conditions for the feasibility of the rate pair $(1,1)$ has been
        characterized in \cite{CCWang:TwoUnicast}. Our proof in Section \ref{sec:one-one} provides an
        alternate proof of the feasibility of the rate pair $(1,1),$ when restricted to two-unicast-$Z$ networks over directed acyclic graphs. Like
        our alternate proof to the max-flow min-cut theorem, our proof of Section \ref{sec:one-one} also relies significantly on elementary linear algebra.
\end{enumerate} 
%new linear network coding algorithm based  For the an intuitive explanation of our algorithm algorithm achieves a rate
%tuple that is at least as large as $(1,1)$. In addition, we show that 

\section{Intuition Behind the algorithm}
\label{sec:intuition}
Consider the two-unicast-$Z$ network described in Fig. \ref{fig:twounicast-Z}. In this network, with linear coding at
all the nodes in the network, the input output relationships can be represented as 
$$\mathbf{Y}_{1} = \mathbf{X}_1 \mathbf{H}_{1}  + \mathbf{X}_2 \mathbf{H}_{2} , 
~~\mathbf{Y}_{2} = \mathbf{X}_2 \mathbf{G}_{2}~,$$ 
{
where for $i \in \{1,2\},$  $\mathbf{X}_{i}$ is a row vector representing the input symbols on the outgoing edges of the $i$th source,
$\mathbf{Y}_{i}$ is row vector representing the symbols on the incoming edges of the $i$th destination node followed by
interference cancellation with the side information if $i=2$. 
}
Note that if $\mathbf{H}_{1},\mathbf{H}_{2}$ are fixed and known, the input-output relations
are essentially akin to the $Z$-interference channel. Using the ideas of El Gamal and Costa \cite{Gamal_Costa} for the
interference channel, the set of achievable rate pairs $(R_1,R_2)$ can be described (see \cite{Zeng_Cadambe_Medard, ramamoorthy_doubleunicast, Ho_hankobayashi}) as the rate tuples $(R_1, R_2)$ satisfying 
\begin{align} 
    R_{1} &\leq \rank\left(\mathbf{H}_{1}\right),  \quad  R_{2} \leq \rank\left(\mathbf{G}_{2}\right) \label{eq:individual}
    \\ R_{1}+R_{2} &\leq \rank\left(\begin{bmatrix}\mathbf{H}_{1}\\
        \mathbf{H}_{2}\end{bmatrix}\right) +
    \textrm{rank}\left(\left[\mathbf{H}_{2}~~\mathbf{G}_{2}\right]\right)-\mathbf{rank}(\mathbf{H}_{2})\label{eq:sumrate}
\end{align}

The goal of our algorithm is to specify the linear coding co-efficients at all the nodes in the network, which in turn
specifies the matrices $\mathbf{H}_{1},\mathbf{H}_{2},\mathbf{G}_{2}$. Once these matrices are specified, 
{
the rate region
achieved in the network is specified by (\ref{eq:individual}), (\ref{eq:sumrate}).
}
Here, we describe our approach to designing the linear coding co-efficients in the two-unicast-$Z$ network.

%$In fact, unlike (\ref{eq:1}), interference alignment is not needed at the source and destination nodes if the
%co-efficients $\mathbf{H}_{1},\mathbf{H}_{2},\mathbf{G}_{2}$ are fixed apriori. 

To describe our intuition, we begin with the familiar single source setting, and describe the ideas of our algorithm
restricted to this setting. Note that with a single source and single destination, with linear coding in the network,
the end-to-end relationship can be represented as $\mathbf{Y}=\mathbf{X}\mathbf{H}$, where $\mathbf{X}$ and $\mathbf{Y}$
respectively represent the symbols carried by the source and destination edges, and $\mathbf{H}$ represents the transfer
matrix from the source to the destination edges. We know from classical results that the linear coding co-efficients can
be chosen such that the rank of $\mathbf{H}$ is equal to the min-cut of the network. Here, we provide an alternate
perspective of this classical result. Our examination of the single-source setting provides a template for our algorithm
for the two-unicast-$Z$ network which is formally described in Section \ref{sec:algorithm}. Our approach also yields
an alternate proof for the max-flow min-cut theorem which is provided in Section \ref{sec:mfmc}.
\subsection{Algorithm for Single-Unicast Network} 
\label{sec:singleunicastalgo}
We focus on a scenario shown in Fig. \ref{fig:singleunicast}. Denote the network communication graph
$\mathcal{G}=(\mathcal{V}, \mathcal{E})$, where $\mathcal{V}$ denotes the set of vertices and $\mathcal{E}$ denotes the
set of edges. Now, suppose that, as shown in Figure \ref{fig:singleunicast}, a linear coding solution has been
formulated for $\tilde{\mathcal{G}} = (\mathcal{V}, \mathcal{E}-\{e\})$, where $e$ denotes an edge coming into the
destination node. The question of interest here is the following: How do we encode the edge $e$ so that the end-to-end
rate is maximized? We assume that our coding strategy is restricted to linear schemes. 

Let $\mathbf{X}$ be a $1 \times S$ vector denoting the source symbols input on the $S$ edges emanating from the source node.
Let $\mathbf{H}$ denote a $S \times (D-1)$ linear transform between the input and $D-1$ destination edges - all the
destination edges excluding edge $e$. Let $\mathbf{p}_1, \mathbf{p}_2, \ldots, \mathbf{p}_k$ be $1 \times S$ vectors
respectively denoting the linear transform between the source and the $k$ edges coming into edge $e$, that is, the $i$th
incoming edge.  Now, our goal is to design the coding strategy for edge $e$, that is, to choose scalars $\alpha_1,
\alpha_2, \ldots, \alpha_k$ such that the rate of the system 
$$\mathbf{Y} = \mathbf{X} \left[\mathbf{H}~~~ \sum_{i=1}^{k}
\alpha_{j}\mathbf{p}_{j}\right]$$ 
is maximized, given matrix $\mathbf{H}$ and vectors
$\mathbf{p}_{1},\mathbf{p}_{2},\ldots, \mathbf{p}_{k}$. Equivalently, the goal is to choose scalars $\alpha_1, \alpha_2,
\ldots, \alpha_k$ such that the rank of $\left[\mathbf{H}~~~ \sum_{i=1}^{k} \alpha_{j}\mathbf{p}_{j}\right]$ is
maximized.  The solution to this problem is quite straightforward - one can notice that if 
\begin{equation}
    \rank\left(\left[\mathbf{H} ~ ~ \mathbf{p}_{1} ~ ~ \mathbf{p}_{2} ~ ~ \ldots ~ ~ \mathbf{p}_{k}\right]\right) 
    > \rank\left(\mathbf{H}\right)
		\label{eq:condition_mflowmincut}
\end{equation}
then the scalars $\alpha_1, \ldots, \alpha_{k}$ can be chosen such that the rank of $\left[\mathbf{H} ~~\sum_{i=1}^{k}
\alpha_{j}\mathbf{p}_{j}\right]$ is equal to the $\rank\left(\mathbf{H}\right)+1$.  Since $\rank\left(\mathbf{H}\right)$ is
the rate obtained by the destination if edge $e$ is ignored, the implication is that if (\ref{eq:condition_mflowmincut})
is satisfied, then, we can design a linear coding strategy such that edge $e$ provides one additional dimension to the
destination.  In fact, if (\ref{eq:condition_mflowmincut}) is satisfied and the field of operation is sufficiently
large, then choosing the scalars $\alpha_1, \ldots, \alpha_k$ randomly, uniformly over the field of operation and
independent of each other increases the rank of $\left[\mathbf{H} ~~\sum_{i=1}^{k} \alpha_{j}\mathbf{p}_{j}\right]$ by
$1,$  implying the existence a linear coding solution.

\begin{figure}[ht]
    \centering
    \includegraphics[scale=0.7]{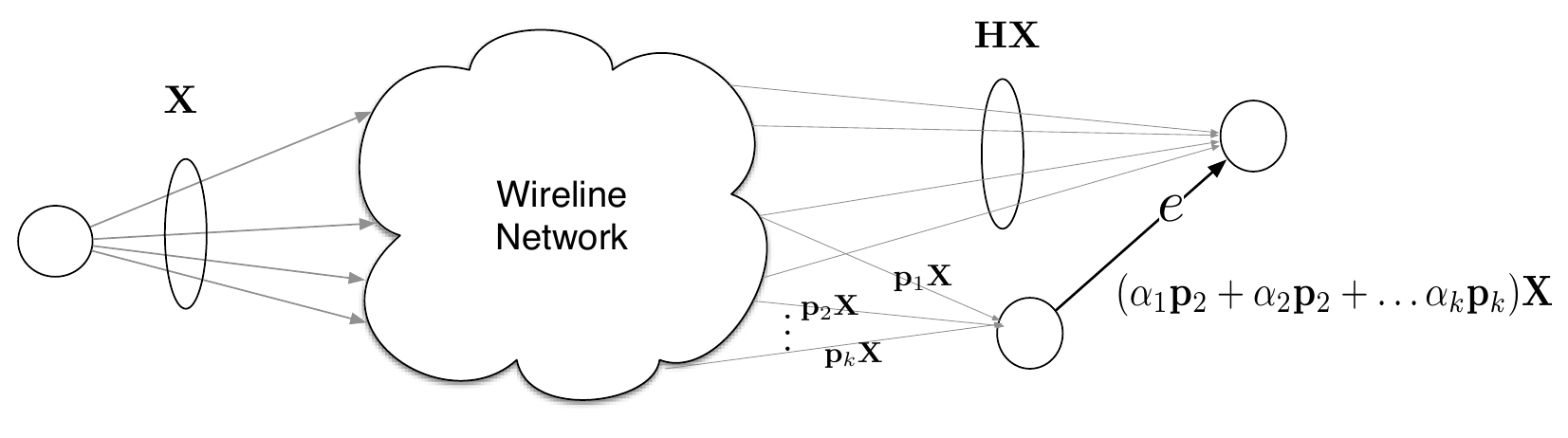}
    \caption{\small{A single-unicast scenario depicted pictorially. The goal is to find scalars $\alpha_1, \ldots, \alpha_k$.}}
    \label{fig:singleunicast}
\end{figure}

A solution to the scenario of Fig. \ref{fig:singleunicast} naturally suggests a linear coding algorithm for the single unicast problem. Suppose we are given a \emph{directed} acyclic graph $\mathcal{G}=(\mathcal{V},\mathcal{E})$, a set of source edges ${S} \subset \mathcal{E}$, a set of destination edges $D \subset \mathcal{E}.$ Our strategy removes the last topologically ordered edge $e \in {D}$ and finds a linear coding solution for the remaining graph. That is, specifically, we develop a linear coding solution for $\mathcal{G}=(\mathcal{V},\mathcal{E}-\{e\})$, with source edges ${S}$ and destination edges $D-\{e\}\cup \textrm{In}(v),$ where $v$ represents the tail node of edge $e$, and $\textrm{In}(v)$ represents the set of edges incoming on to edge $v$.
Therefore, we have reduced our original problem, which intended to design coding co-efficients for $|\mathcal{E}|$ edges, to one which needs to design coding co-efficients for $|\mathcal{E}|-1$ edges, albeit with a slightly different set of destination edges in mind. We can now recursively iterate the same procedure to this smaller problem, removing the last edge as per topological ordering at each iteration and modifying the destination edge set accordingly until all the edges are removed except the source edges. A trivial coding solution applies to this graph, which forms a starting point for the recursive algorithm we have described. 

While our insight might appear superfluous in the context of the single unicast setting, it does lead to an alternate proof for the max-flow min-cut theorem for directed acyclic graphs. To conclude our discussion, we provide an intuitive description of the proof; the proof is formally provided in Section \ref{sec:mfmc}. In our proof, we make the inductive assumption that the max-flow min-cut theorem is valid for the source $S$ and for any destination set which is a subset of $\mathcal{E}-\{e\}.$ Under this assumption, we show using ideas from classical multicast network coding literature that the optimal linear coding solutions for the two possible destination sets $D-\{e\}$ and $D-\{e\} \cup \In(v)$ can be combined into a single linear coding solution that simultaneously obtains the min-cut for both destination sets. Then, we use this combined solution along with the solution to Fig. \ref{fig:singleunicast} and show that this linear coding solution achieves a rate equal to the min-cut for destination set $D$. More specifically, we show that if the edge ${e}$ belongs to a min-cut for destination $D$, then the inductive assumption implies that, for this combined solution, (\ref{eq:condition_mflowmincut}) holds; our strategy of choosing coding co-efficients randomly over the field ensures that a rank that is equal to the min-cut is achieved for graph $\mathcal{G}$ with destination $D$ as well.

%The algorithm presented above is insightful even in the familiar single-unicast setting because of the following observation: both routing and random coding strategies follow naturally as optimal decisions when we try to maximize the rank of the end-to-end linear transform. 
\subsection{Algorithm for Two-Unicast-Z Networks}
\label{sec:twounicastzalgo}
\begin{figure}[ht]
    \centering
    \includegraphics[scale=0.7]{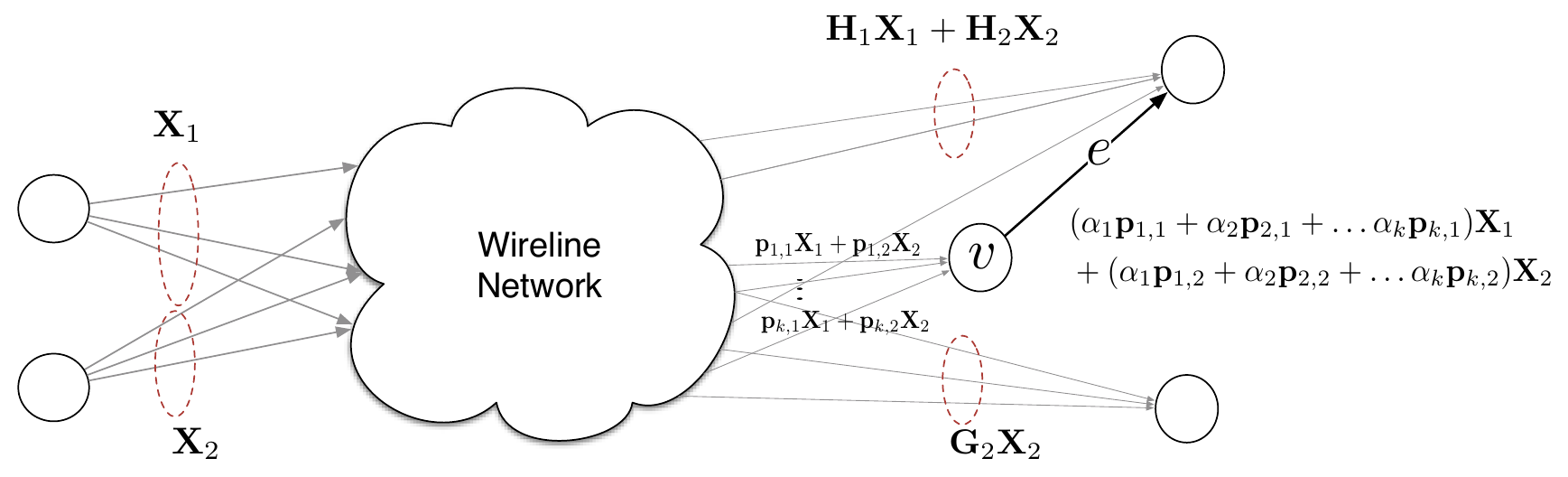}
    \caption{\small{A two-unicast-Z scenario depicted pictorially. The goal is to find scalars $\alpha_1, \ldots, \alpha_k$ to maximize (\ref{eq:2}).}}
    \label{fig:twounicastz-algo}
\end{figure}

Consider a two-unicast-$Z$ network of the form shown in Fig. \ref{fig:twounicastz-algo}. The graph $\mathcal{G}=(\mathcal{V},\mathcal{E})$ consists of two sets of source edges ${S}_{1},{S}_{2}$, two sets of destination edges ${T}_{1},{T}_{2}$, with the destination $2$ being aware of the message of the first source apriori. Now, consider a situation where a coding solution has been formulated for all the edges of the graph, with the exception of edge $e \in \mathcal{D}_{1}$. We are interested in understanding how to encode the edge $e$ so that the end-to-end rate is maximized. %We assume that our coding strategy is restricted to linear schemes. 

Our heuristic is based on maximizing the \emph{sum-rate} that is, the right hand side of equation (\ref{eq:sumrate}).
Based on Fig. \ref{fig:twounicastz-algo}, our goal is to find $\alpha_{1},\alpha_{2},\ldots, \alpha_{k}$ such that
{
\begin{align} 
    \rank\left(
    \begin{bmatrix} 
        \mathbf{H}_{1} & \sum_{i=1}^{k}\alpha_{i}\mathbf{p}_{i,1} \\
        \mathbf{H}_{2} & \sum_{i=1}^{k} \alpha_{i}\mathbf{p}_{i,2} 
    \end{bmatrix}
    \right) +
    \rank\left(
    \begin{bmatrix}
    \mathbf{H}_{2} & \sum_{i=1}^{k} \alpha_{i}\mathbf{p}_{i,2} & \mathbf{G}_{2}
    \end{bmatrix}
    \right)
    -
    \rank \left(
    \begin{bmatrix}
    \mathbf{H}_{2} & \sum_{i=1}^{k} \alpha_{i} \mathbf{p}_{i,2}
    \end{bmatrix}
    \right) \label{eq:2} 
\end{align} 
}
is maximized. To do so, we examine two cases:
\begin{description}
    \item[Case 1 ]\hspace{5pt}  If $\mathbf{p}_{1,2}, \ldots ,\mathbf{p}_{k,2}$ lie in the span of $\mathbf{H}_{2}$,
        then, clearly, choosing $\alpha_{i}$s randomly and uniformly over the field is the best
        strategy with a probability that tends to $1$ as the field size increases. 
        {
        %This is because the third, negative, term in (\ref{eq:2}) is equal to $0$, and random linear coding
        %maximizes the first two terms.
        This is because in the third and negative term in (\ref{eq:2}), the column corresponding to the random linear
        combination of $\mathbf{p}_{i,2}$'s do not contribute to the rank of the matrix, while random linear coding
        maximizes the first two terms with high probability. 
        }
    \item[Case 2 ]\hspace{5pt}  If $\mathbf{p}_{1,2}, \ldots ,\mathbf{p}_{k,2}$ does not lie in the span of
        $\mathbf{H}_{2}$, then the solution is a bit more involved. We divide this case into two sub-cases
		\begin{description}
            \item[Case 2a ]\hspace{5pt} Suppose that $\mathbf{p}_{1,2}, \ldots ,\mathbf{p}_{k,2}$ do not lie in the
                span of $[ \mathbf{H}_{2}~~\mathbf{G}_{2}]$, then chosing $\alpha_{i}$s randomly maximizes the
                expression of (\ref{eq:2}). In particular, we note that choosing $\alpha_{i}$s randomly increases the
                two positive terms and the negative term of (\ref{eq:2}) by $1$, effectively increasing the the
                expression of (\ref{eq:2}) by $1$, as compared with 
                $\rank 
                \left(
                \begin{bmatrix}
                    \mathbf{H}_{1} \\ \mathbf{H}_{2} 
                \end{bmatrix}
                \right) + 
                \rank \left( 
                \begin{bmatrix}
                \mathbf{H}_{2} & \mathbf{G}_{2}
                \end{bmatrix}
                \right) -
                \rank\left(\mathbf{H}_{2}\right)$. We later show in Lemma \ref{lma:grankcol} that the expression of
                (\ref{eq:2}) cannot be increased by more than $1$; this implies the optimality of
                the random coding approach for the case in consideration here.
            \item [Case 2b ]\hspace{5pt} Suppose that $\mathbf{p}_{1,2}, \ldots ,\mathbf{p}_{k,2}$ lies in the span of
                $\begin{bmatrix}
                \mathbf{H}_{2} & \mathbf{G}_{2}
                \end{bmatrix} $, but does not lie in the span of $\mathbf{H}_{2}$. In this case, the
                optimal strategy is to choose co-efficients $\alpha_{1},\alpha_{2},\ldots, \alpha_{k}$ so that
                $\sum_{i=1}^{k} \alpha_{i} \mathbf{p}_{i,2}$ is \emph{a random vector in the intersection of the column
                spaces of }$\mathbf{H}_{2}$ and 
                $\begin{bmatrix}
                \mathbf{p}_{1,2} &\mathbf{p}_{2,2} & \ldots & \mathbf{p}_{k,2}
                \end{bmatrix}~.
                $ In other words, we intend to \emph{align the local coding vector on edge $e$ in the space of
                    $\mathbf{H}_{2}$}.
			\end{description}
                %alpha_{i}$s randomly maximizes the expression of \ref{}. In particular, we note that choosing $\alpha_{i}$s randomly increases the two positive terms and the negative term of \ref{} by $1$, effectively increasing the sum-rate by $1$. It can be verified that the sum-rate cannot be increased by more than $1$ with linear schemes establishing the optimality of the random coding approach.
		\end{description} 
From the above discussion, it is interesting to note that we naturally uncover interference alignment in Case 2b as a technique that maximizes the expression of (\ref{eq:2}).
In fact, we show later in Lemma \ref{lem:alignment} that, if 
\begin{align} 
    \rank \left(
    \begin{bmatrix} 
        \mathbf{H}_{1} & \mathbf{p}_{1,1} & \mathbf{p}_{2,1} & \ldots & \mathbf{p}_{k,1} \\ 
        \mathbf{H}_{2} & \mathbf{p}_{1,2} & \mathbf{p}_{2,2} & \ldots & \mathbf{p}_{k,2} 
    \end{bmatrix} 
    \right)
    +
    \rank\left( 
    \begin{bmatrix}
    \mathbf{H}_{2}& \mathbf{p}_{1,2} & \mathbf{p}_{2,2} & \ldots & \mathbf{p}_{k,2} & \mathbf{G}_{2} 
    \end{bmatrix}
    \right) \nonumber \\ 
    -
    \rank\left(
    \begin{bmatrix}
    \mathbf{H}_{2} &  \mathbf{p}_{1,2} & \mathbf{p}_{1,2} & \ldots & \mathbf{p}_{k,2} 
    \end{bmatrix}
    \right) \nonumber \\ 
    > \rank\left(
    \begin{bmatrix}
     \mathbf{H}_{1} \\ \mathbf{H}_{2} 
    \end{bmatrix}
    \right) + 
    \rank\left(
    \begin{bmatrix}
    \mathbf{H}_{2} & \mathbf{G}_{2}
    \end{bmatrix}
    \right) 
    - \rank\left(\mathbf{H}_{2} \right) 
    \label{eq:condfordest1}
\end{align}
then, our choice of $\alpha_1,\alpha_2,\ldots,\alpha_k$ increases (\ref{eq:2}) by $1$ as compared with
$\rank\left(
\begin{bmatrix}
\mathbf{H}_{1} \\ \mathbf{H}_{2} 
\end{bmatrix}
\right) + 
\rank\left(
\begin{bmatrix}
\mathbf{H}_{2} & \mathbf{G}_{2}
\end{bmatrix}
\right) - \rank\left(\mathbf{H}_{2} \right)$.
 
\begin{figure}[ht]
    \centering
    \includegraphics[scale=0.7]{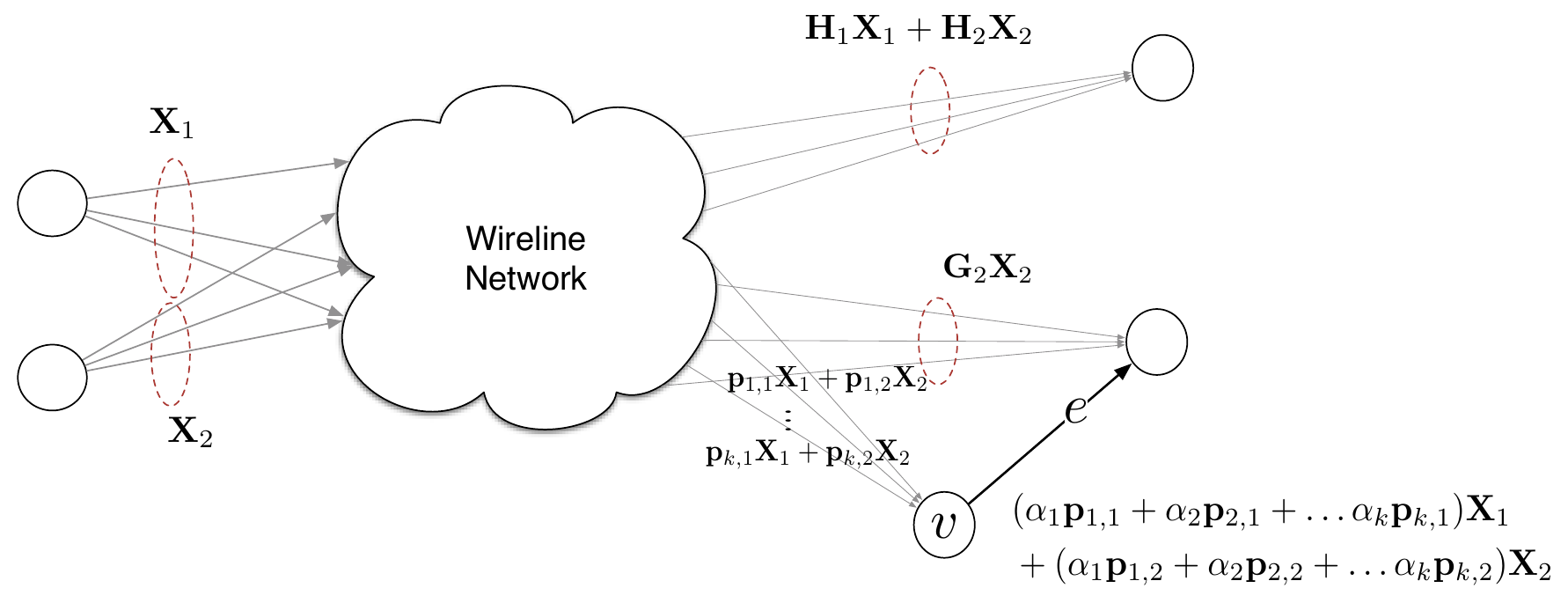}
    \caption{\small{A two-unicast-Z scenario depicted pictorially. The goal is to find scalars $\alpha_1, \ldots, \alpha_k$ to maximize (\ref{eq:2}).}}
    \label{fig:twounicastz-algo1}
\end{figure}

Now consider the scenario of Fig. \ref{fig:twounicastz-algo1}. In this scenario, the goal of maximizing the right hand
side of \eqref{eq:sumrate} is tantamount to choosing scalars $\alpha_1,\alpha_2,\ldots,\alpha_k$ to maximize
$\rank \left( 
\begin{bmatrix}
\mathbf{H}_{2} &\mathbf{G}_{2} & \sum_{i=1}^{k} \alpha_{i}\mathbf{p}_{i,2}
\end{bmatrix}
\right).$ 
The scenario is similar to the single-unicast problem discussed earlier, and choosing the scalars $\alpha_i$ randomly,
uniformly over the field of operation and independent of each other maximizes the sum-rate. In fact, it is easy to see
that this strategy improves the sum-rate by $1$, if 
\begin{align} 
    &\rank\left(
    \begin{bmatrix} \mathbf{H}_{1} \\ \mathbf{H}_{2} \end{bmatrix}
    \right) +
    \rank\left( 
    \begin{bmatrix}
    \mathbf{H}_{2} & \mathbf{G}_{2} & \mathbf{p}_{1,2} & \mathbf{p}_{2,2} & \ldots & \mathbf{p}_{k,2}
    \end{bmatrix}
    \right) - 
    \rank\left(\mathbf{H}_{2} \right) \\
&> \rank
\left(
\begin{bmatrix} 
    \mathbf{H}_{1} \\ \mathbf{H}_{2} 
\end{bmatrix}
\right) + 
\rank\left(
\begin{bmatrix}
\mathbf{H}_{2}~~\mathbf{G}_{2}
\end{bmatrix}
\right) - \rank\left(\mathbf{H}_{2}\right).
\label{eq:condfordest2}
\end{align}

\begin{remark}
    Strictly speaking, the sum-rate of the two-unicast-Z network is equal to the minimum of
    $$\rank\left(\left[\mathbf{H}_{1}~~\sum_{\ell=1}^{k} \alpha_\ell
    \mathbf{p}_{\ell,1}\right]\right)+\rank(\mathbf{G}_{2})$$ and the expression of (\ref{eq:2}). In this paper, for the
    sake of simplicity, we restrict our attention to maximizing the expression of (\ref{eq:2}); our ideas and algorithms
    can be easily modified to maximize the smaller of $$\rank\left(\left[\mathbf{H}_{1}~~\sum_{\ell=1}^{k} \alpha_\ell
    \mathbf{p}_{\ell,1}\right]\right)+\rank(\mathbf{G}_{2})$$ and (\ref{eq:2}).
\end{remark}

To summarize, we observe from (\ref{eq:condfordest1}) and (\ref{eq:condfordest2})  that an edge $e$ can increase the
right hand side of (\ref{eq:sumrate}) by $1$ if, the parent edges of the edge $e$ in combination to the other already
existing edges in the destination can increase the right hand side of (\ref{eq:sumrate}) by at least $1$. Furthermore,
our strategy to maximize the right hand side of (\ref{eq:2}) automatically uncovers the idea of alignment. Before
proceeding, we briefly explain how the problems of Fig. \ref{fig:twounicastz-algo} and \ref{fig:twounicastz-algo1} can
be composed naturally into a recursive algorithm to design linear coding co-efficients for the entire network.

We represent the network as a directed acyclic graph $\mathcal{G}=(\mathcal{V},\mathcal{E})$, where $\mathcal{V}$
denotes the set of vertices, and $\mathcal{E}$ denotes the set of edges. We begin with the last edge of the graph, that
is the edge $e$ with the highest topological order in the graph. This edge $e$ is incident on either the first
destination or the second destination. Given a coding solution for the graph ${\mathcal{G}_1} =
(\mathcal{V},\mathcal{E}-\{e\}),$ we can design a linear coding solution based on the above approach. Therefore, we aim
to design a coding solution for the smaller graph ${\mathcal{G}_1}.$ To do this, we add all the parent edges of $e$ to
the corresponding destination. That is, if edge $e$ is in destination $1$, modify destination 1 in $\mathcal{G}_{1}$ to
include all the edges incoming on to the vertex $v,$ where $v$ is the vertex from which edge $e$ emanates (See Figs.
\ref{fig:twounicastz-algo}). Similarly, if edge $e$ is incident onto destination $2$, we remove edge $e$ to reduce the
problem to a smaller graph $\mathcal{G}_1$ and all the parent edges of $e$ to destination $2$ (See Fig.
\ref{fig:twounicastz-algo1}). Now, our goal is to find a linear coding solution to the smaller problem $\mathcal{G}_1$.
We proceed similarly by identifying the last topologically ordered edge in $\mathcal{G}_1$ and removing it to obtain
$\mathcal{G}_{2}$, and further modifying the destinations. If we proceed similarly, removing one edge at a time from
graph $\mathcal{G}$, we obtain a sequence of graphs $\mathcal{G}_1,\mathcal{G}_{2},\ldots,$  to eventually obtain a
graph $\mathcal{G}_{N}$ where the destination edges coincide with the source edges. Starting with a trivial coding
solution for $\mathcal{G}_{N}$, we build a coding solution for the sequence of graphs
$\mathcal{G}_{N},\mathcal{G}_{N-1},\ldots,\mathcal{G}_{1}$ and eventually obtain a coding solution for graph
$\mathcal{G}$. Our approach is formally outlined in Section \ref{sec:algorithm}. Before we proceed, we note that in our
sequence of graphs obtained above, it can transpire that the last topologically ordered edge in one of the graphs
belongs to both destinations. We omit an explanation of this scenario here, since our approach in handling this scenario
is similar to the one depicted in Fig. \ref{fig:twounicastz-algo}.

\section{System Model}
\label{sec:model}
Consider a
directed acyclic graph (DAG) $\mathcal{G} = (\mathcal{V}, \mathcal{E})$, where $\mathcal{V}$  denotes the set of
vertices and $\mathcal{E}$ denotes the set of edges. We allow multiple edges between vertices, hence, $\mathcal{E}
\subset \mathcal{V} \times \mathcal{V} \times \mathbb{Z}_{+}$, where $\mathbb{Z}_{+}$ denotes the set of positive
integers.  For an edge $e = (u, v, i) \in \mathcal{E}$, we denote $\text{Head}(e)=v$ and $\text{Tail}(e) = u$; in other words, when the direction of the edge is denoted by an arrow, the vertex at the arrow head is the head vertex, and the vertex at the tail of the arrow is the tail vertex of the edge. When
there is only one edge between node $u$ and node $v$, we simply denote the edge as $(u,v)$. For a given vertex
$v \in \mathcal{V}$, we denote $\In(v) = \left\{ e\in \mathcal{E} : \Head(e) = v\right\} $ and $\Out(v) =
\left\{e\in\mathcal{E}: \Tail(e) = v \right\}$.  

% describe the problem setups and notations
In this paper, we focus on the networks with one or more unicast sessions. We define each source as a node in
$\mathcal{V}$, while each destination as \emph{a subset of edges} in $\mathcal{E}$. Subsequently, a single unicast
network problem $\Omega$ can be specified by a $3$-tuple $(\mathcal{G}, s, T)$, where $\mathcal{G}=(\mathcal{V},
\mathcal{E})$ is the underlying graph, $s \in \mathcal{V}$ is the source node and $T \subset \mathcal{E}$ is set of
destination edges. Every node in the graph represents an encoding node, and every edge in the graph represents an orthogonal, delay-free, link of unit capacity. 

In two-unicast-Z networks, we use the set $\mathcal{S} = \left\{ s_1, s_2 \right\}$, where $s_1, s_2
\in \mathcal{V}$, to denote set of two sources.  We use $\mathcal{T}=\left\{ T_1, T_2 \right\}$ to denote the set of two
destinations, where $T_1$ and $T_2$ each is a set of edges, i.e., $T_i \subset \mathcal{E}, i = 1, 2$.  To keep the scenario general, 
we allow an edge to belong both destinations, i.e.,   $T_1 \cap T_2$ need not be the null set. Furthermore, the head
vertices for edges in the same destination may not have to be the same, i.e., for $e_1, e_2
\in T_i, i = 1, 2$    it is possible that $\Head(e_1) \neq \Head(e_2)$. Without loss of generality, we consider graphs where the edges with the highest topological order belongs to $T_1 \cup T_2$.

%\begin{figure}[ht]
%    \centering
%    \includegraphics[scale=0.45]{./figs/NetSketch.eps}
%    \caption{A Two-Unicast-Z network}
%    \label{fig:networksketch}
%\end{figure}
In the two-unicast-Z network, the sources $s_1$ and $s_2$ generate independent messages $W_1$ and $W_2$ respectively. The message $W_1$ is available a priori to destination $T_2$. The goal of the two-unicast-$Z$ network is to design encoding functions at every node in the network and decoding functions such that $W_1$ is recoverable from the symbols carried by the edges in $T_1,$ and $W_2$ is recoverable from the symbols carried by the edges in $T_2$ and the side information $W_1$. Without loss of generality, assume that $s_i$ communicates with at least one edge in $T_i$ for $i \in \{1,2\}$. Similar to the case of single unicast, we can denote a two-unicast-Z network coding problem $\Omega$ using a $3$-tuple, i.e.  $\Omega = (\mathcal{G},
\mathcal{S}, \mathcal{T})$, where $\mathcal{S} = \left\{ s_1, s_2 \right\}, \mathcal{T} = \left\{ T_1, T_2 \right\}$.
%The
%sources transmit the coded symbols $\mathbf{X}_1$ and $\mathbf{X}_2$ respectively into the network. The destinations
%receive collections of symbols $\mathbf{Y}_1$ and $\mathbf{Y}_2$ on edges $T_1$ and $T_2$ respectively. The goals of the two-unicast-Z network is to decode $W_1$ using the symbols $Y_1$, and to decode $W_2$ using the symbols $Y_1$ and the side information $W_1$. 

 % Define rate pair, code, solution, linear code
A rate pair $(R_1, R_2)$ is \textit{achievable} if for every $\epsilon > 0, \delta > 0$, there
exists a coding scheme which encodes message $W_i$ at a rate $R_i - \delta_i$, for some $0 \leq \delta_i \leq \delta$,
such that the average decoding error probability is smaller than $\epsilon$. The capacity region is the closure of
the set of all achievable rate pairs.

\subsubsection*{Topological Order}
Since the graphs considered in this paper are directed acyclic graphs, there exists a standard topological order $\Ord_{\mathcal{V}}$ on the set of vertices $\mathcal{V}$ of the graph. The order $\Ord_{\mathcal{V}}$ satisfies the following property: if there is an edge $(u, v, i) \in \mathcal{E}$, then, 
$\Ord_{\mathcal{V}}(u) < \Ord_{\mathcal{V}}(v)$. We define a partial order $\Ord_{\mathcal{E}}$ on the set of edge $\mathcal{E}$ such that the order of an edge is equal to
the order of the tail node of the edge,  i.e., $\Ord_{\mathcal{E}}(e) = \Ord_{\mathcal{E}} (\Tail(e)), e \in \mathcal{E}$. Note that
all edges sharing the same tail node have the same order. When there is no ambiguity, we omit the subscript in the ordering and simply denote the ordering on the edges (or vertices) as $\Ord$.

% specify submatrix notations.
\subsection*{Linear Network Coding} In this paper, we consider scalar linear coding, where the encoded symbol along each edge is an element of a finite field $\mathbb{F}$. We use the algebraic framework of linear network coding of \cite{Koetter_Medard} to relate the linear coding co-efficients at the vertices of the graph to the encoded symbols carried by the edges. Specifically, we describe a linear coding solution for a network using a local coding matrix $\mathbf{F}$, whose $i$-th column corresponds to the edge with topological order $i$ and stores the local coding vector on the edge from the symbols carried by its parent edges. The
linear transfer matrix of the entire network therefore given by $\mathbf{M} = \left( \mathbf{I} - \mathbf{F}
\right)^{-1}$ and the transfer matrices between sources and destinations can be obtained as submatrices of the network
transfer matrix $\mathbf{M}$. %Therefore, a linear coding algorithm presents a systematic way of designing the local
%coding matrix $\mathbf{F}$ for a network. 

For $i \in \{1,2\},$ we denote the symbols carried by the edges emanating from source $i$ by the $1 \times |\Out(s_i)|$ vector $\mathbf{X}_{i}$. Similarly, we denote the symbols carried by the edges in $T_i$ to be the $1 \times |T_i|$ vector $\mathbf{Y}_{i}.$  For a linear coding scheme, we write 
\begin{align}
    \mathbf{Y}_1 &= \mathbf{X}_1 \mathbf{H}_1  + \mathbf{X}_2 \mathbf{H}_2 \\
    \mathbf{Y}_2 &= \mathbf{X}_1 \mathbf{G}_1 + \mathbf{X}_2 \mathbf{G}_2 ~,
\end{align}
where $\mathbf{H}_i$ is the $|\Out(s_i)|\times |T_1|$ transfer matrix from $\mathbf{X}_i$ to $\mathbf{Y}_1$ and $\mathbf{G}_i$ is the $|\Out(s_i)|\times |T_2|$  transfer
matrix from $\mathbf{X}_i$ to $\mathbf{Y}_2$. Note that the matrices $\mH{i}$ and $\mG{i}$ are sub-matrices of the network
transfer matrix $\mathbf{M}$ by selecting the rows and columns corresponding to the specified
source and destination edges respectively. 
%We denote, 
%\begin{align}
    %\mathbf{H} &= 
    %\begin{bmatrix}
        %\mathbf{H}_1 \\
        %\mathbf{H}_2  
    %\end{bmatrix}, &
     %\mathbf{G} &= 
   %\begin{bmatrix}
        %\mathbf{G}_1 \\
        %\mathbf{G}_2  
    %\end{bmatrix}~, 
    %\label{eqn:transfermatrix}
%\end{align}
%where we use the vertical bar in $\mathbf{M}$ to separate the transfer matrices for $T_1$ and $T_2$ edges. Hence,
%$\mathbf{Y}_1 = \mathbf{X} \mathbf{H}$ and $\mathbf{Y}_2 = \mathbf{X} \mathbf{G}$. 
%For convenience, we abuse the notation to use the edges to indicate the column of the transfer matrices.  For example,
%the columns of $\mathbf{H}_i$ correspond to the edges in the destination set $T_1$, while the columns of matrix
%$\mathbf{G}_i$ correspond to the edges in set $T_2$. For a subset $U \subset T_1$, the matrix $\mathbf{H}_j^U, j = 1,2 $
%denotes the submatrices of $\mathbf{H}_j$ formed by the columns corresponding to destination edges in $U$. The matrices
%$\mathbf{G}_j^U$ are defined similarly when $U$ is a subset of $T_2$. 

For the two-unicast-$Z$ network, the rate pair
$(R_1, R_2)$ achieved by a linear coding scheme is characterized by 
\begin{align} 
    R_{1} &\leq \rank\left(\mathbf{H}_{1}\right),  \quad  R_{2} \leq \rank\left(\mathbf{G}_{2}\right)~, \label{eq1:individual}
    \\ R_{1}+R_{2} &\leq \rank\left(\begin{bmatrix}\mathbf{H}_{1}\\
        \mathbf{H}_{2}\end{bmatrix}\right) +
    \rank\left(
    \begin{bmatrix}
    \mathbf{H}_{2} & \mathbf{G}_{2}
    \end{bmatrix}
    \right)
    -\rank\left(\mathbf{H}_{2}\right)~.\label{eq1:sumrate}
\end{align}
Note that the rate region does not depend on the matrix $\mathbf{G}_{1}$ since destination $2$ cancels the effect of $\mathbf{X}_{1}\mathbf{G}_{1}$ using its side information. The rate region can be derived as a simple corollary of the result of \cite{Gamal_Costa}, which obtains the capacity of a class of deterministic $2$-user interference channels. We refer the reader to \cite{Zeng_Cadambe_Medard} for a proof. 

\subsubsection*{Notations} 
% \emph{Notations}:
The cardinality of a set $E$ is denoted by $|E|$. For sets $A$ and $B$, $A \backslash B$ denotes the set of elements in
$A$ but not in $B$. 
%A index set of $n$ elements is denoted as $[n] = \left\{ 1,2,\dots,n \right\}$. 
{
For a matrix $\mathbf{A}$, $\Span(\mathbf{A})$ denotes its column span and $\text{Ker}(\mathbf{A})$ denotes the
nullspace of $\Span(\mathbf{A})$. 
}
In a graph $\mathcal{G} = (\mathcal{V}, \mathcal{E})$, For $u, v \in \mathcal{V}, u
\leadsto v$ indicates that $u$ communicates with $v$, i.e. there is a path from $u$ to $v$ on $\mathcal{G}$. $u \not
\leadsto v$ means that $u$ does not communicates with $v$ on $\mathcal{G}$. For $i,j \in \mathcal{E}$, $i \leadsto j$ is
equivalent to $\Out(i) \leadsto \In(j)$. We denote sub-matrices of the local and global coding matrices by super-scripts. Specifically, for two edge sets $\mathcal{E}_{1},\mathcal{E}_{2} \subset \mathcal{E},$ the matrices $\mathbf{F}^{\mathcal{E}_{1},\mathcal{E}_{2}}$ and $\mathbf{M}^{\mathcal{E}_{1},\mathcal{E}_{2}}$ respectively represent $|\mathcal{E}_{1}|\times |\mathcal{E}_{2}|$ dimensional sub-matrices of $\mathbf{F}$ and $\mathbf{M}$ derived from the rows corresponding to $\mathcal{E}_{1}$ and columns corresponding to $\mathcal{E}_{2}$. For example, we can write $\mathbf{H}_{1} = \mathbf{M}^{\Out(s_1),T_1}.$

% Define cuts
%For any sets $\mathcal{C} \subseteq \mathcal{E}$, $A \subseteq \mathcal{V}, B \subseteq \mathcal{E}$, we say
%$\mathcal{C}$ is an $A-B$ cut set if there exists no directed path from any vertex in $A$ to any vertex in $\Head(B) =
%\left\{ \Head(e), e \in B \right\}$ on the graph $(\mathcal{V}, \mathcal{E}\backslash \mathcal{C})$.  We define function
%$c^{\mathcal{G}}(A, B) \define \min_{\mathcal{C}} \{|\mathcal{C}| : \mathcal{C} \text{ is a } A-B \text{ cut set}\} .$
%When the graph in consideration can be gleaned unambiguously from the context, we drop the superscript $\mathcal{G}$.
%For convenience, we mildly abuse notation when $A,$ and/or $B$ is a singleton; we write $c(\{v\},\{w\})$ as simply
%$c(v;w)$.

In the context of the two-unicast-Z networks, a \emph{GNS-cut set} as defined in
\cite{Kamath_Tse_Ananthram} is a set $\mathcal{Q} \in \mathcal{E}$, such that $\mathcal{Q}$ is 
\begin{enumerate}
  \item a $s_1 - T_1$ cut-set, and,
  \item a $s_2 - T_2$ cut-set, and,
  \item a $s_2 - T_1$ cut-set.
\end{enumerate}
The \emph{GNS-cut set bound} of a two-unicast-Z network is defined to be the cardinality of the smallest GNS set in the network. It is shown in \cite{Kamath_Tse_Ananthram} that the GNS-cut set bound is an information theoretic upper bound on the sum-rate achievable in the network. 

\section{Recursive, alignment-based, linear network coding algorithm}
\label{sec:algorithm}
In this section, we present a scalar linear network code construction for the two-unicast-Z network problem. The
algorithm consists of two sub-routines, the \emph{destination reduction} algorithm which is described in Section
\ref{subsec:destreduction}, and the \emph{recursive code construction} which is described in Section
\ref{subsec:recursivecoding}. When both of them are run, the recursive coding routine returns the coding matrix
$\mathbf{F}$. The rate achieved can be obtained via (\ref{eq1:individual}),(\ref{eq1:sumrate}). The sub-routines are
pictorially depicted in Fig. \ref{fig:recursive_coding} for the network in Fig. \ref{fig:example1}. Before we describe
these sub-routines, we begin with some preliminary definitions and lemmas that will be useful in the algorithm
description. 
\subsection{Preliminaries}
\begin{definition}[Grank]
    Given matrices $\mathbf{H}_{1},\mathbf{H}_{2},\mathbf{G}_{2}$ of dimensions $P_1 \times Q_1$, $P_2 \times Q_1$ and
    $P_2 \times Q_2$ respectively, where $P_1, P_2, Q_1, Q_2$ are positive integers, the \emph{Grank} is defined as 
   \begin{align*}
        \Grank(\mathbf{H}_{1},\mathbf{H}_{2},\mathbf{G}_{2}) = \rank \left( 
        \begin{bmatrix}
            \mathbf{H}_1 \\ \mathbf{H}_2
        \end{bmatrix}
        \right)
        + \rank \left( 
        \begin{bmatrix}
            \mathbf{H}_2 & \mathbf{G}_2
        \end{bmatrix}
        \right)
        - \rank \left( \mathbf{H}_2 \right).
   \end{align*}
\end{definition}
{Note that the Grank is related to the sum-rate of the two-unicast-Z network where
$\mathbf{H}_{1},\mathbf{H}_{2}$ and $\mathbf{G}_{2}$ respectively represent the transfer matrices between source 1 and
destination $1$, source 2 and destination $1,$ and source $2$ and destination $2$.}
{
{
\begin{remark}
We can show that 
{
$$\textrm{Grank}(\mathbf{H}_{1},\mathbf{H}_{2},\mathbf{G}_{2}) = \min_{\mathbf{G}_1 \in \mathbb{F}^{P_1
\times Q_2}} \textrm{rank} \begin{bmatrix} \mathbf{H}_{1} & \mathbf{G}_{1} \\ \mathbf{H}_{2} &
    \mathbf{G}_{2}\end{bmatrix}. $$
}
Our use of the term ``Grank'' is inspired by the above observation which indicates
the quantity of interest is closely related to the rank of an appropriate matrix.

%{\color{red} not sure what is $S_2 \times T_1$ here}
    \end{remark}
\begin{remark}
    We have shown in \cite{Zeng_Cadambe_Medard}, that, if $\mathbf{H}_{1}, \mathbf{H}_{2}$ and $\mathbf{G}_{2}$ respectively
    represent the transfer matrices between source $1$ and destination $1$, source 2 and destination $1$, and source $2$
    and destination $2$, then $\text{Grank}(\mathbf{H}_1, \mathbf{H}_{2},\mathbf{G}_{2})$ is upper bounded by minimum
    generalized network sharing cut value of the network.
\end{remark}}

}

We state some useful properties of the Grank next.

\begin{lemma}
    Let $\mathbf{H}_{1},\mathbf{H}_{2},\mathbf{G}_{2}, \mathbf{A},\mathbf{B}$ and $\mathbf{C}$ be matrices with entries from a finite field $\mathbb{F},$ respectively having dimensions $P_1
    \times Q_1$, $P_2 \times Q_1,$ $P_2 \times Q_2, P_1 \times M, P_2 \times M$ and $P_2 \times N$, for
    positive integers $P_1, P_2, Q_1, Q_2, M,N$. Then the following properties hold. 

	\begin{enumerate}[(i)]
		\item Concatenation of columns to matrices does not reduce Grank.

    {$$\Grank\left([~\mathbf{H}_{1} ~ ~ \mathbf{A}~], [~\mathbf{H}_{2} ~~ \mathbf{B}~]
    ,[~\mathbf{G}_{2}~~\mathbf{C}~]\right) \geq \Grank(\mathbf{H}_{1},\mathbf{H}_{2},\mathbf{G}_{2}).$$ }
\item Concatenating $M$ column increases the Grank by at most $M$.
		$$ \Grank([\mathbf{H}_{1}~~\mathbf{A}], [\mathbf{H}_{2}~~\mathbf{B}], \mathbf{G}_{2}) \leq \Grank(\mathbf{H}_{1},\mathbf{H}_{2},\mathbf{G}_{2})+ M $$
		$$ \Grank(\mathbf{H}_{1}, \mathbf{H}_{2}, \left[\mathbf{G}_{2} ~~ \mathbf{C}\right]) \leq \Grank(\mathbf{H}_{1},\mathbf{H}_{2},\mathbf{G}_{2})+ N.$$
\item Concatenation of linearly dependent columns does not change the Grank. 
Suppose that 
$$\textrm{colspan}\begin{bmatrix}\mathbf{A} \\ \mathbf{B}\end{bmatrix} \subseteq  \textrm{colspan}\begin{bmatrix}\mathbf{H}_1 \\ \mathbf{H}_2\end{bmatrix}.$$
$$\textrm{colspan}\begin{bmatrix}\mathbf{C} \end{bmatrix} \subseteq  \textrm{colspan}\begin{bmatrix}\mathbf{G}\end{bmatrix},$$
	then 
$$\Grank\left([~\mathbf{H}_{1} ~ ~ \mathbf{A}~], [~\mathbf{H}_{2} ~~ \mathbf{B}~] ,[~\mathbf{G}_{2}~~\mathbf{C}~]\right) = \Grank(\mathbf{H}_{1},\mathbf{H}_{2},\mathbf{G}_{2}).$$
\end{enumerate}
   %\begin{align*}
   %    \mathbf{M} &= 
   %    \left[
   %    \begin{array}{c|c}
   %        \mathbf{H}_1 & \mathbf{G}_1 \\
   %        \mathbf{H}_2 & \mathbf{G}_2 
   %\end{array} \right], 
   %&
   % \mathbf{M}^\prime &= 
   % \left[
   % \begin{array}{cc|cc}
   %        \mathbf{H}_1 &\mathbf{A}_1 & \mathbf{G}_1 &\mathbf{B}_1 \\
           %\mathbf{H}_2 &\mathbf{A}_2 & \mathbf{G}_2 &\mathbf{B}_2
       %\end{array} 
  % \right]
  % \end{align*}
  % then, $\Grank(\mathbf{M}^\prime) \geq \Grank(\mathbf{M})$.
      \label{lma:grankcol}
\end{lemma}
{ Statement $(i)$ follows from submodularity of the rank function. Statements $(ii)$ and $(iii)$ follow from
elementary properties of the rank of a matrix and the definition of the $\Grank$. We omit a proof of the lemma here. } 

Next, we state a lemma that will be useful later on in generating our linear coding solutions.

\begin{lemma}
    Let $\mathbf{H}_{1},\mathbf{H}_{2},\mathbf{G}_{2}, \mathbf{A},\mathbf{B}$ be matrices with entries from a finite field $\mathbb{F},$ respectively having  of dimensions $P_1 \times
    Q_1$, $P_2 \times Q_1,$ $P_2 \times Q_2, P_1 \times M$ and $P_2 \times M$. Suppose that $$
    \Grank([\mathbf{H}_{1}~~\mathbf{A}], [\mathbf{H}_{2}~~\mathbf{B}], \mathbf{G}_{2}) >
    \Grank(\mathbf{H}_{1},\mathbf{H}_{2},\mathbf{G}_{2})$$
	Then, there exists a $M \times 1$ column vector $\mathbf{f}$ such that 
\begin{equation}
    \Grank([\mathbf{H}_{1}~~\mathbf{A}\mathbf{f}], [\mathbf{H}_{2}~~\mathbf{B}\mathbf{f}], \mathbf{G}_{2}) = \Grank(\mathbf{H}_{1},\mathbf{H}_{2},\mathbf{G}_{2})+1\label{eq:resultingcondition}
\end{equation}
Furthermore, if 
\begin{itemize}
		\item[(i)] $\textrm{colspan}(\mathbf{B}) \subset \textrm{colspan}(\left[\mathbf{H}_{2}~~\mathbf{G}_{2}\right])$
		\item[(ii)] $\textrm{colspan}(\mathbf{B}) \not\subset \textrm{colspan}(\mathbf{H}_{2}),$
	\end{itemize}
then, choosing a vector $\mathbf{v}$ randomly from the null space of the column space of $\left[\mathbf{H}_2 ~ ~
\mathbf{B}\right]$ and 
%setting $\mathbf{f}$ to be the last $M$ columns of $\mathbf{v}$ satisfies
{
setting $\mathbf{f}$ to be the last $M$ entries of $\mathbf{v}$ satisfies
}
(\ref{eq:resultingcondition}) with a probability that approaches $1$ as the field size $|\mathbb{F}|$ increases. In this case, the
vector $\mathbf{f}$ satisfies the following property: $\mathbf{B}\mathbf{f} \in \textrm{colspan}(\mathbf{H}_2)$.  If
$(i)$ or $(ii)$ are not satisfied, then picking the entries of $\mathbf{f}$ randomly and uniformly over the field
satisfies (\ref{eq:resultingcondition}) with a probability that approaches $1$ as the field size increases.

\label{lem:alignment}
\end{lemma}
Since adding a single column can increase the $\Grank$ by at most $1$, the vector $\mathbf{f}$ that satisfies
$\eqref{eq:resultingcondition}$ maximizes the $\Grank([\mathbf{H}_{1}~~\mathbf{A}\mathbf{f}],
[\mathbf{H}_{2}~~\mathbf{B}\mathbf{f}], \mathbf{G}_{2}).$  The vector $\mathbf{f}$ will be useful in obtaining the code
construction in the recursive coding routine.  The approach of choosing $\mathbf{f}$ randomly to satisfy
\eqref{eq:resultingcondition} follows the spirit of random linear network coding \cite{Ho_etal_Award}. 
 
\subsection{Destination reduction}
\label{subsec:destreduction}
The destination reduction algorithm takes the original problem $\Omega = (\mathcal{G}, \mathcal{S}, \mathcal{T})$ and
generates a sequence of $N+1$ ordered two-unicast-Z network problems, for some $N \in \mathbb{Z}^+$, starting with the
original problem itself. We denote the sequence of problems as $\mathbb{P} = \left(\Omega^{(0)}, \Omega^{(1)},
\Omega^{(2)}, \dots, \Omega^{(N)} \right)$, where $\Omega^{(0)} = \Omega$ and $\Omega^{(i)} = (\mathcal{G}, \mathcal{S},
\mathcal{T}^{(i)})$ with $\mathcal{T}^{(i)} = \left\{ T_1^{(i)}, T_2^{(i)} \right\}$ being the destination sets for the
problem number $i$. In particular, all the problems have the same underlying graph $\mathcal{G}$ and source set
$\mathcal{S}$, but different destination sets, i.e., $\mathcal{T}^{(i)} \neq \mathcal{T}^{(j)}, i\neq j$. The algorithm
is formally described in Algorithm \ref{alg:dest}, in which the key procedure is to sequentially generate
$\Omega^{(i+1)}$ from the previous problem $\Omega^{(i)}$. We describe the process informally here.

\begin{algorithm}[h]
    \caption{Destination reduction algorithm} \label{alg:dest}
    \begin{algorithmic}[1]
        \Procedure{Reduction}{$\Omega^{(0)}$}
    \State $\mathbb{P} \gets ()$
    \State add $\Omega^{(0)}$ to $\mathbb{P}$
    \State $i \gets 0$
    \State $S \gets \left\{ e: \Tail(e) \in \mathcal{S} \right\}$
    \While{$T^{(i)}_1 \cup T^{(i)}_2 \not \subseteq S$}
    \State $E \gets \left\{ e: \arg \max_{e \in T^{(i)}_1 \cup T^{(i)}_2} \Ord(e)\right\}$
    \State $E_j \gets E \cap T^{(i)}_j, j = 1, 2$ 
    \State $v \gets \Tail(E)$ 
    \For{$j \gets 1, 2$}
        \If{$E_j \neq \varnothing$}
        \State $T^{(i+1)}_j \gets \left(T^{(i)}_j\backslash E_j \right) \cup \In(v)$
        \Else
        \State $T^{(i+1)}_j \gets T^{(i)}_j$
        \EndIf
    \EndFor
    \State $\mathcal{T}^{(i+1)} \gets \left\{ T^{(i+1)}_1, T^{(i+1)}_2 \right\}$
    \State $\Omega^{(i+1)} \gets (\mathcal{G}, \mathcal{S},
    \mathcal{T}^{(i+1)})$
    \State add $\Omega^{(i+1)}$ into $\mathbb{P}$
    \State $i \gets i+1$
    \EndWhile
    \State \textbf{return} $\mathbb{P}$
    \EndProcedure
    \end{algorithmic}
\end{algorithm}
\begin{figure}[tbp]
    \centering
    \includegraphics[scale=0.7]{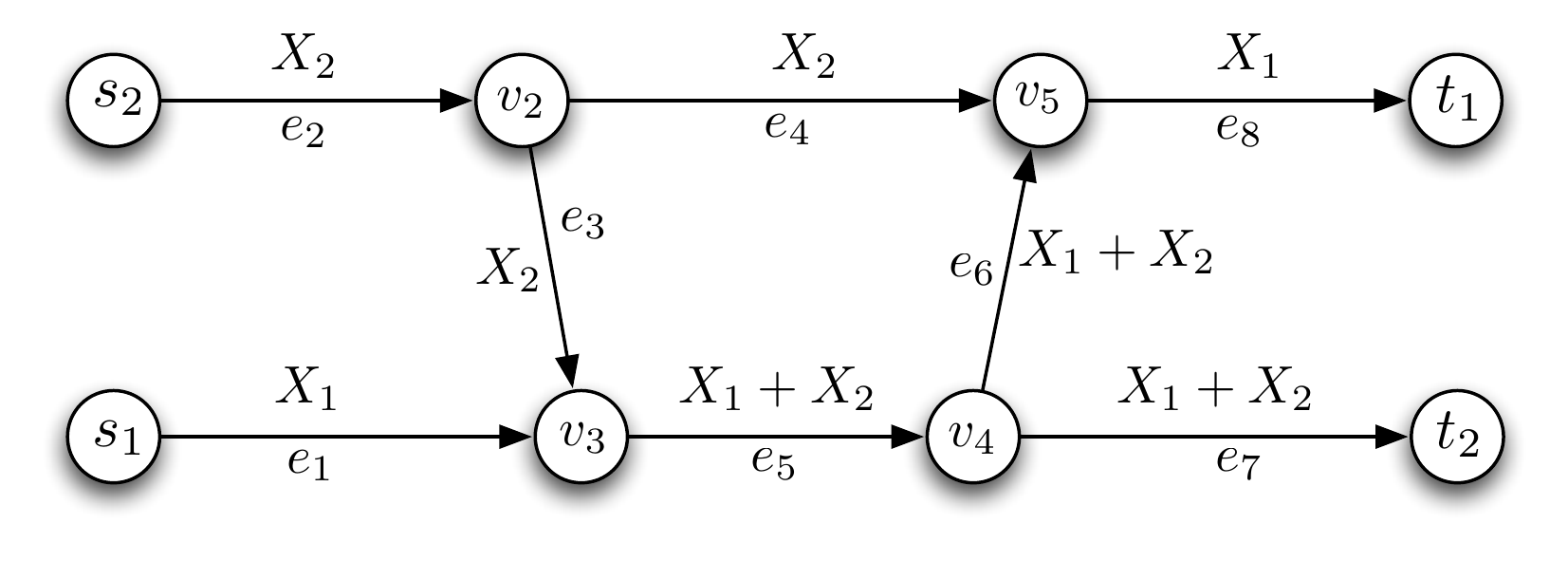}
    \caption{An example network used to demonstrate our algorithm operation}
    \label{fig:example1}
\end{figure}

{
Recall that in a directed acyclic graph, there is a total ordering $\Ord$ on the vertices of the graph. Also recall that
$\Ord$ induces a partial ordering on the edges, where the set of edges of the same topological order share a common tail
node. In brief, the destination reduction algorithm obtains problem $\Omega^{(i+1)}$ from $\Omega^{(i)}$ as follows.  We
find all the highest topologically ordered edges in the union of the two destination sets. In each destination set, if it
contains any of these edges, we replace them with their immediate parent edges. 
} 
{
Specifically, 
%the destination
%reduction algorithm obtains the problem $\Omega^{(i+1)}$ from $\Omega^{(i)}$ as follows.
given $\Omega^{(i)} = (\mathcal{G}, \mathcal{S}, \mathcal{T}^{(i)})$, let $E$ denote the set of edges in $T_1^{(i)} \cup T_2^{(i)}$ with the highest topological order. In other words, all edges in $E \subset T_1^{(i)} \cup T_2^{(i)}$ have the same topological order, and a strictly higher topological order with respect to every edge in $T_{1}^{(i)} \cup T_{2}^{(i)}\backslash E$.    
} For $j \in \left\{ 1, 2 \right\}$, let $E_j = T^{(i)}_j \cap {E}$. For each destination $j \in \left\{ 1, 2
\right\}$, if $T^{(i)}_j$ does not contain any highest topological ordered edge, i.e., if  $E_j = \varnothing$, then the
destination set remains unchanged in $\Omega^{(i+1)}$, i.e., $T^{(i+1)}_j = T^{(i)}_j$. Otherwise, all edges in $E_j$ are
removed in $T^{(i)}_j$ and replaced by $\In(v)$ to produce the new destination set $T^{(i+1)}_j$ in
$\mathcal{T}^{(i+1)}$, {that is, $T_j^{(i+1)} = (T_j^{(i)}\backslash E_{j})\cup \textrm{In}(v)$. }

\begin{figure}[h]
    \centering
    \includegraphics[scale=0.6]{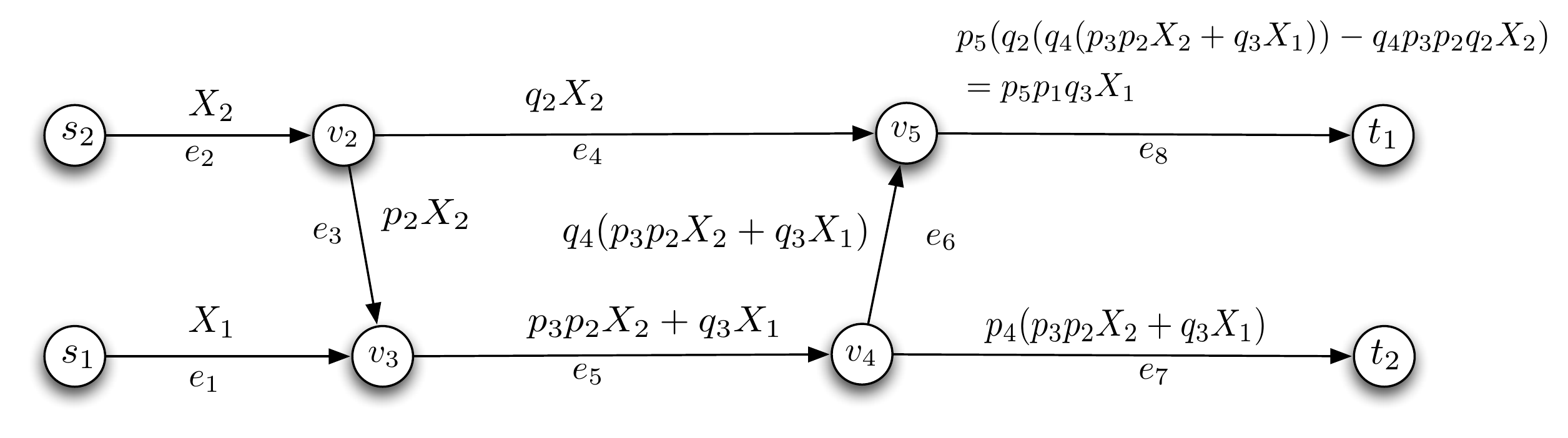}
    \begin{tabular}{|c|c|c|c|c|c|}
        \hline
        & $i=0$ & $i=1$ & $i=2$ & $i=3$ & $i=4$ \\
        \hline
        Destination $T_1^{(i)}$ & $e_8$ & $e_4, e_6$ & $e_4, e_5$ & $e_1, e_3, e_4$ & $e_1, e_2$ \\
        \hline
        Destination $T_2^{(i)}$ & $e_7$ & $e_7$ & $e_5$ & $e_1, e_3$ & $e_1, e_2$ \\
        \hline
    \end{tabular}
    \caption{The destination and recursive coding algorithms shown for the network of Fig. \ref{fig:example1}. The
scalars $p_2, q_2, p_3, p_4, q_4, p_5$ are chosen randomly and independently of each other. Note that the recursive
coding algorithm operating at Stage $0$ performs alignment step at vertex $v_5$.}
\label{fig:recursive_coding}
\end{figure}

{Before proceeding to describing our coding scheme, we list some useful and instructive properties of the destination
reduction algorithm; these properties can be easily checked for the example in Fig. \ref{fig:recursive_coding}.}

{
\begin{enumerate}[Property (i)]%\setlength{\itemindent}{0.6 in}

\item The set of edges $T_j^{(i)} \backslash T_j^{(i+1)}$ has a common tail node $v$, which also forms the common head
    node of all the edges in $T_{j}^{(i+1)} \backslash T_j^{(i)}.$ Furthermore, there are only two possibilities:
    $T_j^{(i+1)} \backslash T_j^{(i)}$ is empty or $T_j^{(i+1)} \backslash T_j^{(i)} = \In(v)$. 
\item An edge in the graph which communicates to at least one edge in $T_1$ appears in $T_1^{(i)}$ for some value 
    $i \in \{0,1,2,\ldots,N\}$. An edge which communicates to at least one edge in $T_2$ appears in $T_2^{(i)}$ for some value of $i$. Consider
    an edge $e$ which communicates to at least one edge in $T_1$ and at least one edge in $T_2.$ Let $k_1$ denote the
    largest number such that the edge $e$ belongs to $T_1^{(k_1)}.$ Let $k_2$ denote the largest number such that the
    edge $e$ belongs to $T_2^{(k_2)}.$ Then $k_1=k_2$.    
\item Consider a two-unicast-$Z$ problem $\Omega$, where every edge in the graph is connected to at least one of the destinations,
    that is, there is a path from every edge to at least one edge in $T_1 \cup T_2$. Then, the set of all edges have a
    lower topological order with respect to $\left(T_1^{(i)} \cup T_2^{(i)}\right)  \backslash \left(T_1^{(i+1)} \cup
    T_2^{(i+1)}\right)$ in $\mathcal{G}$ is equal to $$\bigcup_{i+1 \leq k \leq N} T_{1}^{(k)} \cup T_{2}^{(k)}.$$
\item In the $\Omega^{(N)},$ the destination edges are collocated with the source edges. That is $T_1^{(N)} \cup
    T_2^{(N)} \subseteq S_1 \cup S_2.$ Furthermore, if every edge emanating from the source nodes communicates with at
    least one of the destination nodes, then $T_1^{(N)} \cup T_2^{(N)} = S_1 \cup S_2.$
    {
\item For all $i$, $0 \leq i \leq N$ and $j = 1,2$, the destination set $T_j^{(i)}$ forms a cut set between both sources
    $s_1, s_2$ and the destination set $T_j$ for the original problem $\Omega$.
}
\end{enumerate}
Properties (i) and (iv) can be verified by examining the algorithm. We prove Properties (ii), (iii) and (v) in Appendix
\ref{app:destreduction}. It is instructive to note that Properties (ii) and (iii) imply that the collection of sets
$$\left\{\left(T_1^{(i)} \cup T_2^{(i)}\right)\backslash \left(T_1^{(i+1)} \cup
T_2^{(i+1)}\right):i=0,1,2,\ldots,N\right\}$$ forms a partition of the set of edges of the graph. Furthermore, this
partition is the same as the partition implied by the topological ordering on the edges, i.e., all edges of the same
topological order belong to one unique member of this partition. As we observe next in our description of the recursive
coding algorithm, the recursion at depth $j$ designs the local coding co-efficients for the edges in the $j$th member of
this partition.

}

\subsection{Recursive coding construction}
\label{subsec:recursivecoding}
%{\color{blue}
We describe the recursive coding algorithm formally in Algorithm \ref{alg:recurs}. We present an informal description
here. Without loss of generality, we only consider networks where every source edge communicates with at least one edge
in $T_1 \cup T_2$. The first step of the recursive coding construction begins with a trivial coding scheme for
$\Omega^{(N)}$. In particular, note that Property (iv) states that $T_1^{(N)} \cup T_2^{(N)}$ is equal to $S_1 \cup
S_2$. We set the local coding vector for an edge in $T_1^{(N)} \cup T_2^{(N)}$ to be the vector with co-efficient $1$
corresponding to the edge and $0$ elsewhere. We assume that the local coding vectors for all the edges outside of
$T_1^{(N)} \cup T_2^{(N)}$ to be indeterminate at this point; the local coding co-efficients for these edges will be
determined using the recursive coding algorithm. 
{
Each step of the recursion is referred to as a \emph{stage}. The recursive algorithm has $N$ stages, where at stage $i$,
the algorithm generates the code for $\Omega^{(i)}$ using the coding scheme for the previous stage
for $\Omega^{(i+1)}$.
}
In particular, the recursive coding algorithm accomplishes the
following: Given a linear coding scheme for the $(i+1)$th stage, that is, for $\Omega^{(i+1)},$ the algorithm at the $i$th
stage constructs a linear coding scheme for $\Omega^{(i)}$. Starting with the trivial coding scheme for $\Omega^{(N)}$,
our algorithm recursively constructs coding schemes for the problems $\Omega^{(N-1)}, \Omega^{(N-2)},\ldots,
\Omega^{(1)},$ which eventually leads to a coding scheme for the original problem $\Omega = \Omega^{(0)}$. 
{
Next we focus on the coding algorithm at stage $i$ assuming a linear coding scheme for $\Omega^{(i+1)}$ is given in the
previous stage.
}
%For convenience, we use stages to denote the status of algorithm as
%the recursive coding takes places. At stage $i$, the algorithm finishes coding for $\Omega^{(i)}$ from the linear coding
%shceme generated previously for $\Omega^{(i+1)}$.
%Each step of the recursion is referred to as a \emph{stage}. The recursive algorithm has $N$ stages, where at stage $i$,
%the algorithm generates the linear code $\Omega^{(i)}$ using the linear coding scheme generated in the previous stage
%for $\Omega^{(i+1)}$.

%The recursive coding algorithm takes as input, a code for problem
%	$\Omega^{(i+1)} \in \mathbb{P}, 0 \leq i \leq N-1$, and generates as an output, a code for the
%problem $\Omega^{(i)}$. We define some subsets of edges that will be used by the
%algorithm in this process. Consider the destination sets of problems
%$\Omega^{(i+1)}$ and $\Omega^{(i)}$. For 

A solution to the problem $\Omega^{(i+1)}$ will describe local coding co-efficients for all the edges in $\bigcup_{i+1
\leq k \leq N} T_1^{(i+1)} \cup T_{2}^{(i+1)},$ and leave the co-efficients for the remaining edges to be indeterminate.
Given a linear coding solution for the problem $\Omega^{(i+1)},$ the coding solution for the problem $\Omega^{(i)}$
inherits the linear coding co-efficients from the solution to $\Omega^{(i+1)}$ for all edges in $\bigcup_{i+1 \leq k
\leq N} T_1^{(i+1)} \cup T_{2}^{(i+1)}.$ To complete the description for a solution to $\Omega^{(i)}$, the algorithm
specifies the local coding co-efficients for edges in $\left(T_1^{(i)}\cup T_2^{(i)}\right) \backslash
\left(T_1^{(i+1)}\cup T_2^{(i+1)} \right).$ Because of Properties (ii) and (iii) of the destination reduction algorithm,
we note that specifying local coding co-efficients for edge in $\left(T_1^{(i)}\cup T_2^{(i)}\right) \backslash
\left(T_1^{(i+1)}\cup T_2^{(i+1)} \right)$ suffices to specify the global coding-coefficients for these edges as well,
since all the edges of which have a lower topological order with respect to this set have been assigned coding
co-efficients in the solution to $\Omega^{(i+1)}.$

We use the following notation in our description. For $j = 1,2$, let
\begin{align}
    U^{(i)}_j &= T^{(i)}_j  \cap T^{(i+1)}_j, & 
    I^{(i)}_j &=  T^{(i+1)}_j \backslash U^{(i)}_j,&
    O^{(i)}_j &=  T^{(i)}_j \backslash U^{(i)}_j.
\label{eq:notation1}
\end{align}

Recall from Property (i) that  all the edges in $T_j^{(i)} \backslash T_j^{(i+1)}$ have a common tail node $v$, which is also the head of all the edges in $T_{j}^{(i+1)} \backslash T_{j}^{(i)}$. The set $I_j^{(i)}$ is therefore contained in the set of \emph{incoming} edges on to this node $v$. The set $O_{j}^{(i)}$ is contained in the \emph{outgoing} edges from $v$. Based on Property (i), we observe that if $I_{1}^{(i)} \neq \phi, I_{2}^{(i)} \neq \phi,$ then $I_{1}^{(i)} = I_{2}^{(i)} = \textrm{In}(v)$. The set $U^{(i)}_j$ is the set of \emph{unchanged} destination edges between $T_j^{(i+1)}$ and $T_j^{(i)}$. 
 We divide $O^{(i)}_2$ into two disjoint subsets, $A^{(i)}_2$ and $B^{(i)}_2$, such that $O^{(i)}_2 =
A^{(i)}_2 \cup B^{(i)}_2$, where
\begin{align}
    A^{(i)}_2 &= O^{(i)}_2 \cap O^{(i)}_1 \, , & 
    B^{(i)}_2 &= O^{(i)}_2 \backslash O^{(i)}_1.
\label{eq:notation2}
\end{align}
Note that Property (ii) implies that $B_2^{(i)}$ contains edges that communicate with at least one edge in $T_2$ but not
$T_1$. Similarly, $A_2^{(i)}$ contains edges that communicate with at least one edge $T_1$ and at least one edge in
$T_2$. It is useful to note that $A_2^{(i)} \neq \phi \Rightarrow I_1^{(i)} \neq \phi,I_2 \neq \phi \Rightarrow
I_{1}^{(i)} = I_2^{(i)}.$ Since $B_2^{(i)}$ and $O_1^{(i)}$ form a partition of the set $\left(T_1^{(i)}\cup
T_2^{(i)}\right) \backslash \left(T_1^{(i+1)}\cup T_2^{(i+1)} \right),$ we specify local linear coding co-efficients for
$B_2^{(i)}$ and $O_1^{(i)}$ in the recursive coding algorithm. 

%Specifically, note that the edges in $O_j^{(i)}, j \in \{1,2\}$ have a higher topological order with respect to $T_1^{(k)} \cup T_2^{(k)}$ for all $k > i$, and are therefore indeterminate in the solution to $\Omega^{(i+1)}$.  The goal of our recursive coding algorithm is to determine the linear coding for $O_{1}^{(i)}$ and $O_{2}^{(i)}$ from a solution to $\Omega^{(i+1)}$. In other words, the goal of the recursive coding algorithm is to obtain the local coding co-efficients for the edges emanating from vertex $v.$
{
%Recall that when the destination reduction algorithm generates
%$\mathcal{T}^{(i+1)}$ from $\mathcal{T}^{(i)}$, it removes the highest topological
%ordered edges, all of which share a common tail node $v$, and adds the incoming
%edges of $v$ into the destination sets that have edges removed.  

%it suffices to obtain the local coding vectors for all the outgoing edges $O_{j}^{(i)}$. In fact, Property [??] implies that the linear coding solution of    }
\begin{algorithm}[tbp]
    \caption{Recursive Coding Algorithm} \label{alg:recurs}
    \begin{algorithmic}[1]
	    \Procedure{Recursive Coding}{$\mathbb{P},\mathbb{F}$} \Comment{$\mathbb{F}$ represents the field over which the coding is performed}
	    \State Denote $\mathbb{P}$ by $(\Omega^{(i)},\Omega^{(i+1)},\ldots,\Omega^{(N)})$. 
	    \State Denote $\Omega^{(i)} = (\mathcal{G},\mathcal{S},\{T_1^{(i)},T_2^{(i)}\})$ and use the notation (\ref{eq:notation1}),(\ref{eq:notation2}).
	    \State Denote $\mathcal{E}^{(k)} = {\cup_{j=k}^{N}T_1^{(j)}\cup T_2^{(j)}}$ for $k \in \{i,i+1,\ldots,N-1\}$

	    \State \If{ $\textrm{length}(\mathbb{P}) = 1$} {return $\mathbf{I}_{|\mathcal{E}^{(N)}|}$}
	    \EndIf
	    \State $\mathbf{F} = \mathbf{0}_{|\mathcal{E}^{(i)}|\times |\mathcal{E}^{(i)}|}$
	    \State $\mathbf{F}^{\mathcal{E}^{(i+1)},\mathcal{E}^{(i+1)}} = \Call{Recursive Coding}{(\Omega^{(i+1)}, \ldots, \Omega^{(N)})}$  \Comment{The recursion}
	    \State \Comment{In the next few steps, we will describe coding co-efficients $\mathbf{F}^{I_{j}^{(i)},O_{j}^{(i)}}, j \in \{1,2\}$}
	    \State \If{$B_2^{(i)} \neq \varnothing$}  \Comment{Phase 1}
	    \State $\mathbf{F}^{I_{2}^{(i)},B_2^{(i)}} \leftarrow$ \textrm{Uniformly random from the field} $\mathbb{F}$ 
	    \EndIf
	    \State $\overline{O} = \overline{A} = \phi$ \Comment{Temporary variables (sets) used in the for loop next}
	    \State \For{$e \in O_1^{(i)}$} \Comment{Phase 2: Encoding the edges in $O_1^{(i)}$ one edge at a time.} 
	    \State $\overline{O} = \overline{O} \cup e, \overline{A} = \overline{O} \cap A_2^{(i)}$ \Comment{$\overline{O}$ and $\overline{A}$ respectively represent the subsets of $O_{1}^{(i)}$ and $A_2^{(i)}$. In the next few steps, we will find the coding co-efficients for the edges in $\overline{O},\overline{A}$}
	    \State \If{ 
$\Grank\left(\mathbf{H}_{1}^{T_1^{(i+1)}},\mathbf{H}_{2}^{T_1^{(i+1)}},\mathbf{G}_{2}^{T_2^{(i+1)}}\right) >
\Grank\left(\left[\mathbf{H}_1^{U^{(i)}_1} ~ ~ \mathbf{H}_1^{\overline{O}}\right],\left[\mathbf{H}_2^{U^{(i)}_1} ~ ~ \mathbf{H}_2^{\overline{O}}\right],\left[ \mathbf{G}_2^{U^{(i)}_2} ~ ~ \mathbf{G}_2^{\overline{A}} ~ ~\mathbf{G}_2^{B^{(i)}_2}\right]\right)$
\State $$\&~~\Span\left(\mathbf{H}_2^{I^{(i)}_1}\right) \not \subset
\Span\left(
\begin{bmatrix}
    \mathbf{H}_2^{U^{(i)}_1} & \mathbf{H}_2^{\overline{O}}
\end{bmatrix}
\right)$$ $$\&~ ~
\Span\left(\mathbf{H}_2^{I^{(i)}_1}\right) \subset
\Span\left(
\begin{bmatrix}
    \mathbf{H}_2^{U^{(i)}_1} & \mathbf{H}_2^{\overline{O}} & \mathbf{G}_2^{U^{(i)}_2} &
    \mathbf{G}_2^{\overline{A}} & \mathbf{G}_2^{B^{(i)}_2}
\end{bmatrix}
\right)~.$$} \Comment{Alignment Step}
\State $\mathbf{v} \leftarrow \textrm{Random vector in } \textrm{ker}(\mathbf{H}_{2}^{T_{1}^{(i+1)}})$
\State $\mathbf{F}^{I_{1}^{(i)},\{e\}} \leftarrow \textrm{Last }|I_{1}^{(i)}|\textrm{rows of }\mathbf{v}$
%\label{eqn:con2}
\Else
\State $\mathbf{F}^{I_1^{(i)},\{e\}} \leftarrow \textrm{Uniformly at random from the field } \mathbb{F}$ \Comment{Randomization Step}
\EndIf
\EndFor
	    \State 
	    \State \textbf{return} $\mathbf{F}$
    \EndProcedure
    \end{algorithmic}
\end{algorithm}

{ Henceforth we will use the following notation. For any set of edges $P \subseteq T_1^{(i)}$, $i \in \{1,2,\}$, the transfer matrix between source $i$ and $P$ is denoted as $\mathbf{H}_{i}^{P}.$ For any set of edges $Q \subseteq T_2^{(i)}$,$i \in \{1,2\}$, the transfer matrix between source $i$ and $Q$ is denoted as $\mathbf{G}_{i}^{Q}.$ Furthormore, for the sake of consistency, we will assume that the columns of $\mathbf{H}_{i}^{P},\mathbf{G}_{i}^{Q}$ are ordered based on the topological orderings of the edges in $P$ and $Q$. Note that with our notation, $\mathbf{H}_{2}^{T_1^{(i)}\cap T_{2}^{(i)}} = \mathbf{G}_{2}^{T_1^{(i)}\cap T_{2}^{(i)}}.$ With this notation, the local coding vectors at node $v$ generate $\mathbf{H}^{O^{(i)}_1}_j$ from $\mathbf{H}^{I^{(i)}_1}_j$, and
$\mathbf{G}^{A^{(i)}_2}_j$ and $\mathbf{G}^{B^{(i)}_2}_j$ from
$\mathbf{G}^{I^{(i)}_2}_j$, for $j = 1,2$. 
Therefore, the goal of the recursive coding algorithm is to design local coding matrices}
$\mathbf{F}^{I_1^{(i)},O_1^{(i)}}, \mathbf{F}^{I_2^{(i)},B_2^{(i)}}$ at vertex
$v$, with dimensions $|\In(v)| \times |O^{(i)}_1|$, and $|\In(v)| \times |B^{(i)}_2|$ respectively. The global coding co-efficients will be determined as, for $j \in \{1,2\},$
\begin{align}
    \mathbf{H}_j^{O^{(i)}_1} 
    =
    \mathbf{H}_j^{I^{(i)}_1}
\mathbf{F}^{I_1^{(i)},O_1^{(i)}}
   ~, 
    \mathbf{G}_j^{A^{(i)}_2}
    = 
    \mathbf{G}_j^{I^{(i)}_2} 
\mathbf{F}^{I_2^{(i)},A_2^{(i)}}
   ~,
    \mathbf{G}_j^{B^{(i)}_2}  
    = 
    \mathbf{G}_j^{I^{(i)}_2}   
\mathbf{F}^{I_2^{(i)},B_2^{(i)}}
    \label{eq:stageicoding}
\end{align}
where note that, if $A_2^{(i)}$ is non-empty, then $\mathbf{F}^{I_2^{(i)},A_2^{(i)}}$ is a sub-matrix of  $\mathbf{F}^{I_1^{(i)},O_1^{(i)}}.$ This is because $A_2^{(i)} \subseteq O_1^{(i)}$ and, if $A_2^{(i)}$ is non-empty, $I_1^{(i)} = I_{2}^{(i)}$. 

%algorithm takes as input $\mathbf{H}_{1}^{I_1^{(i)}}, \mathbf{H}_{1}^{I_1^{(i)}}$ at node $v$.

%\begin{align}
    %\begin{bmatrix}
    %\mathbf{H}_1^{O^{(i)}_1}   \\
   %\mathbf{H}_2^{O^{(i)}_1}  
    %\end{bmatrix} 
    %&=
%\begin{bmatrix}
    %\mathbf{H}_1^{I^{(i)}_1}   \\
   %\mathbf{H}_2^{I^{(i)}_1}  
    %\end{bmatrix} 
 %\mathbf{F}^{(i)}_1 
   %,& 
    %\begin{bmatrix}
    %\mathbf{G}_1^{A^{(i)}_2}   \\
   %\mathbf{G}_2^{A^{(i)}_2}  
    %\end{bmatrix} 
    %&= 
%\begin{bmatrix}
    %\mathbf{G}_1^{I^{(i)}_2}   \\
   %\mathbf{G}_2^{I^{(i)}_2}  
    %\end{bmatrix} 
%\mathbf{F}^{(i)}_A 
%,&
    %\begin{bmatrix}
    %\mathbf{G}_1^{B^{(i)}_2}   \\
   %\mathbf{G}_2^{B^{(i)}_2}  
    %\end{bmatrix} 
    %&= 
%\begin{bmatrix}
    %\mathbf{G}_1^{I^{(i)}_2}   \\
   %\mathbf{G}_2^{I^{(i)}_2}  
    %\end{bmatrix} 
%\mathbf{F}^{(i)}_B~. 
%\end{align}

\textbf{Informal Description of the Recursive Coding Algorithm:} A formal description is described by Algorithm
\ref{alg:recurs}. Here, we present an informal description of the recursive coding algorithm for the $i$th stage,
assuming that stage $i+1$ is complete. As previously stated, our goal is to generate $\mathbf{F}^{I_1^{(i)},O_1^{(i)}},
\mathbf{F}^{I_2^{(i)},B_2^{(i)}}$ in (\ref{eq:stageicoding}).  We do this in two phases, in the first phase, we find
$\mathbf{F}^{I_2^{(i)},B_2^{(i)}}$ and in the second phase, we determine $\mathbf{F}^{I_1^{(i)},O_1^{(i)}}$. Our
strategy is based on the idea of designing the coding co-efficients so that the Grank is maximized at each stage. 

\subsubsection{Phase 1}
%determining $\mathbf{G}^{B_2^{(i)}}$.

If $B_2^{(i)} = \varnothing$, proceed to the next phase. Otherwise, select each
entry of the matrix $\mathbf{F}^{I_2^{(i)},B_2^{(i)}}$ uniformly at random from the underlying
finite field $\mathbb{F}$. We note that, over a sufficiently large field, our choice of $\mathbf{F}^{I_2^{(i)},B_2^{(i)}}$  maximizes
$$\Grank\left(\mathbf{H}_{1}^{T_1^{(i+1)}},\mathbf{H}_{2}^{T_1^{(i+1)}},\left[\mathbf{G}_{2} ~~
\mathbf{G}_{2}^{I_2^{(i)}}\mathbf{F}^{I_2^{(i)},B_2^{(i)}}\right]\right)$$ with high probability. This is because
choosing the entries of $\mathbf{F}^{I_2^{(i)},B_2^{(i)}}$ uniformly at random maximizes, the
rank of $\left[\mathbf{H}_{2}^{T_1^{(i+1)}}~~\mathbf{G}_{2} ~~
\mathbf{G}_{2}^{I_2^{(i)}}\mathbf{F}^{I_2^{(i),B_2^{(i)}}}\right]$
with a probability that tends to $1$ as $|\mathbb{F}|$ increases. After phase 1, we have found coding vectors for edges $B_2^{(i)}$.
\subsubsection{Phase 2} 
%determining $\mathbf{H}^{O_1^{(i)}}$.

Let $O_1^{(i)} = \left\{e_{i,1}, e_{i,2}, \dots,
e_{i,m}  \right\}$, where $m = |O_1^{(i)}|$. We design the coding co-efficients for the $m$ edges in $O_1^{(i)}$
one-by-one, with the co-efficients designed to maximize the Grank at each step. Our choice of coding co-efficients is
motivated by Lemma \ref{lem:alignment}. {In particular, the second phase is divided into $q$ steps, where in the $j$th
    step, we design the coding co-efficients for edge $e_{i,j}$.} Each step is classified as a \emph{alignment step} or
    a \emph{randomization step} as follows.

Let $O_{1,0}^{(i)} = \varnothing$ and $O_{1,j}^{(i)} =
\left\{ e_{i,1}, \dots , e_{i,j} \right\}$ be the subset of the first $j$
elements of $O_1^{(i)}$ for $1 \leq j \leq m$. Let $A_{2,j}^{(i)} = A_2^{(i)}
\bigcap O_{1,j}^{(i)}.$ 
%\textbf{Case I:} $\Grank\left( \mathbf{M}^{(i+1)} \right) =
%\Grank\left(\overline{\mathbf{M}}^{(i)}\right)$.

%In this case, pick each entry of the matrices $\mathbf{F}^{(i)}_1, j = 1,2$
%uniformly at random from the underlying finite field $\mathbb{F}$.

\textbf{Case I (Alignment):} 
The following conditions are satisfied.
\begin{align}
\Grank\left(\mathbf{H}_{1}^{T_1^{(i+1)}},\mathbf{H}_{2}^{T_1^{(i+1)}},\mathbf{G}_{2}^{T_2^{(i+1)}}\right) &>
\Grank\left(\left[\mathbf{H}_1^{U^{(i)}_1} ~ ~ \mathbf{H}_1^{O_{1,j-1}^{(i)}}\right],\left[\mathbf{H}_2^{U^{(i)}_1} ~ ~ \mathbf{H}_2^{O_{1,j-1}^{(i)}}\right],\left[ \mathbf{G}_1^{U^{(i)}_2} ~ ~ \mathbf{G}_1^{A_{2,j-1}^{(i)}} ~ ~\mathbf{G}_1^{B^{(i)}_2}\right]\right)\label{eqn:con0}\\
\Span\left(\mathbf{H}_2^{I^{(i)}_1}\right) &\not \subset
\Span\left(
\begin{bmatrix}
    \mathbf{H}_2^{U^{(i)}_1} & \mathbf{H}_2^{O_{1,j-1}^{(i)}}
\end{bmatrix}
\right) 
\label{eqn:con1}
\\
\Span\left(\mathbf{H}_2^{I^{(i)}_1}\right) &\subset
\Span\left(
\begin{bmatrix}
    \mathbf{H}_2^{U^{(i)}_1} & \mathbf{H}_2^{O_{1,j-1}^{(i)}} & \mathbf{G}_2^{U^{(i)}_2} &
    \mathbf{G}_2^{A^{(i)}_{2,j-1}} & \mathbf{G}_2^{B^{(i)}_2}
\end{bmatrix}
\right)~.
\label{eqn:con2}
\end{align}
In this case, we choose a vector $\mathbf{v}$ randomly in the nullspace of the column
space spanned by $\mathbf{H}_2^{T_1^{(i+1)}}$ and choose the column vector
$\mathbf{F}^{I_1^{(i)},\{e_{ij}\}}$ to be the last $|\mathbf{I}_{1}^{(i)}|$ rows
of column vector $\mathbf{v}$.  Note that the sets $O_{1,j}^{(i)},A_2^{(i)}$ are denoted by temporary variables $\overline{O},\overline{A}$ in Algorithm \ref{alg:recurs}.

\textbf{Case II (Randomization):} Otherwise, that is, at least one of the conditions (\ref{eqn:con0}),(\ref{eqn:con1}),(\ref{eqn:con2}) is violated. In this case, we select the vector $\mathbf{F}^{I_1^{(i)},\{e_{ij}\}}$ by choosing each of its entries uniformly at random from $\mathbb{F}$. 

If (\ref{eqn:con0})-(\ref{eqn:con2}) are satisfied, then
we refer to our coding step as the \emph{alignment} step. If at least one of (\ref{eqn:con0})-(\ref{eqn:con2}) is violated, we refer to the coding step as a \emph{randomization} step. After the $q$ steps of Phase 2, the coding co-efficients for 
$$\bigcup_{i \leq k \leq N} T_1^{(i+1)} \cup T_{2}^{(i+1)},$$ 
are determined. This completes the description of the recursive coding algorithm. 
Before proceeding, we discuss some properties of our algorithm.
\subsubsection*{Discussion}
  We note that if (\ref{eqn:con0}) is satisfied, then, for a sufficiently large field size, we can show using Lemma \ref{lem:alignment} that 
\begin{align*}
&\Grank\left(\left[\mathbf{H}_1^{U^{(i)}_1} ~ ~ \mathbf{H}_1^{O_{1,j}^{(i)}}\right],\left[\mathbf{H}_2^{U^{(i)}_1} ~ ~ \mathbf{H}_2^{O_{1,j}^{(i)}}\right],\left[ \mathbf{G}_2^{U^{(i)}_2} ~ ~ \mathbf{G}_2^{A_{2,j}^{(i)}} ~ ~\mathbf{G}_2^{B^{(i)}_2}\right]\right) \\&= 
\Grank\left(\left[\mathbf{H}_1^{U^{(i)}_1} ~ ~ \mathbf{H}_1^{O_{1,j-1}^{(i)}}\right],\left[\mathbf{H}_2^{U^{(i)}_1} ~ ~ \mathbf{H}_2^{O_{1,j-1}^{(i)}}\right],\left[ \mathbf{G}_2^{U^{(i)}_2} ~ ~ \mathbf{G}_2^{A_{2,j-1}^{(i)}} ~ ~\mathbf{G}_2^{B^{(i)}_2}\right]\right)+1
\end{align*}
with a probability that approaches $1$ if the size of the field $\mathbb{F}$ grows arbitrarily. We show this below 
\begin{align*}
&\Grank\left(\left[\mathbf{H}_1^{T_1^{(i+1)}} ~ ~ \mathbf{H}_1^{U^{(i)}_1} ~ ~ \mathbf{H}_1^{O_{1,j}^{(i)}}\right],\left[\mathbf{H}_2^{T^{(i+1)}_1} ~ ~ \mathbf{H}_2^{U^{(i)}_1} ~ ~ \mathbf{H}_2^{O_{1,j}^{(i)}}\right],\left[ \mathbf{G}_2^{T^{(i+1)}_2} ~ ~ \mathbf{G}_2^{U^{(i)}_2} ~ ~ \mathbf{G}_2^{A_{2,j-1}^{(i)}} ~ ~\mathbf{G}_2^{B^{(i)}_2} \right]\right)\\
& \stackrel{(a)}{=} \Grank\left(\left[\mathbf{H}_1^{T_1^{(i+1)}}\right],\left[\mathbf{H}_2^{T^{(i+1)}_1}\right],\left[ \mathbf{G}_2^{T^{(i+1)}_2}\right]\right)\\
& \stackrel{(b)}{>} \Grank\left(\left[\mathbf{H}_1^{U^{(i)}_1} ~ ~ \mathbf{H}_1^{O_{1,j-1}^{(i)}}\right],\left[\mathbf{H}_2^{U^{(i)}_1} ~ ~ \mathbf{H}_2^{O_{1,j-1}^{(i)}}\right],\left[ \mathbf{G}_2^{U^{(i)}_2} ~ ~ \mathbf{G}_2^{A_{2,j-1}^{(i)}} ~ ~\mathbf{G}_2^{B^{(i)}_2}\right]\right)\\
%\Grank\left(\left[\mathbf{H}_1^{U^{(i)}_1} ~ ~ \mathbf{H}_1^{O_{1,j}^{(i)}}\right],\left[\mathbf{H}_2^{U^{(i)}_1} ~ ~ \mathbf{H}_2^{O_{1,j}^{(i)}}\right],\left[ \mathbf{G}_1^{U^{(i)}_2} ~ ~ \mathbf{G}_1^{A_{2,j}^{(i)}} ~ ~\mathbf{G}_1^{B^{(i)}_2}\right]\right)
\end{align*}
\begin{align*}
\Rightarrow 
&\Grank\left(\left[\mathbf{H}_1^{U^{(i)}_1} ~ ~ \mathbf{H}_1^{O_{1,j}^{(i)}}\right],\left[\mathbf{H}_2^{U^{(i)}_1} ~ ~ \mathbf{H}_2^{O_{1,j}^{(i)}}\right],\left[ \mathbf{G}_2^{U^{(i)}_2} ~ ~ \mathbf{G}_2^{A_{2,j}^{(i)}} ~ ~\mathbf{G}_2^{B^{(i)}_2}\right]\right) \\
& = \Grank\left(\left[\mathbf{H}_1^{U^{(i)}_1} ~ ~ \mathbf{H}_1^{O_{1,j-1}^{(i)}} ~~ \mathbf{H}_1^{T_1^{(i+1)}}\mathbf{F}^{T_1^{(i+1)},\{e_{i,j}\}}\right],\left[\mathbf{H}_2^{U^{(i)}_1}  ~~ \mathbf{H}_2^{O_{1,j-1}^{(i)}} ~ ~ \mathbf{H}_2^{T_1^{(i+1)}}\mathbf{F}^{T_1^{(i+1)},\{e_{i,j}\}}\right]\right. \\ &~ ~ ~ ~ ~ ~ ~ ~ ~ ~ ~ ~ ~ ~ ~ \left.\left[ \mathbf{G}_2^{U^{(i)}_2} ~ ~ \mathbf{G}_2^{A_{2,j-1}^{(i)}} ~ ~ \mathbf{G}_2^{B^{(i)}_2}  ~ ~ \mathbf{G}_2^{T_2^{(i+1)}}\mathbf{F}^{T_1^{(i+1)},\{e_{i,j}\}}\right]\right) \\&\stackrel{(c)}{=} 
\Grank\left(\left[\mathbf{H}_1^{U^{(i)}_1} ~ ~ \mathbf{H}_1^{O_{1,j-1}^{(i)}}\right],\left[\mathbf{H}_2^{U^{(i)}_1} ~ ~ \mathbf{H}_2^{O_{1,j-1}^{(i)}}\right],\left[ \mathbf{G}_2^{U^{(i)}_2} ~ ~ \mathbf{G}_2^{A_{2,j-1}^{(i)}} ~ ~\mathbf{G}_2^{B^{(i)}_2}\right]\right)+1
\end{align*}
where $(a)$ follows from Statement $(iii)$ of Lemma \ref{lma:grankcol}, $(b)$ follows from (\ref{eqn:con0}) and $(c)$ follows from Lemma \ref{lem:alignment}.

{ 
It is also instructive to note that, if $p > q,$
$$ \Grank(\mathbf{H}_1^{T_1^{(p)}},\mathbf{H}_{2}^{T_1^{(p)}},\mathbf{G}_2^{T_2^{(p)}}) \geq \Grank(\mathbf{H}_1^{T_1^{(q)}},\mathbf{H}_{2}^{T_1^{(q)}},\mathbf{G}_2^{T_2^{(q)}}).$$
That is, the Grank of a stage $q$ is no smaller than stage $p$ which is downstream with respect to a stage $q.$ 
}

We note that our algorithm performs better than optimal routing (multi-commodity flow) for certain networks, for e.g., Fig. \ref{fig:example1}. On the other hand, for certain networks our simulations have revealed that routing outperforms our algorithm. %One such network is provided in Fig. \ref{fig:routingbeatsus} in the appendix.

\subsubsection*{Complexity}
First consider the destination reduction algorithm. Since each iteration corresponds to one vertex, i.e., the common tail node of the last topologically ordered edges, the algorithm terminates in 
$\mathcal{O}(|\mathcal{V}|)$ steps. 
%Moreover, by the construction of the algorithm,
%when it terminates, the following property is satisfied. 

Now consider the recursive coding algorithm. The algorithm traverses through all the intermediate
nodes in their topological order and performs either random coding or
alignment step for outgoing edges. At each node $v$, the complexity of the
coding operations is dominated by the complexity of the alignment step, if
it is performed, which is bounded by $\mathcal{O}(d^3)$, where $d$ is the
in-degree of node $v$. Therefore, we have the following property for the
recursive coding algorithm.

\begin{lemma}
   For the class of directed acyclic graphs with in-degree bounded by $D$, the
   complexity of the recursive coding algorithm is $\mathcal{O}(|\mathcal{V}|
   D^3)$.
\end{lemma}

%From a practical standpoint, our algorithm can be easily used on networks of the order of a few hundred of nodes. 

\section{Alternative proof of the max-flow min-cut theorem}
\label{sec:mfmc}
In this section, we give an alternate proof of the max-flow min-cut theorem. The intuition of our proof is provided in Section \ref{sec:intuition}. Our proof is based on induction on the number of edges of the graph.  We begin with some notation that we will use in this section. 

We denote the single unicast network graph as $\mathcal{G}_n = (\mathcal{V}_n, \mathcal{E}_n)$ with $n$ edges the subscript $n$ is a notation that is useful in our proof. As per the notation introduced in Section \ref{sec:model}, we denote the unicast
problem as $\Omega_n = \left( \mathcal{G}_n, s, T_n \right)$, where $s \in \mathcal{V}_n$ is the source and $T_n \subset
\mathcal{E}_n$ is the set of destination edges. A linear coding solution for $\Omega_n$ can be described by the local
coding matrix $\mathbf{F}_n$, which results in a linear network transfer $\mathbf{M}_n =
\left( \mathbf{I}_n - \mathbf{F}_n \right)^{-1}$. The source to destination linear transfer matrix $\mathbf{H}_n$ can be
obtained as the submatrix of $\mathbf{M}_n$, whose rows correspond to the source edges and whose columns
correspond to the destination edges. Mathematically, it can be done by multiplying $\mathbf{M}_n$ with with an incident
matrix $\mathbf{A}_n$ of size $S \times n$ and then an exit matrix $\mathbf{B}_n$ of size $D \times n$, where $S$ is
out-degree of the source node and $D$ is the cardinality of the destination edge set. Each row of $\mathbf{A}_n$ is a
length $n$ unit vector indicating the corresponding source edge coming out of the source.  Likewise, each row
$\mathbf{B}_n$ is a length $n$ unit vector indicating the index of the corresponding destination edge\footnote{Note the
    difference the incident/exit matrices and the input/output matrices $\mathbf{A}$ and $\mathbf{B}$ defined in
    \cite{Koetter_Medard}. The latter are not restricted to unit vector in rows and encompass the encoding and decoding
    operations at the source and receiver respectively.  Here, we focus on the transfer matrix observed for the network.
Hence, we are concerned with only incident and exit matrices.}. To the end, the
$(s-T_n)$ source to destination linear transfer matrix is given by $\mathbf{H}_n = \mathbf{A}_n \left( \mathbf{I}_n -
\mathbf{F}_n \right)^{-1} \mathbf{B}_n^T$. Our goal is to show that the rank of $\mathbf{H}_{n}$ can be made equal to
the min-cut between the source and the destination by choosing $\mathbf{F}_{n}$ appropriately. We use the following
notation for the min-cut: for any problem $\Omega=(\mathcal{G},s,T)$, we denote the min-cut between the source node $s$
and the destination edges $T$ as $c_{\mathcal{G}}(s,T)$. 

Without loss of generality, we assume $\mathbf{A}_n \cdot \mathbf{B}_n = \mathbf{0}$, that is, the source node $s$ is
not a tail node of any destination edges in $T$ in the graph $\mathcal{G}_n$. If there is any destination edge that is
coming directly from the source node, then this edge can be removed; it suffices to show the max-flow min-cut theorem
can be proved on the remaining graph. Next we introduce a few definitions and lemmas which will be useful in our proof.
\subsection{Preliminary Lemmas}
\begin{definition}[Atomic matrix]
    An atomic matrix of size $n \times n$ is an upper-triangular matrix, where all the off-diagonal elements are zero,
    except those elements in a single column. Given an upper-triangular matrix $\mathbf{U}$, the $i$-th atomic matrix of
    $\mathbf{U}$, denoted as $\mathbf{U}^{[i]}$, is the atomic matrix formed by setting all the off-diagonal elements of
    $\mathbf{U}$ to zero, except those in column $i$.
\end{definition}
A standard property in matrix algebra captures the relation between an upper-triangular matrix and its atomic matrices.% is captured by the lemma below.
\begin{lemma}[Atomic decomposition]
\label{lma:atomicdecomp}
    An $n \times n$ triangular matrix $\mathbf{U}$ with all one diagonal
    elements can be written as the product of its atomic matrices in the reverse
    index order, i.e. $\mathbf{U} = \mathbf{U}^{[n]} \cdot \mathbf{U}^{[n-1]} \cdots
    \mathbf{U}^{[1]}$.
\end{lemma}

Since the local coding matrix $\mathbf{F}_n$ is an $n \times n$ strict upper-triangular matrix, the quantity $\left(
\mathbf{I}_n - \mathbf{F}_n \right)$ is also an upper-triangular matrix but with all diagonal elements begin equal to $1$.  We are
interested in decomposing the inverse of atomic matrix $\left( \mathbf{I}_n - \mathbf{F}_n \right)^{[i]}$ in order to
understand the linear network transfer matrix. For that, we start with the following property.
\begin{property}
    \begin{align}
        \left(\left( \mathbf{I}_n - \mathbf{F}_n \right)^{[i]}\right)^{-1} = \left( \mathbf{I}_n + \mathbf{F}_n
        \right)^{[i]}~.
        \label{eqn:atomicinverse}
    \end{align}
\end{property}

For simplicity, we denote $\mathbf{E}_n = \left( \mathbf{I}_n + \mathbf{F}_n \right)$ and define $i$-th atomic matrix
$\mathbf{E}_n^{[i]}$ to be the $i$-th \emph{edge coding matrix}. Note that the $i$-th column entries of
$\mathbf{E}_n^{[i]}$ above the diagonal represent exactly the local coding vector of the $i$-th edge of the network.
With property~\eqref{eqn:atomicinverse}, we note the following. \begin{lemma}[Network transfer matrix decomposition]
    The network transfer matrix $\left( \mathbf{I}_n -\mathbf{F}_n \right)^{-1}$
    can be decomposed into the product of its edge matrices in the forward
    topological order, i.e. $\left( \mathbf{I}_n-\mathbf{F}_n \right)^{-1} =
    \prod_{i=1}^{n} \mathbf{E}_n^{[i]}$
    \label{lma:networkdecomp}
\end{lemma}
The proof is simple and omitted. The above lemma is interesting because it decomposes the network transfer matrix into
the local coding based into a prodoct of $n$ matrices, where each matrix in the product captures the influence of the
local coding vector corresponding to a single edge. The edge reduction lemma of \cite{Zeng_Cadambe_Medard} can be shown simply using the above decomposition. We now proceed to a proof of the max flow min-cut theorem.

    %$\mathcal{C}$ is an $A-B$ cut set if there exists no directed path from any vertex in $A$ to any vertex in $\Head(B) =
%\left\{ \Head(e), e \in B \right\}$ on the graph $(\mathcal{V}, \mathcal{E}\backslash \mathcal{C})$.  We define function
%$c^{\mathcal{G}}(A, B) \define \min_{\mathcal{C}} \{|\mathcal{C}| : \mathcal{C} \text{ is a } A-B \text{ cut set}\} .$
%When the graph in consideration can be gleaned unambiguously from the context, we drop the superscript $\mathcal{G}$.
%For convenience, we mildly abuse notation when $A,$ and/or $B$ is a singleton; we write $c(\{v\},\{w\})$ as simply
%$c(v;w)$.

\subsection{The proof}
Now we are ready to prove the max flow min cut theorem. It is sufficient to show the linear achievability of a flow that
equals to the min cut. That is equivalent to the following proposition.
%That is, it suffices to prove that there exists a linear scalar coding solution for the unicast problem
%$\Omega_n$ on an arbitrary $n$-edge directed acyclic graph $\mathcal{G}_n$, such that the rank of the source to
%destination transfer matrix $\mathbf{H}_n$ equals to the min-cut between source $s$ and destination set $T_n$. 
\begin{proposition}[Achievability of Min-Cut]
\label{prop:mfmc}
For all $k \in \mathbb{Z}^+$, given an arbitrary directed acyclic graph $\mathcal{G}_k=(\mathcal{V},\mathcal{E})$ with $k$ edges, and a
    unicast problem $\Omega_k = \left( \mathcal{G}_k, s, T \right)$, there exists a $k \times k$ local coding matrix
    $\mathbf{F}_k$, such that the rank of the transfer matrix from $s$ to $T$ equals to the min cut between $s$ and
    $T$, i.e.
\begin{align}
    \rank\left( \mathbf{H} \right) = \rank \left(  \mathbf{A} \left(
    \mathbf{I}_k - \mathbf{F}_k
    \right)^{-1} \mathbf{B}^T \right) =
    c_{\mathcal{G}_k}(s,T)~.
    \label{eqn:rankclaim}
\end{align}
\end{proposition}

Note that this is equivalent of showing that there exists an edge coding matrices
$\mathbf{E}_k^{[i]}$ for the $i$-th edge, such that 
\begin{align}
    \rank\left( \mathbf{H} \right) = \rank\left( \mathbf{A} \left( \mathbf{I}_k - \mathbf{F}_k^{-1}
    \right)\mathbf{B}^T \right) = \rank\left(\mathbf{A} \cdot \prod_{i=1}^n \mathbf{E}^{[i]}_k \cdot
    \mathbf{B}^T\right) = c_{\mathcal{G}_k}(s,T)~.
\end{align}

We prove Proposition~\ref{prop:mfmc} using mathematical induction on $k \in \mathbb{Z}^+$. When $k=1$, the claim is
trivially true.  For the inductive step, we start with the following inductive assumption and claim
\begin{assumption}[Inductive assumption]
    Proposition~\ref{prop:mfmc} holds for all $k$ such that $ 1 \leq k \leq n-1$.
\end{assumption}

Using the above assumption, we prove Proposition \ref{prop:mfmc} for the case where $k=n,$ for an arbitrary, unicast problem $\Omega_n$. It will be convenient to denote $\mathcal{E}_{n} = \mathcal{E}, T_n = T, \mathbf{H}_{n}=\mathbf{H}, \mathbf{A}_{n}=\mathbf{A}, \mathbf{B}_{n}=\mathbf{B}.$ Let $e_n$ be an edge with the highest topological order in graph $\mathcal{G}_n$. Let $\mathcal{G}_{n-1}=(\mathcal{V},\mathcal{E}-\{e_n\})$. In the problem $\Omega_n$, if $e_n$ is not a destination
edge, i.e.  $e_n \not \in T_n$,  then $e_n$ does not communicates with any destination edges, i.e. $e_n \not \leadsto e,
~\forall e \in T$, since $e_n$ has the highest topological order in $\mathcal{G}_n$. Consequently, erasing $e_n$ do not affect the capacity and the min cut from $s$ to $T_n$, i.e.
$c_{\mathcal{G}_n}(s, T_n) = c_{\mathcal{G}_{n-1}}(s, T_n)$. %Therefore, it is sufficient to show the claim 
%for the problem $(\mathcal{G}_{n-1}, s, T_n)$, which is guaranteed by the inductive assumption, since
%$\mathcal{G}_{n-1}$ contains only $n-1$ edges.

Next we focus on the case when $e_n \in T_n$. Note that $e_n$ is not emitted directly from the source node $s$ and is a
destination edge. As a result, we can decompose the incident and exit matrices as follows, 
\begin{align}
    \mathbf{A}_n &= \begin{bmatrix}
        \mathbf{A}_{n-1} & \mathbf{0} 
    \end{bmatrix},   & 
    \mathbf{B}_n &= 
\begin{bmatrix}
    \mathbf{B}^*_{n-1} & \mathbf{0} \\
    \mathbf{0} & 1
\end{bmatrix} ~,
\end{align} 
where $\mathbf{A}_{n-1}$ is the incident matrix from the source to subgraph $\mathcal{G}_{n-1}$, while
$\mathbf{B}^*_{n-1}$ is the exit matrix from $\mathcal{G}_{n-1}$ to the first $D-1$ destination edges. The zero column
following $\mathbf{A}_{n-1}$ indicates that $e_n$ is not emitted from the source node, whereas the unit vector in the
last column of $\mathbf{B}_n$ indicates that $e_n$ is a destination edge.  Subsequently, consider the decomposition of
the $s-T_n$ transfer matrix on $\mathcal{G}_n$. In the following sequence of equations, we use the notation $\mathbf{e}_{n}$ to denote the first $n-1$ entries of the $n$th column of matrix $\mathbf{F}_{n}$; note that the last entry of the $n$th column is $0$. We can write the transfer matrix as,
\begin{align}
    \mathbf{H}_n
&=\mathbf{A}_n \left( \mathbf{I}_n - \mathbf{F}_n \right)^{-1} \mathbf{B}_n^T 
= \mathbf{A}_n \cdot \prod_{i=1}^{n} \mathbf{E}_n^{[i]} \cdot
\mathbf{B}_n^T \nonumber \\
&= 
\begin{bmatrix}
    \mathbf{A}_{n-1} &  \mathbf{0}
\end{bmatrix}\cdot
\prod_{i=1}^{n} \mathbf{E}^{[i]}_i
\cdot
\begin{bmatrix}
    \mathbf{B}^*_{n-1} & \mathbf{0} \\
    \mathbf{0} & 1
\end{bmatrix}^T \\
&=
\begin{bmatrix}
    \mathbf{A}_{n-1} & \mathbf{0} 
\end{bmatrix}\cdot
\left(
\prod_{i=1}^{n-1} \mathbf{E}_n^{[i]} 
\right)
\cdot
\mathbf{E}_n^{[n]}
\cdot
\begin{bmatrix}
    \mathbf{B}^*_{n-1} & \mathbf{0} \\
    \mathbf{0} & 1
\end{bmatrix}^T  \nonumber \\
&\overset{(a)}{=} 
\begin{bmatrix}
    \mathbf{A}_{n-1} & \mathbf{0}
\end{bmatrix}\cdot
\left[
\begin{array}{c|c}
    \left( \mathbf{I}_{n-1} - \mathbf{F}_{n-1} \right)^{-1} & \mathbf{0} \\
    \hline
    \mathbf{0} & 1
\end{array}
  \right]
\cdot
\mathbf{E}_n^{[n]}
\cdot
\begin{bmatrix}
    \mathbf{B}^*_{n-1} & \mathbf{0} \\
    \mathbf{0} & 1
\end{bmatrix}^T  \nonumber \\
&\overset{(b)}{=} 
\begin{bmatrix}
    \mathbf{A}_{n-1} & \mathbf{0}
\end{bmatrix}\cdot
\left[
\begin{array}{c|c}
    \left( \mathbf{I}_{n-1} - \mathbf{F}_{n-1} \right)^{-1} & \mathbf{0} \\
    \hline
    \mathbf{0} & 1
\end{array}
  \right]
\cdot
\left[
\begin{array}{c|c}
    \mathbf{I}_{n-1} & \mathbf{e}_n \\
    \hline
    \mathbf{0} & 1
\end{array}
\right]
\cdot
\begin{bmatrix}
    \mathbf{B}^*_{n-1} & \mathbf{0} \\
    \mathbf{0} & 1
\end{bmatrix}^T  \nonumber \\
&=
\begin{bmatrix}
    \mathbf{A}_{n-1} \left( \mathbf{I}_{n-1} - \mathbf{F}_{n-1} \right)^{-1} & \mathbf{0}
\end{bmatrix}
\begin{bmatrix}
    \mathbf{B}^{*^T}_{n-1} & \mathbf{e}_n \\ 
    \mathbf{0} & 1
\end{bmatrix} \nonumber \\
&=\mathbf{A}_{n-1} \left( \mathbf{I}_{n-1} - \mathbf{F}_{n-1} \right)^{-1}
    \begin{bmatrix}
        \mathbf{B}^{*^T}_{n-1}  & \mathbf{e}_n
    \end{bmatrix} ~.
\end{align}
In this process, for step $(a)$, we apply Lemma~\ref{lma:networkdecomp} on the product of first $n-1$ to obtain the
network transfer matrix of the subgraph $\mathcal{G}_{n-1}$, while in $(b)$ we simply write out the matrix
$\mathbf{E}_n$ explicitly using the local coding vector $\mathbf{e}_n$ on the edge $e_n$ and idenity matrix of
$(n-1)\times(n-1)$. To the end, we decompose the source-destination transfer matrix into two parts that can be
intuitively understood. The first part, $\mathbf{H}_{n-1} = \mathbf{A}_n\left( \mathbf{I}_{n-1}- \mathbf{F}_{n-1}
\right)^{-1} \mathbf{B}^{*{^T}}_{n-1}$ is simply the source to destination transfer matrix of the subgraph
$\mathcal{G}_{n-1}$.  The second part is the contribution of the last topologically ordered edge $e_n$. Therefore,
\begin{align}
    \mathbf{H}_n = 
    \begin{bmatrix}
        \mathbf{H}_{n-1} &\mathbf{A}_{n-1} \left( \mathbf{I}_{n-1} - \mathbf{F}_{n-1} \right)^{-1} \mathbf{e}_n
    \end{bmatrix}.
    \label{eqn:twopartdecomp}
\end{align}

Now consider the last edge $e_n$. If it is not a part of any $s-T_n$ min cut on graph $\mathcal{G}_n$, then removal of
$e_n$ does not affect the min cut of the graph, i.e.  $c_{\mathcal{G}_{n-1}}(s,t) = c_{\mathcal{G}_{n}}(s,t)$. In this
case, leveraging the inductive assumption on the subgraph $\mathcal{G}_{n-1}$, we claim that there exists local coding
matrix $\mathbf{F}^*_{n-1}$ for the graph $\mathcal{G}_{n-1}$ such that, 
\begin{align}
    \rank\left(\mathbf{H}_{n-1} \right)
        = 
    \rank\left( \mathbf{A}_{n-1} \left( \mathbf{I}_{n-1} - \mathbf{F}_{n-1}
    \right)^{-1} \mathbf{B}^{*^T}_{n-1} \right)
    = c_{\mathcal{G}_{n-1}}(s,t) = c_{\mathcal{G}_n}(s,t)~.
    \label{eqn:nonmincutachieve}
\end{align}
To put it simply, when $e_n$ is not a part of any min cut set, we can ignore the
this edge and achieve a flow equal to the min cut, using only the remaining
subgraph $\mathcal{G}_{n-1}$. To do this, we can simply choose the existing
local coding matrix $\mathbf{F}^*_{n-1}$ which satisfies
\eqref{eqn:nonmincutachieve} and set the local coding vector at $\mathbf{e}_n$
to be zero, i.e. $\mathbf{e}_n = \mathbf{0}$. Doing this eliminates the last column 
in \eqref{eqn:twopartdecomp} and guarantees that 
\begin{align}
    \rank\left(\mathbf{H}_n\right) =
    \rank\left( \mathbf{H}_{n-1}\right) = c_{\mathcal{G}_{n-1}}(s,T_n) = c_{\mathcal{G}_{n}}(s, T_n)
\end{align}

It remains to prove the claim in the case when $e_n$ belongs to some min-cut set for the graph $\mathcal{G}_n$. Let $v =
\Tail(e_n)$ be the tail node of theg edge $e_n$. We examine two unicast problems on $\mathcal{G}_{n-1}$, $\Omega_{n-1} =
\left( \mathcal{G}_{n-1}, s, T_{n-1} \right)$ and $\Omega^*_{n-1} = \left( \mathcal{G}_{n-1}, s, T^*_{n-1} \right)$,
where $T_{n-1}$ and $T^*_{n-1}$ are given by, 
\begin{align}
 T_{n-1} &= T_n \backslash \{e_n\} \cup \In(v) \quad, &   T^*_{n-1} &= T_n \backslash \{e_n\} 
    \label{eqn:Tchanges}
\end{align}

Note that both $\Omega_{n-1}$ and $\Omega^*_{n-1}$ are unicast problems whose underlying graph has exactly $n-1$ edges.
For the source-destination min cuts, we have the following lemma.
\begin{lemma}
    When $e_n$ is a edge in some min cut between $s$ and $T_n$ on $\mathcal{G}_n$, the min cut values on $\mathcal{G}_n$
    and $\mathcal{G}_{n-1}$ satisfy the following 
    \begin{align}
        c_{\mathcal{G}_{n-1}}(s, T_{n-1}) &\geq c_{\mathcal{G}_{n}}(s, T_n)  \\
        c_{\mathcal{G}_{n-1}}(s, T^*_{n-1}) &= c_{\mathcal{G}_{n}}(s, T_n) - 1 
        \label{lma:mincuts}
    \end{align}
\end{lemma}
    The proof follows from elementary graph theoretic arguments and is omitted. It is straightforward to see that the $s-T_{n-1}$ transfer matrix for $\Omega^*_{n-1}$ is exactly given by
$\mathbf{H}_{n-1}$, which is the first part of the matrix $\mathbf{H}_n$. For $\Omega_{n-1}$, the destination edge set
is formed by replacing $e_n$ with its parens edges, i.e. incoming edges to $v$. Let these edge be $e_{i_1}, \dots,
e_{i_h}$, where $i_1, \dots, i_h$ are the edge indices,  the exit matrix of $\Omega_{n-1}$ is given by
\begin{align}
    \mathbf{B}^T_{n-1} =
    \left[
        \begin{array}[h]{c|c|c|c}
            \mathbf{B}^{*^T}_{n-1} & \mathbf{u}_{i_1} & \dots &
            \mathbf{u}_{i_h}
        \end{array}
    \right]~,
\end{align}
where $\mathbf{u}_{i_j}, ~ 1 \leq j \leq h$ is a length $n$ unit vector with $1$ at position $i_j$ and zero elsewhere.
Hence, the transfer matrix for $\Omega_{n-1}$ is,
\begin{align}
    \mathbf{H}_{n-1} &= \mathbf{A}_{n-1} \left( \mathbf{I}_{n-1} - \mathbf{F}_{n-1} \right)^{-1} \mathbf{B}^T_{n-1}
    \nonumber \\
    &= 
    \begin{bmatrix}
        \mathbf{H}_{n-1}^T & \mathbf{A}_{n-1} \left( \mathbf{I}_{n-1} - \mathbf{F}_{n-1} \right)^{-1}
        \begin{bmatrix}
            \mathbf{u}_{i_1} & \dots & \mathbf{u}_{i_h}
        \end{bmatrix}
    \end{bmatrix}. \nonumber 
\end{align}

Next, we invoke the inductive assumptions on the problems $\Omega_{n-1}$ and $\Omega^*_{n-1}$. We can conclude that,
there exists a local coding matrix $\mathbf{F}_{n-1}$ for $\Omega_{n-1}$ and a local coding matrix $\mathbf{F}^*_{n-1}$
for $\Omega^*_{n-1}$, such that,
\begin{align}
    \rank\left( \mathbf{H}_{n-1} \right) 
    &=
    \rank\left( \mathbf{A}_{n-1} \left( \mathbf{I}_{n-1} -
    \mathbf{F}_{n-1} \right)^{-1} \mathbf{B}^T_{n-1} \right) =
    c_{\mathcal{G}_{n-1}} (s,T_{n-1}) \geq c_{\mathcal{G}_n} (s,T_n) 
    \label{eqn:Gnksoln} \\
    \rank\left( \mathbf{H}^*_{n-1} \right) 
    &=
    \rank\left( \mathbf{A}_{n-1} \left( \mathbf{I}_{n-1} -
    \mathbf{F}^*_{n-1} \right)^{-1} \mathbf{B}^{*^T}_{n-1} \right) 
    = c_{\mathcal{G}_{n-1}} (s,T^*_{n-1})
    = c_{\mathcal{G}_n} (s,T_n) - 1
    \label{eqn:Gnkssoln} 
\end{align}

Finally, the following lemma helps us to find the local coding matrix for the original graph $\mathcal{G}_n$
\begin{lemma}
\label{lma:solutioncombining}
    Given a local coding matrix $\mathbf{F}_{n-1}$ that
    satisfying~\eqref{eqn:Gnksoln} and a local coding matrix
    $\mathbf{F}^*_{n-1}$ satisfying~\eqref{eqn:Gnkssoln}, for a sufficiently
    large field $\mathbb{F}$, there exists $p, ~ q \in \mathbb{F}$ and a local
    coding vector $\mathbf{e}_{n}$ such that 
\begin{align}
    \rank\left(
\mathbf{A}_{n-1} \left( \mathbf{I}_{n-1} - \left(p\mathbf{F}_{n-1} + q
\mathbf{F}_{n-1}^* \right) \right)^{-1}
    \begin{bmatrix}
        \mathbf{B}^{*^T}_{n-1}   & \mathbf{e}_n
    \end{bmatrix}
\right) = c_{\mathcal{G}_n}(s,T_n) 
\end{align}
\end{lemma}
The proof follows from the spirit of the algebraic framework of network coding introduced in \cite{Koetter_Medard} and is presented in Appendix~\ref{solutioncombining}.
Therefore, from the two local coding matrices $\mathbf{F}_{n}$ and $\mathbf{F}^*_{n-1}$, we can choose the local coding
vectors for the last $k$ edges randomly, and construct the local coding matrix $\mathbf{F}_n$ for $\mathcal{G}_n$ as
follows, 
\begin{align}
    \mathbf{F}_n = 
    \begin{bmatrix}
        p\mathbf{F}_{n-1} + q \mathbf{F}_{n-1}^* & \mathbf{e}_n \\
        \mathbf{0} & \mathbf{0} 
    \end{bmatrix}
    \label{eqn:recursiveF}
\end{align}
The new local coding matrix $\mathbf{F}_n$ will satisfy the original claim in \eqref{eqn:rankclaim} which provides an
achievable flow that equals the min cut on any arbitray graph of $n$ edges. 

\begin{remark}
Note that in the process of searching for the local coding matrix $\mathbf{F}_n$ which satisfies 
\eqref{eqn:rankclaim} and achieves the maximum flow, we essentially rely on recursive construction the local coding
matrices throught the sequence $\{\mathbf{F}_{n-1}, \mathbf{F}_{n-2}, \dots, \mathbf{F}_1\}$
using~\eqref{eqn:recursiveF}. These local coding matrices in turn correspond to maximum flow achieving coding matrices
for smaller problems we constructed as subproblems, i.e. $\left\{ \Omega_{n-1}, \Omega_{n-2},\dots \Omega_1 \right\}$.
This sequence of problems would, in fact, result in running the destination reduction algorithm of Section \ref{sec:algorithm} on the original problem $\Omega$ (with the second destination set equal to the null set). 
\end{remark}

\section{Achievability of Rate Pair $\left( 1, 1 \right)$}
\label{sec:one-one}
%We prove for the special case that the rate pair $(1,1)$ is always achievable through our algorithm unless there is a
%single edge GNS cut. An implication of our results is when there is a one-edge GNS cut in the two-unicast-Z network, our
%algorithm is optimal.
We prove for the special case that the rate pair $(1,1)$ is always achievable through our algorithm unless there is a
single edge GNS cut. We begin by stating our result formally.

\begin{theorem}
\label{thm:oneone}
Consider a two-unicast-Z problem $\Omega = (\mathcal{G}, \mathcal{S}, \mathcal{T})$ where, for $i \in \{1,2\}$, source $s_i$ is connected to at least one edge in destination $T_i$, and the GNS-cut set bound is at least $2$. Then, 
the coding matrix $\mathbf{F}$ returned by $\textsc{RecursiveCoding}\left( \textsc{Reduction}\left(
\Omega \right), \mathbb{F} \right)$ achieves the rate pair $(1,1)$ with a probability that tends to $1$ as the field size $\mathbb{F}$ increases.
\end{theorem}

We first state the a few lemmas on the transfer matrices produced by the recursive coding algorithm. They are useful for
the proof of Theorem~\ref{thm:oneone}. 

\begin{lemma}
    \label{lma:grank1properties}
    Let $\mH{1}, \mH{2}$ and $\mG{2}$ be matrices of dimensions $P_1 \times Q_1$, $P_2 \times Q_1$ and $P_2 \times Q_1$
    respecitvely. If $\mH{1} \neq \mathbf{0}$ and $\mG{2} \neq \mathbf{0}$, then $\Grank \left( \mH{1}, \mH{2}, \mG{2}
    \right) = 1$ if and only if the following holds,
    \begin{align*}
        \rank \left( \mH{2} \right) = 
        \rank\left( 
        \begin{bmatrix}
            \mH{1} \\ \mH{2}
        \end{bmatrix}
        \right) 
        = \rank \left( \mG{2} \right)
        = \rank \left( 
        \begin{bmatrix}
            \mH{2}  & \mG{2}
        \end{bmatrix}
        \right) = 1.
    \end{align*}
    %In particular, $\Span\left( \mH{2} \right) = \Span \left( \mG{2} \right)$.
\end{lemma}
\begin{lemma}
    \label{lma:emptycols}
    {
    In the recursive coding algorithm,  for any $i \in \{0,1,2,\ldots,N\},$ the matrix $\HTall$ does not contain an all zeroes column with a probability that tends to $1$ as the size of field $|\mathbb{F}|$ increases.
}
\end{lemma}
\begin{lemma}
\label{lma:emptyG2}
 Let $Q^{(i)} = T_2^{(i)} \backslash T_1^{(i)}$. Given that the recursive algorithm performs the alignment step
 from some stage $k+1$ to stage $k$, then, with a probability that tends to $1$ as the field size $|\mathbb{F}|$ increases,
 $\mathbf{G}_2^{Q^{(i)}} \neq \mathbf{0}$ for every stage $i$ in $\{0,1,2\ldots, k\}$.
\end{lemma}
It is straightforward to verify Lemma~\ref{lma:grank1properties} from the definition of Grank. The proofs of Lemmas
\ref{lma:emptycols} and \ref{lma:emptyG2} are provided in Appendix \ref{app:emtyG2proof} and \ref{app:emtycolsproof}
respectively.

\begin{remark}
	Henceforth, all the statements of the proof hold true in a probabilistic sense. That is, the statements are true with a probability that tend to $1$ as the field size $|\mathbb{F}|$ tends to infinity. To avoid laborious notation and wording, we omit mentioning this explicitly in our proof.
\end{remark}

\subsection{Proof Overview}
We first provide a brief overview of the proof of Theorem~\ref{thm:oneone} and summarize the basic ideas of the proof.
Note that we only prove the necessity part of the theorem in this paper, as the sufficiency follows directly from the
GNS outer bound~\cite{Kamath_Tse_Ananthram}. In particular, it suffices to show that the recursive coding algorithm will
generate source-to-destination transfer matrices which give an achievable rate region that contains the point $(1,1)$.
In the proof, we show that this is true for an arbitrary two-unicast-Z problem $\Omega = (\mathcal{G}, \mathcal{S},
\mathcal{T})$ with underlying graph $\mathcal{G}$. We shall use the achievable region given in Theorem $1$ in
\cite{Zeng_Cadambe_Medard}, which is simplified to the following in the case of two-unicast-Z networks,
\begin{align}
    R_1  &\leq \rank \left(\HT{1}{0} \right)~, \label{eqn:achievableregion1}\\
    R_2  &\leq \rank \left(\GT{2}{0} \right)~, \label{eqn:achievableregion2}\\
    R_1 + R_2 &\leq \Granki{0}. \label{eqn:achievableregion3}
\end{align}
In order to show that the above achievable region contains the point $(1,1)$, we prove three claims on the transfer
matrices, each corresponds to an upper bound in \eqref{eqn:achievableregion1} to \eqref{eqn:achievableregion3}. 
Recall that the algorithm is divided into stages from $N$ to $0$. At stage $i$, the algorithm specifies coding
co-efficients for the edges in the destination sets of $\Omega^{(i)}$, specifically, for edges in $T_1^{(i)} \cup
T_2^{(i)} \backslash \left(T_1^{(i+1)} \cup T_2^{(i+1)}  \right)$, by using the coding co-efficients in the previous
stage $i+1$. Also, recall that in the process of coding for stage $i$, for destination $T_j^{(i)}$, $U_j^{(i)}$ denotes
the unchanged edges; $I_j^{(i)}$ denotes the incoming edges which are coded in stage $i+1$ and are removed from the
destination set; $O_j^{(i)}$ denotes the outgoing edges which are to be coded from $I_j^{(i)}$ edges and enter the
destination set. Furthermore, $O_2^{(i)}$ is the union of two disjoint set $B_2^{(i)}$ and $A_2^{(i)}$. The former set
of edges are coded in phase $1$ of the algorithm while the latter edges are coded in phase $2$ as a subset of
$O_1^{(i)}$. To prove Theorem \ref{thm:oneone}, we show that
\begin{align}
    \rank \left( \HT{1}{i} \right) &\geq 1 ~, &  \rank \left(  \GT{2}{i} \right) &\geq 1~, & \Granki{i} &\geq 2~,  
    \label{eqn:matrixrankbounds}
\end{align}
for every stage $i$, $0 \leq i \leq N$, unless there is a single edge GNS cut in the graph. Note that the above
statement is trivially true for the initial stage, i.e. $i=N$. Therefore, to prove the theorem, we assume that
\eqref{eqn:matrixrankbounds} holds true for some stage $i+1$ and we show that the conditions remain true in stage $i$
unless there is a single edge GNS cut set. We outline our proof steps below.
}
\begin{itemize}
    \item In Claims \ref{claim1} and \ref{claim2}, we show that the transfer matrices $\HT{1}{i}$ and $\GT{2}{i}$ are
        non-zero, assuming that $\HT{1}{i+1},\GT{2}{i+1}$ are non-zero. We prove this claim by considering the the cases
        of randomization and alignment steps separately. Since in the previous stage, the matrices $\HT{1}{i+1}$ and
        $\GT{2}{i+1}$ are non-zero by assumption, and since random linear combination of a set of non-zero vector does
        not result in a zero vector, a randomization step preserves the desired non-zero property
        of the respective matrices at stage $i$. In case that the stage contains an alignment step, we make use of the
        fact that conditions \eqref{eqn:con0}-\eqref{eqn:con2} have to be fulfilled and the generated column has to
        satsify Lemma~\ref{lem:alignment}. We verify that these conditions imply that neither $\GT{2}{i}$ nor
        $\HT{1}{i}$ are set to 0.
    \item We show in Claim \ref{claim3} that in stage $i$, the Grank of the transfer matrices is always lower bounded by
        $2$ only if there does not exist a single edge whose removal disconnects $(s_1, T_1), (s_2, T_2)$ and $(s_2,
        T_1)$. Specifically, we assume that the claim holds for stage $i+1$ and we show that it remains true for the
        next stage $i$. To do that, we first show that if any alignment step takes place between the two stages, the
        rank of the resulting matrix $\HTall$ is strictly greater than the rank of $\HT{2}{i}$. This implies that the
        Grank is strictly larger than the rank of $\left[\mathbf{H}_{2}^{(i)}~~\mathbf{G}_{2}^{(i)}\right],$ which is at
        least 1. As a consequence, we show that the Grank is at least $2$. 
        
        Subsequently, if the algorithm results in Grank less than $2$, it can only carry out randomization steps at that
        stage. In this case, we use Lemmas~\ref{lma:grank1properties}, \ref{lma:emptycols} and \ref{lma:emptyG2} to show
        that if the Grank at stage $i$ reduces to $1$, then $(a)$ the edges in $U_2^{(i)}$ cannot be communicated from
        source $s_2$ $(b)$  the sets $U_1^{(i)}$ and $B_2^{(i)}$ are empty, and $(c)$ there can be only one newly coded
        edge, common to both destination set, i.e. $\left|O_1^{(i)}\right| = \left|A_2^{(i)}\right|=1$. $(a), (b),$ and $(c)$ imply that
        $T_2^{(i)}$ contains only one edge that $s_2$ communicates with, while $T_1^{(i)}$ contains this particular edge
        as the only edge in the set. 
        {
        Consequently, applying Property (v) in Section \ref{sec:algorithm}, we  
        know that $O_1^{(i)}$ is a cut set between source $s_1, s_2$ and $T_1$. Furthermore, because of $(a)$ and the
        fact that $T_2^{(i)}$ is a cut set between source $s_1, s_2$ and $T_2$, we establish that $A_2^{(i)} = T_2^{(i)}
        \backslash U_2^{(i)}$ is a cut set between $s_2$ and $T_2$. 
        }
        As a result, we infer that the removal of the edge
        in $O_1^{(i)}$ simultaneously disconnects the source-destination pairs $(s_1, T_1), (s_2, T_2)$ and $(s_2,
        T_1)$. In other words, we infer that the edge in $O_1^{(i)}$ forms a GNS-cut set to complete the proof.
\end{itemize}
We now formally present proofs of Claims \ref{claim1},\ref{claim2} and \ref{claim3} which combine to serve as a proof of
Theorem~\ref{thm:oneone}.

\subsection{Proof of Theorem~\ref{thm:oneone}}
\begin{claim}
    For all $0 \leq i \leq N$, $\mathbf{H}_1^{T^{(i)}_1} \neq \mathbf{0}$.
    \label{claim1}
\end{claim}

\begin{proof}
Consider first the case of $i = N$. Since $s_1$ communicates with $t_1$, there exists some source edge $(s_1, v_1) \in
T^{(N)}_1$.  Recall that in the trivial solution for $\Omega^{(N)}$, we simply set the column corresponding the each
edge in $T^{(N)}_1$ to a unit vector indicating the source edge it represents. Since $(s_1, v_1) \in T^{(N)}_1 $ is a
source edge at $s_1$, it gives rise to a non-zero column in $\mathbf{H}^{T_1^{(N)}}_1$. Thus, the claim holds for stage
$N$.

Suppose the claim of $\mathbf{H}^{T_1^{(i+1)}}_1 \neq \mathbf{0}$ is true at some stage $i+1$, consider the stage $i$. Note
that only phase $2$ of the algorithm affects $\mathbf{H}^{T^{(i)}_1}_1$. If $\mathbf{H}_1^{U_1^{(i)}} \neq \mathbf{0}$, since $U_1^{(i)}
\subseteq T_1^{(i)}$, the claim holds for $\HT{1}{i}$. Otherwise, we must have $\mathbf{H}_1^{I_1^{(i)}} \neq
\mathbf{0}$. In this case, consider some step $j$ in phase $2$. 
\begin{itemize}
    \item If it is a randomization step, then $\HT{1}{i} \neq \mathbf{0}$. 
    \item If it is an alignment step, by Lemma~\ref{lem:alignment}, the newly generated column must satisfy,  
        %$\mathbf{H}_1^{I_1^{(i)}} \mathbf{F}^{I_1^{(i)}, \left\{ e_{i, j} \right\}}$ 
        \begin{align}
            \mathbf{H}_2^{I_1^{(i)}} \mathbf{F}^{I_1^{(i)}, \left\{ e_{i, j} \right\}}
            &\in
            \Span\left(  
            \begin{bmatrix}
                \mathbf{H}_2^{U^{(i)}_1} & \mathbf{H}_2^{O_{1,j-1}^{(i)}}
            \end{bmatrix}
            \right),
            \nonumber\\
            \begin{bmatrix}
            \mathbf{H}_1^{I^{(i)}_1} \\
            \mathbf{H}_2^{I^{(i)}_1}   
            \end{bmatrix}
            \mathbf{F}^{I_1^{(i)}, \left\{ e_{i, j} \right\}}
                &\not \in
            \Span \left(
            \begin{bmatrix}
                \mathbf{H}_1^{U^{(i)}_1} & \mathbf{H}_1^{O_{1,j-1}^{(i)}} \\ 
                \mathbf{H}_2^{U^{(i)}_1} & \mathbf{H}_2^{O_{1,j-1}^{(i)}} 
            \end{bmatrix}
            \right)\nonumber~.
        \end{align}
        This is only true when $\mathbf{H}_1^{I_1^{(i)}} \mathbf{F}^{I_1^{(i)}, \left\{ e_{i, j} \right\}} \neq
        \mathbf{0}$.  By the construction of the algorithm, $\mathbf{H}_1^{I_1^{(i)}} \mathbf{F}^{I_1^{(i)}, \left\{
            e_{i, j} \right\}}$ is a column in $\mathbf{H}_1^{T_1^{(i)}}$. 
\end{itemize}
Therefore, $\mathbf{H}_1^{T^{(i)}_1} \neq 0$. This completes the proof.
\end{proof}

\begin{claim}
    For all $0 \leq i \leq N$, $\GT{2}{i} \neq \mathbf{0}$. 
    \label{claim2}
\end{claim}
\begin{proof}
For $i=N$, since $s_2$ communicates with $t_2$, there exists a source edge $(s_2, v_2) \in T_2^{(N)}$.  Similar to the
previous case, in the trivial solution for $\Omega^{(N)}$, we simply set the column corresponding the each edge in
$T^{(N)}_2$ to a unit vector indicating the source edge it represents. Since $(s_2, v_2) \in T^{(N)}_2 $ is a source
edge at $s_2$, it gives rise to a non-zero column in $\GT{2}{N}$. Thus, the claim holds for stage $N$.

Assume the claim is true for stage $i+1$ and consider stage $i$. If $\mathbf{G}_2^{U_2^{(i)}} \neq \mathbf{0}$, then the
claim holds for $\GT{2}{i}$, as $U_2^{(i)} \subseteq T_2^{(i)}$. Otherwise, we must have $\mathbf{G}_2^{I_2^{(i)}}\neq
\mathbf{0}$. In this case, if random linear coding is performed, either in phase 1 (for $B_2^{(i)}$ edges) or for an edge in 
$A_2^{(i)}$ in phase 2, then the resulting column is non-zero and hence
$\mathbf{G}_2^{T_2^{(i)}} \neq \mathbf{0}$.

Therefore, we only need to consider the case where alignment occurs in some step $j$ in phase $2$. In this case, because
of the conditions under which the algorithm performs alignment, we have
\begin{align*}
\Span\left(\mathbf{H}_2^{I^{(i)}_1}\right) &\not\subset
\Span\left(
\begin{bmatrix}
    \mathbf{H}_2^{U^{(i)}_1} & \mathbf{H}_2^{O_{1,j-1}^{(i)}}
\end{bmatrix}
\right) \\
\Span\left(\mathbf{H}_2^{I^{(i)}_1}\right) &\subseteq
\Span\left(
\begin{bmatrix}
    \mathbf{H}_2^{U^{(i)}_1} & \mathbf{H}_2^{O_{1,j-1}^{(i)}} & \mathbf{G}_2^{U^{(i)}_2} &
    \mathbf{G}_2^{A^{(i)}_{2,j-1}} & \mathbf{G}_2^{B^{(i)}_2}
\end{bmatrix}
\right)~.
\end{align*}
Since $A_{2,j-1}^{(i)} \subseteq O_{1,j-1}^{(i)}$ by the construction of the algorithm, $\GAI{2}{j-1}$ is a submatrix of
$\HOI{2}{j-1}$. Subsequently, for the above alignment conditions to hold we must have
\begin{align}
   \begin{bmatrix}
     \mathbf{G}_2^{U^{(i)}_2} & \mathbf{G}_2^{B^{(i)}_2}
   \end{bmatrix} &\neq \mathbf{0}.
\end{align}
Since $\begin{bmatrix} \mathbf{G}_2^{U^{(i)}_2} & \mathbf{G}_2^{B^{(i)}_2} \end{bmatrix}$ is a submatrix of
$\mathbf{G}_2^{T_2^{(i)}}$, the alignment step does not set $\mathbf{G}_2^{T_2^{(i)}}$ to $\mathbf{0}$ and the claim is
true for all $i, 0\leq i \leq N$.
\end{proof}

From the two claims, we have established that for all $i$, 
\begin{align}
    \rank\left(\mathbf{H}_1^{T_1^{(i)}}\right) &\geq 1~, & 
    \rank\left(\mathbf{G}_2^{T_2^{(i)}}\right) &\geq 1~. 
    \label{eqn:h1g2lowerbound}
\end{align}
\begin{claim}
    For all $0 \leq i \leq N$, $\Granki{i} \geq 2 $ if and only if the graph $\mathcal{G}$ does not contain a single
    edge GNS cut set.
    \label{claim3}
\end{claim}
\begin{proof}
	First note that Claim \ref{claim1} implies that $\Granki{i} \geq \rank \left( \GT{2}{i} \right) \geq 1$. 
Observe that \eqref{eqn:h1g2lowerbound}, $\Granki{i} = 1$ if and only if the following hold:
\begin{align}
\rank\left(\GT{2}{i}\right) 
= \rank \left(\HGall\right)
= \rank \left(\HT{2}{i}\right)
= \rank \left(\HTall\right) = 1.
    \label{eqn:grank1iff}
\end{align}

Next, we prove Claim~\ref{claim3}. First we show that the claim is true for stage $N$. 

Suppose that $\Granki{N} = 1$. By \eqref{eqn:grank1iff}, 
\begin{align*}
\rank \left( \GT{2}{N} \right)  = &1 = \rank \left( \HT{2}{N} \right)~, &
\Span \left( \HT{2}{N} \right) &= \Span \left( \GT{2}{N} \right)~.
\end{align*}
By the construction of the algorithm, at stage $N$, we assign a different unit vector for each of the different
source edge. Therefore, to satisfy the above equations, we must have $T_1^{(N)} = T_2^{(N)}$ and $\left|T_1^{(N)}\right| = 1$. In
this case, removal the single edge in $T_1^{(N)}$ will disconnect $(s_1, T_1), (s_2, T_2)$ and $(s_2, T_1)$
simultaneously.  Therefore, the claim holds for stage $N$. 

Suppose that the claim holds for stage $i+1$, i.e. $\Granki{i+1} \geq 2$. We show the claim for stage $i$. 
%We show that $\Granki{i} \geq 2$ if and only if there does exists an edge whose removal disconnects $(s_1, t_1), (s_2,
%t_2)$ and $(s_2, t_1)$. We divide the inductive step into two parts, considering whether the coding algorithm performs
%alignment step between stage $i+1$ and $i$.

\subsubsection*{Alignment steps do not violate Claim 3.} 

We first show that if the recursive coding algorithm performs some alignment step in Phase $2$ in stage $i$, then
$\Granki{i} \geq 2$ and therefore, the claim holds. Without loss of generality, assume that the alignment step is
performed for some edge $e_{i,j} \in O_{1}^{(i)}$. Then, by Lemma~\ref{lem:alignment}, the generated column $\HIall
\FIj{1}{j}$ will satisfy,
\begin{align*}
    \HI{2} \FIj{1}{j} &\in 
    \Span \left( 
    \begin{bmatrix}
        \HU{2} & \HOI{2}{j-1}
    \end{bmatrix}
    \right) ~, \\
    \HIall \FIj{1}{j} &\not \in 
    \Span \left( 
    \begin{bmatrix}
        \HU{1} & \HOI{1}{j-1} \\
        \HU{2} & \HOI{2}{j-1}
    \end{bmatrix}
    \right)~.
\end{align*}
As a result,
\begin{align}
    \rank \left( 
     \begin{bmatrix}
        \HU{1} & \HOI{1}{j} \\
        \HU{2} & \HOI{2}{j}
    \end{bmatrix}
    \right)
    > \rank \left( 
    \begin{bmatrix}
        \HU{2} & \HOI{2}{j} 
    \end{bmatrix}
    \right)~.
\end{align}
Therefore, at stage $i$ we have,
\begin{align*}
    & \quad \rank \left( \HTall \right) - \rank \left( \HT{2}{i} \right)  \\
    &=  
    \rank \left( 
     \begin{bmatrix}
         \HU{1} & \HOI{1}{j} & \mathbf{H}_1^{O_1^{(i)} \backslash O^{(i)}_{1,j}} \\
         \HU{2} & \HOI{2}{j} & \mathbf{H}_2^{O_1^{(i)} \backslash O^{(i)}_{1,j}}
    \end{bmatrix}
    \right)
-
    \rank \left( 
     \begin{bmatrix}
         \HU{2} & \HOI{2}{j} & \mathbf{H}_2^{O_1^{(i)} \backslash O^{(i)}_{1,j}}
    \end{bmatrix}
    \right) \\
    &\geq 
    \rank \left( 
     \begin{bmatrix}
        \HU{1} & \HOI{1}{j} \\
        \HU{2} & \HOI{2}{j}
    \end{bmatrix}
    \right)
    - \rank \left( 
    \begin{bmatrix}
        \HU{2} & \HOI{2}{j} 
    \end{bmatrix}
    \right) > 0 ~. 
\end{align*}
Hence $\rank \left( \HTall \right) > \rank \left( \HT{2}{i} \right)$. Since \eqref{eqn:h1g2lowerbound} holds for all
stages, by Lemma~\ref{lma:grank1properties}, we have $$\Granki{i} \geq 2.$$ That implies, if any alignment
step takes place at stage $i$, then the claim holds for stage $i$. 

\subsubsection*{Randomization steps do not violate Claim 3.} 
It remains to show the claim holds if all the coding steps performed at stage $i$ are randomization steps. Consider the
stage $i+1$ where the claim is assumed to be true. That is, we have $$\Granki{i+1} \geq 2.$$ We consider two possible
cases: 
%ivide the inductive assumptions two possible cases based on the rank of $\GT{2}{i}$, i.e.
\begin{enumerate}[(I)]
    \item $\rank \left( \GT{2}{i+1} \right) \geq 2$, and   
    \item $\Granki{i+1} > \rank \left( \GT{2}{i+1} \right) = 1$.
\end{enumerate}
In each case, we show that if 
\begin{equation}
	\Grank\left(\mathbf{H}_{1}^{T_1^{(i)}},\mathbf{H}_{2}^{T_1^{(i)}},\mathbf{G}_{2}^{T_2^{(i)}}\right)= 1
	\label{eq:grankisone}
\end{equation}
then, the following hold: 
\begin{align}
    \GU{2} &= \mathbf{0}~, & U_1^{(i)} &= B_2^{(i)} = \varnothing~, & O_1^{(i)} &= A_2^{(i)}~, & \left|O_1^{(i)}\right| &= 1~.
    \label{eqn:oneedgecutcondition}
\end{align}
Once we show this, the claim can be proved as follows.
Given that $U_1^{(i)}$ is empty at stage $i$, clearly $T_1^{(i)} = O_1^{(i)}$. 
{
Therefore, based on Property (v) of the
algorithm, we infer that $s_1$ and $s_2$ communicate
with $T_1$ only through the single edge in $O_1^{(i)}$. 
}
In other words, removing the edge in $O_1^{(i)}$ disconnects
$T_1$ from the two sources. We argue here that (\ref{eqn:oneedgecutcondition}) implies that removing the edge
disconnects $s_2$ from $T_2$ as well. Since $B_2^{(i)}$ is empty, we have $T_2^{(i)} \backslash T_1^{(i)} = U_2^{(i)}$.
Because $\GU{2} = \mathbf{0}$, we infer from Lemma~\ref{lma:emptyG2} that there is no alignment step in any
stage $k \geq i$. Therefore, the matrix $\GU{2}$ is generated through performing random linear coding at every edge that
lies between the sources and the edges in $T_1^{(i)}\cup T_2^{(i)}$.  The fact that $\GU{2}$ is a zero matrix in
conjunction with the classical result of \cite{Koetter_Medard} which indicates that random linear network coding
achieves a rank that is equal to the min-cut to any set of destination edges together imply that $s_2$ does not communicate with
$U_2^{(i)}$. 
{
Based on Property (v) in Section \ref{sec:algorithm}, $T_2^{(i)}$ is a cut set between $s_1, s_2$ and $T_2$. But since $s_2$ does not
communicate with $U_2^{(i)}$,
we conclude that $s_2$ communicates with $T_2$ only through $T_2^{(i)} \backslash U_2^{(i)}
= A_2^{(i)} = O_1^{(i)}$, 
}
i.e., $s_2$ communicates with $T_2$ through the single edge in $O_1^{(i)}$. Consequently, the
removal of the single edge in $O_1^{(i)}$ disconnects $s_2$ from $T_2$ as well. Thus, $O_1^{(i)}$ is a single-edge GNS
cut set.

To complete the proof we show that if (\ref{eq:grankisone}) holds, then \eqref{eqn:oneedgecutcondition} also holds for each of the two cases.

\noindent\textbf{Case (I):} $\rank \left( \GT{2}{i+1} \right) = \rank \left( \left[\GU{2} ~~ \GI{2}\right] \right) \geq 2$.
{
Suppose that $\Granki{i} = 1$. Clearly, in this case, $A_2^{(i)}$ and $B_2^{(i)}$ cannot be both empty; otherwise
$\GT{2}{i} = \GT{2}{i+1}$ and $\rank \left( \GT{2}{i} \right) = 2$, which contradicts \eqref{eqn:grank1iff}.
Now, consider the possible ranks of the matrix $\GU{2}$. From (\ref{eqn:grank1iff}) we have, $\rank \left( \GU{2} \right)
\leq \rank \left( \GT{2}{i} \right) = 1$. 
}
We first show that the rank of $\GU{2}$ is 0. If $\rank\left(\GU{2}\right) = 1$, then $\Span \left( \GI{2} \right) \not \subset \Span \left( \GU{2}
        \right)$, as $\rank \left( \GT{2}{i} \right) \geq 2$. If $A_2^{(i)}\cup B_2^{(i)}$ is
        non-empty, then $\GA{2}$ and/or $\GB{2}$ will be created from random coding of columns of $\GI{2}$.  As a
        result, we have, 
        \begin{align*}
            \rank \left(\GT{2}{i}\right) = 
            \rank \left( 
            \begin{bmatrix}
                \GU{2} & \GA{2} & \GB{2}
            \end{bmatrix}
            \right) > \rank \left( \GU{2} \right) = 1, 
        \end{align*}
        which contradicts \eqref{eqn:grank1iff}. But $A_2^{(i)}$ or $B_2^{(i)}$ cannot both be empty. Hence, we
        cannot have $\rank \left( \GU{2} \right) = 1$.
        %However, if both $A_2^{(i)}$ or $B_2^{(i)}$ are empty, then
        %$\GT{2}{i} = \GT{2}{i+1}$ and $\rank \left( \GT{2}{i} \right) = 2$, which again contradicts Equation 
        %\eqref{eqn:grank1iff}. Hence, if $\rank \left( \GU{2} \right) = 1$, then \eqref{eqn:grank1iff} cannot hold. 
    Therefore, $\rank\left(\GU{2}\right) = 0$. Therefore, we have $\rank \left( \GI{2} \right) \geq 2$. Since
    $I_{2}^{(i)}$ is non-empty, note that columns of $ \begin{bmatrix} \HO{2} & \GA{2} & \GB{2} \end{bmatrix}$ are
    generated by random coding from columns of $\GI{2}$. If the set $O_1^{(i)} \cup A_2^{(i)} \cup B_2^{(i)} = O_1^{(i)}
    \cup B_2^{(i)}$ (since $A_2^{(i)} \subseteq O_1^{(i)}$) contains at least two elements, then
        \begin{align*}
            \rank \left( \HGall \right) \geq \rank \left( 
            \begin{bmatrix}
                \HO{2} & \GA{2} & \GB{2}
            \end{bmatrix}
            \right)  \geq 2~,
        \end{align*}
        which contradicts \eqref{eqn:grank1iff}. We have already shown that $A_2^{(i)} \cup B_2^{(i)}$
	is non-empty. Thus, we conclude ${O}_{1}^{(i)} \cup {B}_{2}^{(i)}=1.$ Since ${O}_{1}^{(i)}$ and ${B}_{2}^{(i)}$ are
    disjoint, there are only two possiblities: $\left|{O}_{1}^{(i)}\right|=1, {A}_{2}^{(i)}=O_1^{(i)}, {B}_{2}^{(i)}=0$ or  ${O}_{1}^{(i)}=0, {B}_{2}^{(i)}=1.$ We show next that the latter is impossible, and show that the former condition implies that $U_1^{(i)}$ is empty.
        %implies a single edge removal that disconnects $(s_1, T_1), (s_2, T_2)$ and $(s_2, T_1)$ simultaneously.
        \begin{enumerate}
            \item $\left|B_2^{(i)}\right| = 1$, $\left|O_1^{(i)}\right| = \left|A_2^{(i)}\right| = 0$. From \eqref{eqn:grank1iff}, we have $\rank \left(
                \HT{2}{i} \right) = 1$. Since $\rank \left( \GI{2} \right) = 2$, $\Span \left( \GI{2} \right) \not
                \subset \Span \left( \HT{2}{i} \right)$. As $\GB{2}$ is generated from $\GI{2}$ by random coding, we
                have $\Span \left( \GB{2} \right) \not \subset \Span \left( \HT{2}{i} \right)$. Consequently, 
                \begin{align*}
                    \rank \left( 
                    \HGall
                    \right) \geq 
                    \rank \left( 
                    \begin{bmatrix}
                        \HT{2}{i} & \GB{2}
                    \end{bmatrix}
                    \right) \geq 2~, 
                \end{align*}
                which contradicts \eqref{eqn:grank1iff}.
            \item $\left|B_2^{(i)}\right| = 0$, $\left|O_1^{(i)}\right| = \left|A_2^{(i)}\right| = 1$. Since we have already shown that $\rank\left(\GU{2}\right) = 0,$
                to show \eqref{eqn:oneedgecutcondition}, it remains to prove that $U_1^{(i)}$ is empty.
                We do this by showing the column spaces of both $\HU{2}$ and $\HU{1}$ are the null
		space. By the construction of the algorithm, $I_1^{(i)} = I_2^{(i)};$ since $B_2^{(i)} = \phi$, we have  $\HI{2} = \GI{2}$. From \eqref{eqn:grank1iff}, we have,
                \begin{align*}
                \rank \left( \HU{2} \right) \leq \rank \left( \HT{2}{i} \right) = 1.
                \end{align*}
               Since $\HO{2}$ is generated by random coding from $\GI{2}$ whose rank is $2$, it is clear that $\Span
	       \left( \HO{2} \right) \not \subset \Span \left(\HU{2}\right).$ As a result, if
               $\HU{2} \neq \mathbf{0}$ then
                    \begin{align*}
                        \rank \left( \HT{2}{i} \right) = \rank \left( 
                        \begin{bmatrix}
                            \HU{2} & \HO{2}
                        \end{bmatrix}
                        \right)  \geq 2~,
                    \end{align*}
                    which contradicts \eqref{eqn:grank1iff}. Hence $\HU{2} = \mathbf{0}$ and $\rank \left( \HO{2}
                    \right) = 1$. Now, if $\HU{1} \neq \mathbf{0}$, we have,
                    \begin{align*}
                        \rank \left(
                        \HUOall
                        \right) 
                        \geq \rank \left( 
                        \HU{1}
                        \right) + \rank \left( 
                        \HO{2} 
                        \right) \geq 2~, 
                    \end{align*}
                    which also contradicts \eqref{eqn:grank1iff}. Thus, $\HUall = \mathbf{0}$. But by
                    Lemma~\ref{lma:emptycols}, $\HTall$ does not contain the all zeroes column. That implies that $U_1^{(i)}
                    = \varnothing$ and subsequently, $T_1^{(i)} = O_1^{(i)}$. Therefore, we have shown  
                    \eqref{eqn:oneedgecutcondition} as required.

                    %Next, consider the matrix $\GT{2}{i}$ and the set $T_2^{(i)}$. Recall that we have established so
                    %far that $\rank \left( \GU{2} \right) = 0$ and $|B_2^{(i)}| = 0$. Also, since $A_2^{(i)}$ is
                    %a subset of $O_1^{(i)}$ and $|A_2^{(i)}| = |O_1^{(i)}| =1$, we have  $A_2^{(i)} = O_1^{(i)}$ and
                    %$\GA{2} = \HO{2}$. If $U_2^{(i)}$ is non-empty, 
                    %%we show that no alignment step has taken place prior to stage
                    %%$i$ by observing 
                    %then the matrix $\GT{2}{i}$ is given by 
                    %\begin{align*}
                        %\GT{2}{i} = 
                        %\begin{bmatrix}
                            %\mathbf{0} & \GA{2}
                        %\end{bmatrix}
                        %= \begin{bmatrix}
                            %\mathbf{0} & \HO{2}
                        %\end{bmatrix}~.
                    %\end{align*}
                    %Clearly, all the columns in $\GT{2}{i}$ corresponding to the edges in $T_2^{(i)}\backslash T_1^{(i)}$
                    %that are zero. This would contradicts Lemma~\ref{lma:emptyG2} if any alignment happens before
                    %stage $i$. Therefore all the coding operations the algorithm performed until stage $i$ are random
                    %codings. In particular, since $\GU{2} = \mathbf{0}$ with all random coding, this means that $s_2$
                    %does not communicate with any edges in $U_2^{(i)}$. As a results, regardless of whether $U_2^{(i)}$
                    %is empty, $s_2$ communicates with $T_2$ only through edges in $T_2^{(i)} \backslash U_2^{(i)} =
                    %A_2^{(i)} = O_1^{(i)}$, while at the same time $s_1$ and $s_2$ communicate with $T_1$ only through
                    %edges $T_1^{(i)} = O_1^{(i)}$. But $|O_1^{(i)}| = 1$. Consequently, the removal of the single edge
                    %in $O_1^{(i)}$ disconnects $(s_1, T_1), (s_2, T_2)$ and $(s_2, T_1)$ simultaneouly.  
        \end{enumerate}
In summary, if $\rank \left( \GT{2}{i+1} \right) \geq 2$, then $\Granki{i} = 1$ holds if and only if there exists
a single edge GNS cut set in graph $\mathcal{G}$.

\vspace{0.5cm}
\noindent\textbf{Case (II):} $\Granki{i+1} > \rank \left( \GT{2}{i+1} \right) = 1$.
{
Again we start by assuming $\Granki{i} = \rank \left( \GT{2}{i}  \right) = 1$ in stage $i$. We show that
\eqref{eqn:oneedgecutcondition} is true in three steps.
}
%$U_1^{(i)}$ and $B_2^{(i)}$ are empty sets while $s_2$ does not communicate with
%$U_2^{(i)}$ (if it is not empty). Moreover, $O_1^{(i)} = A_2^{(i)}$ and $|O_1^{(i)}| =1$, which gives a single edge
%whose removal that disconnects $(s_1, T_1)$, $(s_2, T_2)$ and $(s_2, T_1)$. We do this in several steps.  
{
\begin{enumerate}[(i)]
    \item We first show that $O_1^{(i)}$ is non-empty, $ \rank \left( \begin{bmatrix} \HT{1}{i+1} \\ \HT{2}{i+1} \end{bmatrix} \right)   
    = \rank \left( 
   \HUIall
    \right) \geq 2$ and that $\Span \left( \HT{2}{i+1} \right) = \Span \left( \GT{2}{i} \right).$
    \item Next we use the result of (i) to show that $U_1^{(i)}$ is empty and $\left|O_1^{(i)}\right| = 1$.
    \item Finally we use the results of (i) and (ii) to show that $\GU{2} = \mathbf{0},$ $B_2^{(i)}=\phi$, and $A_1^{(i)} = O_1^{(i)}$. 
\end{enumerate}
}
We describe the three steps in detail next.
\begin{enumerate}[(i)]
    \item 
We show that $O_1^{(i)}$ is non-empty, $ \rank \left( \begin{bmatrix} \HT{1}{i+1} \\ \HT{2}{i+1} \end{bmatrix} \right)   
    = \rank \left( 
   \HUIall
    \right) \geq 2$ and that $\Span \left( \HT{2}{i+1} \right) = \Span \left( \GT{2}{i} \right).$
        %We first show that $O_1^{(i)}$ is non-empty and the following holds, 
        %\begin{align}
    %\rank \left( 
    %\begin{bmatrix}
        %\HT{1}{i+1} \\ \HT{2}{i+1}
    %\end{bmatrix}
       %\right)   
    %&= \rank \left( 
   %\HUIall
    %\right) \geq 2~.
%\end{align}

Because $\rank \left( \GT{2}{i} \right) = \rank \left( \GT{2}{i+1} \right) = 1$, and because $\GT{2}{i}$ is a linear
combination of columns of $\GT{2}{i+1}$, the column space of $\GT{2}{i+1}$ is the same as that of $\GT{2}{i}$.
Therefore, if $O_1^{(i)}$ is empty, we have 
$$\rank \left( \mathbf{H}_{2}^{T_1^{(i)}}~~\mathbf{G}_{2}^{T_2^{(i)}} \right) = \rank \left( \mathbf{H}_{2}^{T_1^{(i+1)}}~~\mathbf{G}_{2}^{T_2^{(i+1)}} \right),$$ 
which implies the following:
%Suppose that $O_1^{(i)}$ is empty, then $T_1^{(i+1)} = T_1^{(i)}$ while the column space of $\GT{2}{i+1}$ is the same as
%that of $\GT{2}{i}$ following the assumption of $\rank \left( \GT{2}{i} \right) = \rank \left( \GT{2}{i+1} \right) =
%1$. As a result, the Grank is preserved from $i+1$ to $i$, i.e. 
\begin{align*}
\Granki{i} = \Granki{i+1} \geq 2~.
\end{align*}

This contradicts the assumption that the Grank at stage $i$ is $1$. Thus $O_1^{(i)}$ is non-empty. 

We now show that $\Span \left( \HT{2}{i+1} \right) = \Span \left( \GT{2}{i} \right).$
From \eqref{eqn:grank1iff}, we have, 
\begin{align*}
    \rank \left( \HGall \right) =
     \rank \left( 
     \begin{bmatrix}
         \HU{2} & \HO{2} & \GT{2}{i}
     \end{bmatrix}
     \right) = 1~.
\end{align*}
Thus $\Span \left( \begin{bmatrix} \HU{2} & \HO{2} \end{bmatrix} \right) \subseteq \Span \left( \GT{2}{i} \right)$.
However, $\HO{2}$ is generated by random coding from $\HI{2}$. Hence, we must have, 
\begin{align}
\Span \left( 
\begin{bmatrix}
    \HU{2} &  \HI{2}   
\end{bmatrix} \right) \subseteq \Span \left(\GT{2}{i}\right) = \Span \left( \GT{2}{i+1} \right)~.
\label{eqn:HI2span}
\end{align}
As a result, 
%$\rank \left( \begin{bmatrix} \HT{2}{i+1} & \GT{2}{i+1} \end{bmatrix} \right) =1$ and 
$\rank \left( \HT{2}{i+1} \right) = 
\rank \left( 
\begin{bmatrix}
    \HU{2} &  \HI{2}   
\end{bmatrix} 
\right) \leq 1$. Furthermore, the rank of $\HT{2}{i+1}$ is lower bounded by the rank of $\HT{2}{i}$, which is $1$ by
\eqref{eqn:grank1iff}. Therefore, 
\begin{align}
\rank \left(
\HT{2}{i+1}
\right) &= 1 ~,  &
\Span \left( \HT{2}{i+1} \right) = \Span\left( \GT{2}{i+1} \right) = \Span \left( \GT{2}{i} \right).
\label{eqn:lowerspans}
\end{align}

We now show that $ \rank \left( \begin{bmatrix} \HT{1}{i+1} \\ \HT{2}{i+1} \end{bmatrix} \right)   
    = \rank \left( 
   \HUIall
    \right) \geq 2.$ 

    Equation (\ref{eqn:lowerspans}) implies that 
    \begin{equation} rank \left( \HT{2}{i+1} \right) = 1,~~\rank \left( \begin{bmatrix} \HT{2}{i+1} & \GT{2}{i+1} \end{bmatrix} \right) =1.\label{eqn:implicationoflowerspans} \end{equation}
By assumpion, we have,
\begin{align}
    \Granki{i+1} &= \rank \left(  \begin{bmatrix} \HT{1}{i+1} \\ \HT{2}{i+1} \end{bmatrix}\right) 
    + \rank \left( \begin{bmatrix} \HT{2}{i+1} & \GT{2}{i+1} \end{bmatrix}\right) - \rank \left( \HT{2}{i+1} \right) \nonumber \\
        &\geq 2~.
	\label{eqn:Grankagain}
\end{align}

Combining (\ref{eqn:implicationoflowerspans}) and (\ref{eqn:Grankagain}) we get 
\begin{align}
    \rank \left( 
    \begin{bmatrix}
        \HT{1}{i+1} \\ \HT{2}{i+1}
    \end{bmatrix}
       \right)   
    &\geq 2 
   \label{eqn:hpartrank}
\end{align}

We have showed (i) and we proceed to (ii).

\item We show that $U_1^{(i)}$ is empty, and $\left|O_1^{(i)}\right| = 1$.  

Since $\Granki{i} = 1$  we have 
\begin{align*}
\rank \left( 
    \begin{bmatrix}
        \HU{1} \\ 
        \HU{2} 
    \end{bmatrix}
    \right) \leq  \Granki{i} =1 ~.
\end{align*}
Since we have shown that $$  \rank \left( 
    \begin{bmatrix}
        \HT{1}{i+1} \\ \HT{2}{i+1}
    \end{bmatrix}
       \right) =   
\rank \left( \HUIall \right) \geq 2,$$ we must have
    \begin{align}
\Span\left( 
    \begin{bmatrix}
        \HI{1} \\ 
        \HI{2} 
    \end{bmatrix}
\right) \not \subset
\Span \left(
    \begin{bmatrix}
        \HU{1} \\ 
        \HU{2} 
    \end{bmatrix}
    \right)~.
\end{align}
Since $
\begin{bmatrix}
    \HO{1} \\ \HO{2}
\end{bmatrix}
$ is generated from $
\begin{bmatrix}
    \HI{1} \\ \HI{2}
\end{bmatrix}
$ by random coding, we have $\Span \left(\HOall \right) \not \subset \Span \left( \HUall \right)$. Consequently,
\begin{align*}
    \Granki{i} \geq \rank \left( 
    \begin{bmatrix}
        \HU{1} & \HO{1} \\
        \HU{2} & \HO{2}
    \end{bmatrix}
    \right) \geq 2~,
\end{align*}
which contradicts the assumption. Therefore, we conclude,
\begin{align}
\rank \left( 
    \begin{bmatrix}
        \HU{1} \\ 
        \HU{2} 
    \end{bmatrix}
    \right) &= 0 ~, &
\rank \left( 
    \begin{bmatrix}
        \HI{1} \\ 
        \HI{2} 
    \end{bmatrix}
    \right) &\geq 2 ~.
\end{align}
But by Lemma~\ref{lma:emptycols}, there is no all zero submatrix of $\HTall$. Hence, we conclude that $U_1^{(i)} = \varnothing$.

The fact that $\rank \left(\HIall \right) \geq 2$ combined with \eqref{eqn:grank1iff} can be used to show that
$\left|O_1^{(i)}\right| < 2$. In particular if $\left|O_1^{(i)}\right| \geq 2$, then we have $\rank
\left(
\begin{bmatrix}
    \HO{1} \\ \HO{2}
\end{bmatrix}
\right) \geq 2$ and thus $\rank \left( \HTall \right) \geq 2$, which violates \eqref{eqn:grank1iff}. Since we have
already shown that $O_1^{(i)}$ is non-empty, we conclude that $\left|O_1^{(i)}\right| = 1$.

\item $\GU{2} = \mathbf{0}$, $B_2^{(i)}$ is empty and $A_1^{(i)} = O_1^{(i)}$. 

    Suppose that $ \begin{bmatrix} \GU{2} & \GB{2} \end{bmatrix} \neq \mathbf{0}$. We show that this implies that there
    is an alignment step at stage $i$ contradicting our assumption that the stage contains only randomization steps.
    This will imply that $ \begin{bmatrix} \GU{2} & \GB{2} \end{bmatrix} = \mathbf{0}.$

		Since, by Claim \ref{claim2}, 
		$$\rank \left( \GT{2}{i} \right) =  \rank\left(  \begin{bmatrix}  \GU{2} & \GA{2}& \GB{2} \end{bmatrix} \right) =
1,$$ we have 
\begin{align*}
 \Span \left(  \begin{bmatrix} \GU{2} & \GB{2} \end{bmatrix} \right) = \Span \left( \GT{2}{i}.
 \right) 
\end{align*}
However, from (i) and the fact that $U_1^{(i)}$ is empty, we have 
\begin{align*}
 \Span \left(  \begin{bmatrix} \GU{2} & \GB{2} \end{bmatrix} \right) = \Span \left( \HI{2} \right).
\end{align*}
As a result, we have the following at the beginning of recursive algorithm Phase $2$ between stage $i+1$ and $i$, 
\begin{align*}
    \Granki{i+1} &>  \Grank \left( \HI{1}, \HI{2}, \GU{2} \right)\\
    \Span \left( \HI{2} \right) &\not \subset \varnothing \\
    \Span \left( \HI{2} \right) &\subseteq \Span \left( \begin{bmatrix} \GU{2} & \GB{2} \end{bmatrix} \right)~.
\end{align*}
In other words, the conditions for alignment step are satsfied. Thus, the algorithm will carry out
an alignment step, which is a contradiction. Therefore, we conclude that $\begin{bmatrix} \GU{2} & \GB{2} \end{bmatrix} = \mathbf{0}$. Given that $\GB{2}$ is obtained by random coding from $\GI{2}$, we infer that $B_2^{(i)} = \varnothing$; otherwise we
would have $\GB{2} \neq 0$. Furthermore, with $\rank \left(\GT{2}{i}\right) = 1$, we must have $\rank \left( \GA{2} \right) = 1$ and thus $A_2^{(i)}$ is not empty. But
$A_2^{(i)}$ is a subset of $O_1^{(i)}$ and $\left|O_1^{(i)}\right| = 1$, so we get $A_2^{(i)} = O_1^{(i)}$. 

Therefore, we have $\GU{2} = \mathbf{0}$, $B_2^{(i)} = \varnothing$ and $A_2^{(i)} = O_1^{(i)}$. 
\end{enumerate}
This completes the proof.
\end{proof}

\section{Conclusion}
\label{sec:conclusion}
The techniques of routing and random network coding have served as pillars of our encoding function design in networks.
These techniques are loosely analogous to wireless network achievability techniques of orthogonalization and random
coding combined with treating interference as noise respectively. The paradigms of orthogonalization and random coding
were challenged by interference alignment in \cite{Cadambe_Jafar_int}. An important milestone in the development of
interference alignment for wireless networks was the development of numerical alignment algorithms
\cite{Gomadam_Cadambe_Jafar, Peters_Heath}. In this paper, we have undertaken an analogous effort for alignment in
wireline network coding. 

%Our algorithm is based on a simple observation: the goal of maximizing the rank of the end-to-end transform leads naturally to random linear network coding. In two-unicast-Z networks, the rank is not the appropriate metric; instead we aim to maximize the metric that represents sum-rate of the corresponding implied $Z$-interference channel. In other words, our paper treats the sum-capacity function of the $Z$-interference channel as \emph{the commodity} to be maximized, and studies this commodity.

Our paper leads to several open problems. We initiate the study of two-unicast-Z networks. For two-unicast networks, we
know that linear coding is insufficient for capacity in general, and that edge-cut outerbounds \cite{Kramer_Savari,
Kamath_Tse_Ananthram} are loose. In contrast, it is an open question whether even scalar linear network is sufficient
for two-unicast-Z networks; similarly, it is not known whether the GNS-cut set bound is loose for two-unicast-Z
networks. Second, our approach to maximizing the sum-rate is rather myopic, since we greedily optimize the network
coding co-efficients one edge at a time. In comparison, linear programming  based algorithms have been formulated for
routing, and for network coding restricted to binary field. Development of similar formulations for optimizing the
sum-rate of the two-unicast-Z network is a promising research direction. Finally, an interesting question is how our
approach compares to other approaches when translated to the index coding problem through the construction of
\cite{Effros_Equivalence}.
%optimization problem has been well formulated linear programming solutions exist that maximize 
%its development, the field of interference alignment has matured and is relatively well understood \cite{}, even though several open problems remain. In contrast, 

\bibliographystyle{ieeetr}
\bibliography{main}
%\begin{appendix}
\appendix
	\section{Proof of Properties (ii), (iii) and (v)}
	\label{app:destreduction}
    We first state and prove Lemma \ref{lem:propertyofdestred} which describes a useful property of the destination
    reduction algorithm. Then we prove Properties $(ii)$ and $(iii)$.
    \begin{lemma} 
        For any two positive integers $p,q$ with $p < q$, the set of all edges in $T_1^{(q)}\cup T_2^{(q)}$
        has a strictly lower topological order with respect to the edges in $\left(T_1^{(p)}\cup T_2^{(p)}\right)
        \backslash \left(T_1^{(p+1)}\cup T_2^{(p+1)}\right).$
	\label{lem:propertyofdestred}
	\end{lemma}
	\begin{proof}
For any integer $m$, let us denote by $e^{(m)},$ an edge of the highest topological order in set $T_1^{(m)} \cup
T_2^{(m)}.$ From the destination reduction algorithm, we note that for any integer $m$, every edge in $\left(T_1^{(m)}
\cup T_2^{(m)}\right)\backslash \left(T_1^{(m+1)} \cup T_2^{(m+1)}\right)$ has the same topological order as $e^{(m)}$.
Furthermore, from the algorithm, we note that $e^{(m)}$ has a strictly higher topological order with respect to every
edge in $T_1^{(m+1)}\cup T_2^{(m+1)}$; specifically, $e^{(m)}$ has a higher topological order with respect to
$e^{(m+1)}$. This is because, in the process of obtaining $\Omega^{(m+1)}$ from $\Omega^{(m)}$, we remove the edges with
the highest topological order from the two destination sets $T_1^{(m)}\cup T_2^{(m)}$ and replace them with the
immediate parent edges, which have a strictly lower topological order. By the transitive nature of partial ordering, we
infer that $e^{(p)}$ has a higher topological order with respect to $e^{(q)}$. Therefore $e^{(p)}$ has a higher
topological order with respect to  every edge in $T_1^{(q)} \cup T_2^{(q)}.$ Furthermore, since every edge in
$\left(T_1^{(p)}\cup T_2^{(p)}\right) \backslash \left(T_1^{(p+1)}\cup T_2^{(p+1)}\right)$ has the same topological
order as $e^{(p)}$, we conclude that every edge in this set has a strictly higher topological order with respect to
every edge in $T_1^{(q)} \cup T_2^{(q)}$.
		%claim that $e^{(p){ we know that $e^{(p+1)}$ has a higher topological order with respect to $e^{(p+2)}$ which has a higher topological order with respect to $e^{(p+3)}$. 		where $\prec$ represents the topological order, that is for $e_1 \prec e_2$ . Therefore $e^{(p)}$ 
	\end{proof}

	\emph{Proof of Property (ii):}
    For $j \in \{1,2\}$, let $K_j$ be the set of edges that communicate with destination $T_j,$ but do not belong to
    $\bigcup_{0 \leq k \leq N} T_j^{(k)}$. We show that $K_1$ and $K_2$ are empty by contradiction. 

    Suppose that $K_1$ is non-empty. Let $e$ be the highest topologically ordered edge in $K_1.$ Since there is a path
    from $e$ to destination $1$, there is at least  on edge $e'$ which is a member of $\Out(\Head(e)),$ such that there
    is a path from $e'$ to destination $1$. Furthermore, since $e'$ does not belong to $K_1$, there exists an integer
    $k$ such that $e'$ lies in $T_{1}^{(k)} \cup T_{2}^{(k)}$ for some value of $k$. Let $k^*$ denote the largest among
    such integers. Therefore, we note that $e'$ does not lie in $T_1^{(k^*+1)}\cup T_2^{(k^*+1)}$. By examining the
    destination reduction algorithm, we infer that $e'$ is among the edges of the highest topological order in
    $T_1^{(k^*)}\cup T_2^{(k^*)}.$ In the process of obtaining $\Omega^{(k^*+1)}$ from $\Omega^{(k^*)}$, all the edges
    in $\In(\Tail(e'))$ are added to $T_{1}^{(k^*+1)}$ or $T_2^{(k^*+1)}$ or both sets, depending on whether $e'$
    belongs to $T_1^{k^*}$, or $T_2^{(k^*)}$ or both sets. In particular, edge $e$ is added to $T_1^{(k^*+1)} \cup
    T_2^{(k^*+1)}$. Therefore edge $e$ does not lie in $K_1$ which contradicts our earlier assumption. Therefore the set
    $K_1$ is empty. We can similarly show that $K_2$ is empty.

    Therefore, we have shown that if an edge communicates with destination $j$ for $j \in \{1,2\}$, then it belongs to
    $T_j^{(k)}$ for some integer $k$. Let $k_1$ be the largest integer such that $e$ belongs to $T_1^{(k_1)}$. Let $k_2$
    be the largest integer such that $e$ belongs to $T_2^{(k_2)}$. To complete the proof, we show that $k_1=k_2$. As a
    contradiction, suppose that $k_1 \neq k_2$. Without loss of generality, assume that $k_1 > k_2$. Then, we note that
    edge $e$ belongs to $T_{2}^{(k_2)}$, but does not belong to $T_2^{(k_2+1)}$. By examination of the destination
    reduction algoritm, we infer that this means that edge $e$ is among the highest topologically ordered edges in
    $T_2^{(k_2)},$. and that edge $e$ has a higher topological order with respect to all edges in $T_1^{(k_1+1)}\cup
    T_2^{(k_2+1)}.$ 

    Since $k_1 > k_2$ and $e \in \left(T_1^{(k_2+1)}\cup T_2^{(k_2+1)}\right) \backslash \left(T_1^{(k)} \cup
    T_2^{(k)}\right),$ Lemma \ref{lem:propertyofdestred} implies that $e$ has a strictly higher topological order with
    respect to all the edges in $T_1^{(k_1)}\cup T_2^{(k_1)}$. Therefore $e$ cannot belong to $T_1^{(k_1)}$
    contradicting our earlier assumption. Therefore $k_1 = k_2$. 

	\emph{Proof of Property (iii):}
    Let ${K}^{(i)}$ represent the set of all edges with a lower topological order with respect to $\left(T_1^{(i)} \cup
    T_2^{(i)}\right)  \backslash \left(T_1^{(i+1)} \cup T_2^{(i+1)}\right).$ From Lemma \ref{lem:propertyofdestred} we
    infer that $$\bigcup_{i+1 \leq k \leq N} T_{1}^{(k)} \cup T_{2}^{(k)} \subseteq K^{(i)}.$$ To complete the proof, we
    need to show the reverse, that is, we need to show that $$K^{(i)} \subseteq \bigcup_{i+1 \leq k \leq N} T_{1}^{(k)}
    \cup T_{2}^{(k)} .$$

    Consider any edge $e$ in $K^{(i)}$. Since $e$ communicates with one of the two destinations, Property $(ii)$ implies
    that $e$ lies in $T_1^{(m)} \cup T_2^{(m)}$ for some integer $m$. Let $\overline{m}$ denote the largest integer such
    that  $e$ lies in  $T_1^{(\overline{m})} \cup T_2^{(\overline{m})}$. To complete the proof, it suffices to show that
    $\overline{m} \geq i+1$. We show this by contradiction. Suppose that $\overline{m} < i$. Because $e$ lies in
    $\left(T_1^{(\overline{m})} \cup T_2^{(\overline{m})}\right)  \backslash \left(T_1^{(\overline{m}+1)} \cup
    T_2^{(\overline{m}+1)}\right),$ we infer from Lemma \ref{lem:propertyofdestred} that $e$ has a higher topological
    order with respect to every edge in $T_1^{(i)}\cup T_{2}^{(i)}.$ This contradicts the assumption that $e$ lies in
    $K^{(i)}$. Therefore, we have $\overline{m} \geq i$. If $\overline{m} = i,$ then $e$ lies in $\left(T_1^{({i})} \cup
    T_2^{({i})}\right)  \backslash \left(T_1^{({i}+1)} \cup T_2^{({i}+1)}\right).$ This also contradicts the assumption
    that $e$ lies in $K^{(i)}$. Therefore we have $m > i$.  This completes the proof.

    {
    \emph{Proof of Property (v):}
    We show that each set $T_1^{(i)}$ is a cut set between $s_1, s_2$ and destination edge set $T_1$. The same argument
    applies to the case of $T_2^{(i)}$. Note that it suffices to show that any path between $s_1$ or $s_2$ and some edge
    in $T_1$ passes through at least one edge in each set $T_1^{(i)}$, $i=0,1, \dots, N$. Specifically, consider an
    arbitrary path $P$ with length $n$, denoted by a sequence of edges, i.e. $P = \left\{ e_1, e_2 ,
    \dots, e_n \right\}$, where $\Head(e_i) = \Tail(e_{i+1})$ for $1 \leq 1 \leq n$.
    We show that if $\Tail(e_1) = s_1$ or $\Tail(e_1) = s_2$ and $e_n \in T_1$, then $T_1^{(i)} \cap P \neq
    \varnothing$ for all $i=1,2,\dots, N$. 
    
    First consider the case $i=0$,
    since $e_n \in T_1$ and $T_1^{(0)} = T_1$, clearly $T_1^{(0)} \cap P \neq \varnothing$. To complete the proof, we assume that for some
    $k$, $0\leq k \leq N$, $T_1{(k)} \cap P \neq \varnothing$ holds and show that $T_1^{(k+1)} \cap P \neq \varnothing$. Let
    $e_j$ be some edge belonging to both $P$ and $T_1^{(k)}$. If $e_j \in T_1^{(k+1)}$, then clearly, $T_1^{(k+1)} \cap
    P$ is non-empty. Suppose otherwise, i.e. $e_j \not \in T_1^{(k+1)}$, then $ e_j \in T_1^{(k)} \backslash
    T_1^{(k+1)}$. By the construction of the destination reduction algorithm, we know that $e_j$ is not a source edge
    and $\In(\Tail(e_j)) \subseteq T_1^{(k+1)}$. As a result, $e_j \neq e_1$ and $e_{j-1} \in \In(\Tail(e_j))$.  We
    have $e_{j-1} \in P$ and $e_{j-1} \in T_1^{(k+1)}$, that is, we have shown that $T_1^{(k+1)} \cap P \neq \varnothing$. This completes the proof. 
    }

\section{Proof of Lemma~\ref{lma:solutioncombining}}
\label{solutioncombining}
   For convenience of notation, let $c_{\mathcal{G}_n}(s,T_n) = M$. Since,
 \begin{align}
     \rank(\mathbf{H}_{n-1}) = \rank\left( \mathbf{A}_{n-1} \left( \mathbf{I}_{n-1} -
            \mathbf{F}_{n-1} \right)^{-1}\mathbf{B}_{n-1}^T \right) 
            &= M - 1 \\
      \rank(\mathbf{H}^*_{n-1}) = \rank\left( \mathbf{A}_{n-1} \left( \mathbf{I}_{n-1} -
            \mathbf{F}^*_{n-1} \right)^{-1} \mathbf{B}^{*^T}_{n-1} \right) &\geq M ~,
\end{align}
there exists an $(M-1) \times (M-1)$ submatrix $S_1$ in $\mathbf{H}_{n-1}$ whose determinant $f_1$ is not zero.
Similarly in matrix $\mathbf{H}^*_{n-1}$, there exists some $M \times M$ submatrix $S_2$, whose determinant $f_2$ is not
zero. 

Next consider the submatrix $S_1$ in $\mathbf{A}_{n-1} \left( \mathbf{I}_{n-1} - p\mathbf{F}_{n-1} - q\mathbf{F}^*_{n-1}
\right)^{-1} \mathbf{B}^{*^T}_{n-1}$ as well as the submatrix $S_2$ in $\mathbf{A}_{n-1} \left( \mathbf{I}_{n-1} -
p\mathbf{F}_{n-1} - q\mathbf{F}^*_{n-1} \right)^{-1} \mathbf{B}^{T}_{n-1}$. Let their determinants be
$f_1(p,q)$ and $f_2(p,q)$. Clearly $f_1(p,q) \neq 0$ and $f_2(p,q) \neq 0$ since $f_1(1,0) \neq 0$ and $f_2(0,1)\neq 0$.
Thus, $f(p,q) = f_1(p, q) f_2(p,q)$ is a non-zero polynomial.  Applying Lemma $4$ in \cite{Ho_etal_Award} to $f(p,q)$,
we conclude that if the underlying field $\mathbb{F}$ is large enough, choosing $p, q$ uniformly at random from
$\mathbb{F}$ will yield $f(p,q) \neq 0$ with a probability that tends to $1$ as the field size increases.  Equivalently, with the random choices of $p$ and $q$,  
 \begin{align}
    \rank\left( \mathbf{A}_{n-1} \left( \mathbf{I}_{n-1} -
    p\mathbf{F}_{n-1} - q\mathbf{F}^*_{n-1} \right)^{-1} \mathbf{B}_{n-1}^T \right) 
    &= M - 1 \\
    \rank\left( \mathbf{A}_{n-1} \left( \mathbf{I}_{n-1} -
    p\mathbf{F}_{n-1} - q\mathbf{F}^*_{n-1} \right)^{-1} \mathbf{B}^{*^T}_{n-1}
    \right) &\geq M ~.
\end{align}

Furthermore, since $ 
\mathbf{B}^{*^T}_{n-1} =
    \left[
        \begin{array}[h]{c|c|c|c}
            \mathbf{B}^T_{n-1} & \mathbf{u}_{i_1} & \dots &
            \mathbf{u}_{i_h}
        \end{array}
    \right] 
$, 
we conclude that the matrix
\[
\mathbf{A}_{n-1} \left( \mathbf{I}_{n-1} -
    p\mathbf{F}_{n-1} - q\mathbf{F}^*_{n-1} \right)^{-1} 
    \left[
        \begin{array}[h]{c|c|c}
             \mathbf{u}_{i_1} & \dots &
            \mathbf{u}_{i_h}
        \end{array}
    \right]~,
\]
which are the symbols carried by the edges $e_{i_1}$ to $e_{i_h}$, contains at least $1$ columns that are linearly
independent from the columns of $\mathbf{A}_{n-1} \left( \mathbf{I}_{n-1} - p\mathbf{F}_{n-1} - q\mathbf{F}^*_{n-1}
\right)^{-1} \mathbf{B}_{n-1}^T $. Because $e_{i_1}, \dots, e_{i_h}$ are the parent edges of $e_{n-1+1}, \dots, e_n$, by
choosing the local coding vectors for the last edge uniformly at random from $\mathbb{F}$, we are setting the columns of  
\[
\mathbf{A}_{n-1} \left( \mathbf{I}_{n-1} -
    p\mathbf{F}_{n-1} - q\mathbf{F}^*_{n-1} \right)^{-1} 
    \begin{bmatrix}
         \mathbf{e}_{n-1+1}  & \dots & \mathbf{e}_n 
    \end{bmatrix}
\]
to be a random linear combinations of the columns of
\[
\mathbf{A}_{n-1} \left( \mathbf{I}_{n-1} -
    p\mathbf{F}_{n-1} - q\mathbf{F}^*_{n-1} \right)^{-1} 
    \left[
        \begin{array}[h]{c|c|c}
             \mathbf{u}_{i_1} & \dots &
            \mathbf{u}_{i_h}
        \end{array}
    \right]~.
\]
Therefore,
\[
\mathbf{A}_{n-1} \left( \mathbf{I}_{n-1} -
    p\mathbf{F}_{n-1} - q\mathbf{F}^*_{n-1} \right)^{-1} 
    \begin{bmatrix}
         \mathbf{e}_n
    \end{bmatrix}
\]
contains exactly one linearly independent column from 
$\mathbf{A}_{n-1} \left( \mathbf{I}_{n-1} -
    p\mathbf{F}_{n-1} - q\mathbf{F}^*_{n-1} \right)^{-1}
    \mathbf{B}_{n-1}^T$.
Hence, we conclude that,
\begin{align}
    \rank\left(
\mathbf{A}_{n-1} \left( \mathbf{I}_{n-1} - p\mathbf{F}_{n-1} - q
\mathbf{F}_{n-1}^* \right)^{-1}
    \left[
    \begin{array}[h]{c|c}
        \mathbf{B}_{n-1}^T   & \mathbf{e}_{n}
    \end{array}
\right]\right) = M 
\end{align}

\section{Proof of Lemma~\ref{lem:alignment}}
\label{neutralizeproof}
Since $\Grank\left( [\mH{1} ~~\mathbf{A}], \left[ \mH{2} ~~ \mathbf{B} \right], \mG{2}\right) > \Grank\left(
\mH{1}, \mH{2}, \mG{2} \right)$, it is clear that, 
$\Span\left(
\begin{bmatrix}
    \mathbf{A} \\ \mathbf{B}
\end{bmatrix}\right) \not \subset
\Span\left(
\begin{bmatrix}
    \mH{1} \\ \mH{2}
\end{bmatrix}
\right)$.

We first consider the case when at least one of conditions $(i)$ and $(ii)$ does not hold.  In this case, we pick
entries of $\mathbf{f}$ uniformly at random from a large enough field.
\begin{enumerate}
\item  If $\Span\left( \mathbf{B} \right) \in \Span\left( \mH{2} \right)$,  then $\mathbf{Bf} \in \Span\left(
    \mH{2} \right)$, but $\mathbf{Af} \not \in \Span\left( \mH{1} \right)$. Consequently,  
    \begin{align*}
        \rank\left( \begin{bmatrix} \mH{1} & \mathbf{Af} \\ \mH{2} & \mathbf{Bf}\end{bmatrix} \right) &= 
            \rank\left( \begin{bmatrix} \mH{1} \\ \mH{2} \end{bmatrix} \right) + 1 ~, \\
            \rank\left( [\mH{2} ~~ \mathbf{Bf} ~~ \mG{2} ] \right) &= \rank\left([\mH{2}~~ \mG{2}] \right)~, \\
            \rank\left( [\mH{2} ~~ \mathbf{Bf} ] \right) &= \rank\left( \mH{2}\right)~.
    \end{align*}
\item If $\Span\left( \mathbf{B} \right) \not \in \Span \left( [\mH{2} ~~ \mG{2}] \right)$, then $\mathbf{Bf}
    \not \in \Span\left([\mH{2} ~~ \mG{2}] \right)$, and  $\mathbf{Af} \not \in \Span\left( \mH{1} \right)$.
    Hence,
    \begin{align*}
        \rank\left( \begin{bmatrix} \mH{1} & \mathbf{Af} \\ \mH{2} & \mathbf{Bf}\end{bmatrix} \right) &= 
            \rank\left( \begin{bmatrix} \mH{1} \\ \mH{2} \end{bmatrix} \right) + 1 ~, \\
            \rank\left( [\mH{2} ~~ \mathbf{Bf} ~~ \mG{2} ] \right) &= \rank\left([\mH{2}~~ \mG{2}] \right)+1~, \\
            \rank\left( [\mH{2} ~~ \mathbf{Bf} ] \right) &= \rank\left( \mH{2}\right)+1~.
    \end{align*}
\end{enumerate}
Therefore, in both cases, we $\Grank\left( [\mH{1} ~~\mathbf{Af}], \left[ \mH{2} ~~ \mathbf{Bf} \right], \mG{2}\right) =
\Grank\left( \mH{1}, \mH{2}, \mG{2} \right) + 1$.

It remains to find the vector $\mathbf{f}$ that satisfies \eqref{eq:resultingcondition} when conditions $(i)$ and $(ii)$
both hold. In particular, it suffices to find a vector $\mathbf{f}$ such that
\begin{align}
    \begin{bmatrix}
        \mathbf{A} \\ \mathbf{B}
    \end{bmatrix}\mathbf{f} 
    &\not\in 
    \Span\left(
    \begin{bmatrix}
        \mH{1} \\ \mH{2}
    \end{bmatrix}\right)~, \label{eqn:fcondition1}\\
    \mathbf{Bf} &\in \Span\left(\mH{2}\right)~. \label{eqn:fcondition2}
\end{align}
{
Suppose the columns of the matrix $\begin{bmatrix}\mH{1} \\ \mH{2}\end{bmatrix}$ are indexed by the index set $[Q_1] =
\{1,2,\dots, Q_1\}$, while
the columns are from $\begin{bmatrix} \mA \\ \mB \end{bmatrix}$ are indexed by the index set $[M] = \{1,2,\dots, M\}$. 
}
Let $\Gamma$ be the subset
of $[Q_1]$ such that, for the matrix $\mH{2}$, its submatrix consists columns indexed by $\Gamma$ forms basis of the
matrix. Denote this submatrix as $\mH{2}^\Gamma$. Thus, $\rank\left( \mH{2} \right) = |\Gamma|.$ By the basis extension theorem, there exists a subset of columns of $\begin{bmatrix} \mA \\ \mB \end{bmatrix}$, indexed by $\alpha_B \subset [M]$, such
that the columns of matrix $[\mH{2}^\Gamma ~~ \mB^{\alpha_B}]$ form a basis of $\left[ \mH{2}~~\mB \right]$.  Since columns of $\mH{2}^\Gamma$ are linearly independent, columns of $\mH{}^\Gamma =  \begin{bmatrix} \mH{1}^\Gamma \\ \mH{2}^\Gamma \end{bmatrix}$ are also linearly independent. Hence, there
exists a subset $\alpha_H \subset [Q_1]$, such that $\alpha_H \cap \Gamma = \varnothing$, such that the columns of $\begin{bmatrix}
    \mH{1}^\Gamma & \mH{1}^{\alpha_H}  \\
    \mH{2}^\Gamma & \mH{2}^{\alpha_H} 
\end{bmatrix}
$ form a basis of $\begin{bmatrix} \mH{1} \\ \mH{2} \end{bmatrix}$. Consider the matrix 
$
\begin{bmatrix}
    \mH{1}^\Gamma & \mH{1}^{\alpha_H} & \mA^{\alpha_B} \\
    \mH{2}^\Gamma & \mH{2}^{\alpha_H} & \mB^{\alpha_B} 
\end{bmatrix}
$. We show that its columns are linearly independent. 
Consider  a length $|\Gamma|$ vector $\mathbf{c}^\Gamma$, a length $|\alpha_H|$ vector
$\mathbf{c}^{\alpha_H}$ and a length $|\alpha_B|$ vector $\mathbf{c}^{\alpha_B}$, such that
\begin{align}
\begin{bmatrix}
    \mH{1}^\Gamma & \mH{1}^{\alpha_H} & \mA^{\alpha_B} \\
    \mH{2}^\Gamma & \mH{2}^{\alpha_H} & \mB^{\alpha_B} 
\end{bmatrix}
\begin{bmatrix}
    \mathbf{c}^\Gamma \\ \mathbf{c}^{\alpha_H} \\ \mathbf{c}^{\alpha_B}
\end{bmatrix} 
= \mathbf{0}~.
\label{eqn:atemporaryone}
\end{align}
We aim to show that $\begin{bmatrix}
    \mathbf{c}^\Gamma \\ \mathbf{c}^{\alpha_H} \\ \mathbf{c}^{\alpha_B}
\end{bmatrix} = \mathbf{0}.$
Since the columns of $\mH{2}^\Gamma$ is a basis for $\mH{2}$, there exists a matrix $\mathbf{D}$ such that
$\mH{2}^{\alpha_H} = \mH{2}^\Gamma$. Thus, 
\begin{align}
\begin{bmatrix}
    \mH{2}^\Gamma & \mH{2}^{\alpha_H} & \mB^{\alpha_B} 
\end{bmatrix}
\begin{bmatrix}
    \mathbf{c}^\Gamma \\ \mathbf{c}^{\alpha_H} \\ \mathbf{c}^{\alpha_B}
\end{bmatrix} 
= 
\mathbf{H}_2^\Gamma (\mathbf{c}^\Gamma + \mathbf{D} \mathbf{c}^{\alpha_H}) +
\mB^{\alpha_B} \mathbf{c}^{\alpha_B} &= \mathbf{0}.
\end{align}
But $\begin{bmatrix} \mathbf{H}_2^{\Gamma} & \mB^{\alpha_B} \end{bmatrix}$ forms a basis. As a result,
$\mathbf{c}^{\alpha_B} = \mathbf{0}$ and $\mathbf{c}^\Gamma + \mathbf{D}\mathbf{c}^{\alpha_H} = \mathbf{0}$.
Substituting this in (\ref{eqn:atemporaryone}), we get
\begin{align}
\begin{bmatrix}
    \mH{1}^\Gamma & \mH{1}^{\alpha_H} \\ 
    \mH{2}^\Gamma & \mH{2}^{\alpha_H} 
\end{bmatrix}
\begin{bmatrix}
    \mathbf{c}^\Gamma \\ \mathbf{c}^{\alpha_H} 
\end{bmatrix} 
= \mathbf{0}~.  
\end{align} 
Since columns of $
\begin{bmatrix}
    \mH{1}^\Gamma & \mH{1}^{\alpha_H} \\ 
    \mH{2}^\Gamma & \mH{2}^{\alpha_H} 
\end{bmatrix}$ also form a basis, we must have $\mathbf{c}^\Gamma = \mathbf{0}$ and $\mathbf{c}^{\alpha_H} = \mathbf{0}$.
Hence, columns of 
$
\begin{bmatrix}
    \mH{1}^\Gamma & \mH{1}^{\alpha_H} & \mA^{\alpha_B} \\
    \mH{2}^\Gamma & \mH{2}^{\alpha_H} & \mB^{\alpha_B} 
\end{bmatrix}
$ are linearly independent.

By the basis extension theorem,  there exits a set $\beta_B  \subset
[M] \backslash \alpha_B$, such that the columns of
$
\begin{bmatrix}
    \mH{1}^\Gamma & \mH{1}^{\alpha_H} & \mA^{\alpha_B} & \mA^{\beta_B} \\
    \mH{2}^\Gamma & \mH{2}^{\alpha_H} & \mB^{\alpha_B} & \mB^{\beta_B} 
\end{bmatrix}
$
 form a basis for the matrix 
$
\begin{bmatrix}
    \mH{1} & \mA \\
    \mH{2} & \mB
\end{bmatrix}
$. We next show that $|\beta_B| > 0.$

We know that $\Grank \left( [\mH{1}~~\mA], [\mH{2}~~ \mB], \mG{2} \right) > \Grank\left( \mH{1}, \mH{2}, \mG{2} \right)$. Therefore, 
\begin{align*}
&\Grank \left( [\mH{1}~~\mA], [\mH{2}~~ \mB], \mG{2} \right) - \Grank\left( \mH{1}, \mH{2}, \mG{2} \right) \\
  &= \rank\left(
  \begin{bmatrix}
      \mH{1} & \mA \\
      \mH{2} & \mB 
  \end{bmatrix}
  \right) 
  - \rank\left( 
  \begin{bmatrix}
      \mH{1} \\ \mH{2}
  \end{bmatrix}
  \right) 
  +
  \rank\left( 
\begin{bmatrix}
    \mH{2} &  \mB & \mG{2} \\
\end{bmatrix} \right)  
-\rank\left( 
\begin{bmatrix}
    \mH{2} & \mG{2} \\
\end{bmatrix} 
\right) \\
  & \qquad -\rank\left(
  \mH{2}
    \right)
+\rank\left(
\begin{bmatrix}
    \mH{2} & \mB 
\end{bmatrix}
\right) \\
&\stackrel{(a)}{=}
\rank\left(
  \begin{bmatrix}
      \mH{1} & \mA \\
      \mH{2} & \mB 
  \end{bmatrix}
  \right) 
  - \rank\left( 
  \begin{bmatrix}
      \mH{1} \\ \mH{2}
  \end{bmatrix}
  \right) 
-
\left(
\rank\left(
  \mH{2}
    \right)
-\rank\left(
\begin{bmatrix}
    \mH{2} & \mB 
\end{bmatrix}
\right) \right) \\
& = (|\Gamma| + |\alpha_H| + |\alpha_B| + |\beta_B| - |\Gamma|-|\alpha_H|) -(|\Gamma|  - |\Gamma|+ |\alpha_B|)  \\
&= |\beta_B| > 0~,
\end{align*}
where $(a)$ follows from condition (ii), i.e. $\Span \left(\mB\right) \subset 
\Span\left(
\begin{bmatrix}
    \mH{2} & \mG{2}
\end{bmatrix}
\right)$.

Since $\beta_B \neq \varnothing$, note that the columns of $\begin{bmatrix} \mathbf{H}_2^\Gamma & \mB^{\alpha_B} &
    \mB^{\beta_B} \end{bmatrix}$ is a linearly dependent set of vectors. This is because, this set of vectors spans $[\mathbf{H}_{2}~~\mathbf{B}]$, whose basis is $ \begin{bmatrix} \mathbf{H}_2^\Gamma & \mB^{\alpha_B}\end{bmatrix}.$ The nullspace of $\begin{bmatrix} \mathbf{H}_2^\Gamma & \mB^{\alpha_B} &
    \mB^{\beta_B} \end{bmatrix}$ therefore contains at least one non-zero
vector. Pick an arbitrary non-zero vector from this nullspace. Denote the vector as 
be $\begin{bmatrix}{\mathbf{c}^{\Gamma}} \\  {\mathbf{c}^{\alpha_B}}
     \\ {\mathbf{c}^{\beta_B}} \end{bmatrix}$, where $\mathbf{c}^\Gamma$ is a length $|\Gamma|$ subvector,
$\mathbf{c}^{\alpha_B}$ is a length $|\alpha_B|$ subvector and $\mathbf{c}^{\beta_B}$ is a length $|\beta_B|$ subvector.
We have,
$\mathbf{c}^{\beta_B} \neq \mathbf{0}$ and $ \mH{2}^{\Gamma}
\mathbf{c}^{\Gamma} + \mB^{\alpha_B} \mathbf{c}^{\alpha_B} +
\mathbf{B}^{\beta_B} \mathbf{c}^{\beta_B} = \mathbf{0}$, hence,
$\mB^{\alpha_B} \mathbf{c}^{\alpha_B} + \mathbf{B}^{\beta_B} \mathbf{c}^{\beta_B} \in 
\Span\left( \mH{2}^\Gamma
\right) = \Span\left( \mH{2} \right)$. On the other hand,
since columns of
$
\begin{bmatrix}
    \mH{1}^\Gamma & \mH{1}^{\alpha_H} & \mA^{\alpha_B} & \mA^{\beta_B} \\
    \mH{2}^\Gamma & \mH{2}^{\alpha_H} & \mB^{\alpha_B} & \mB^{\beta_B} 
\end{bmatrix}
$
is a basis,
%form a basis for 
%$
%\begin{bmatrix}
    %\mH{1} & \mA \\
    %\mH{2} & \mB
%\end{bmatrix}
%$,
the columns of 
$
\begin{bmatrix}
    \mA^{\alpha_B} & \mA^{\beta_B}  \\
    \mB^{\alpha_B} & \mB^{\beta_B}
\end{bmatrix}
$ is linearly independent of the columns of
$\begin{bmatrix} 
    \mH{1}^\Gamma & \mH{1}^{\alpha_H} \\
    \mH{2}^\Gamma & \mH{2}^{\alpha_H}
\end{bmatrix}$, 
which is a basis for 
$\begin{bmatrix}
    \mH{1} \\ \mH{2}
\end{bmatrix}$. Therefore, 
$\begin{bmatrix} \mA^{\alpha_B} \\ \mB^{\alpha_B}\end{bmatrix} 
\mathbf{c}^{\alpha_B} + \begin{bmatrix} \mA^{\alpha_B} \\ \mB^{\alpha_B}\end{bmatrix}  \mathbf{c}^{\beta_B}
\not \in \Span\left( 
\begin{bmatrix}
    \mH{1} \\ \mH{2}
\end{bmatrix}
\right)$. 

We are now ready to specify $\mathbf{f}$ satisfying (\ref{eqn:fcondition1}) and (\ref{eqn:fcondition2}). For any set $I \subseteq [M]$, let $\mathbf{f}^{I}$ denote the entries of vector $\mathbf{f}^{I}$ corresponding to set $I$.  We specify the $M \times 1$ vector $\mathbf{f}$ as follows: $\mathbf{f}^{\alpha_{B}} = \mathbf{c}^{\alpha_B},$  $\mathbf{f}^{\beta_{B}} = \mathbf{c}^{\beta_B},$ and $\mathbf{f}^{I} = \mathbf{0}$ for any set $I$ which is disjoint with $\alpha_{B} \cup \beta_B$. We note that with this choice of $\mathbf{f}$, we have 
$$\mathbf{B}\mathbf{f} = \mB^{\alpha_B} \mathbf{c}^{\alpha_B} + \mathbf{B}^{\beta_B} \mathbf{c}^{\beta_B} $$  
$$\begin{bmatrix}\mathbf{A}\\ \mathbf{B}\end{bmatrix} \mathbf{f} = \begin{bmatrix} \mA^{\alpha_B} \\ \mB^{\alpha_B}\end{bmatrix} 
\mathbf{c}^{\alpha_B} + \begin{bmatrix} \mA^{\alpha_B} \\ \mB^{\alpha_B}\end{bmatrix}  \mathbf{c}^{\beta_B}
 $$  
It can be readily verified that (\ref{eqn:fcondition1}) and (\ref{eqn:fcondition2}) are satisfied.

{
So far, we have showed that there is at least one vector $\mathbf{f}$ which satisfies \eqref{eqn:fcondition1} and
\eqref{eqn:fcondition2} and subsequently satisfies \eqref{eq:resultingcondition}. Next, we show that if
conditions $(i)$ and $(ii)$ hold, then
picking a random vector $\mathbf{v}$ from the nullspace of 
$
\begin{bmatrix}
    \mH{2} & \mathbf{B}
\end{bmatrix}
$ 
and setting $\mathbf{f}$ to be the last $M$ entries of the random vector $\mathbf{v}$ satisfies \eqref{eq:resultingcondition} with high probability,. In particular, we show that such a vector $\mathbf{f}$ satisfies \eqref{eqn:fcondition1} and \eqref{eqn:fcondition2} with high probability. 

Let $\mathbf{v} = \begin{bmatrix} \mathbf{v}_1 \\ \mathbf{v}_2 \end{bmatrix}$, $\mathbf{v}_1 \in \mathbb{F}^{Q_1},
\mathbf{v}_2 \in \mathbb{F}^{M}$, be a random vector from the nullspace of 
$
\begin{bmatrix}
    \mH{2} & \mathbf{B}
\end{bmatrix}~,
$
i.e. $\mH{2} \mathbf{v}_1 + \mB \mathbf{v}_2 = \mathbf{0}$. Clearly, $\mathbf{f} = \mathbf{v}_2$ satisfies
\eqref{eqn:fcondition2}. We need to show that $\mathbf{f}$ chosen this way satisfies (\ref{eqn:fcondition1}) with high probability. Note that the vector $\mathbf{v}_2$ fails to satisfy \eqref{eqn:fcondition1} if and only if there exists some
$\mathbf{v}_1^\prime \in \mathbb{F}^{Q_1}$, such that, 
\begin{align}
    \begin{bmatrix}
        \mH{1} & \mA  \\
        \mH{2} & \mB
    \end{bmatrix}
    \begin{bmatrix}
    \mathbf{v}_1^\prime \\
    \mathbf{v}_2
    \end{bmatrix}
    = \mathbf{0}~.
\end{align}
Let $\mathcal{R}$ be the set of all such vectors in the nullspace of 
$
\begin{bmatrix}
    \mH{2} & \mathbf{B}
\end{bmatrix}~,
$
i.e.
\begin{align}
    \mathcal{R} = \left\{ 
        \begin{bmatrix}
            \mathbf{v}_1 \\ \mathbf{v}_2
        \end{bmatrix}
        \in \text{Ker}\left( \begin{bmatrix}
    \mH{2} & \mathbf{B}
\end{bmatrix}
\right)
        : 
        \exists~ \mathbf{v}_1^\prime \in \mathbb{F}^{Q_1}, s.t. 
\begin{bmatrix}
        \mH{1} & \mA  \\
        \mH{2} & \mB
    \end{bmatrix}
    \begin{bmatrix}
    \mathbf{v}_1^\prime \\
    \mathbf{v}_2
    \end{bmatrix}
    = \mathbf{0}
    \right\}~.
\end{align}
Clearly, if
$
\begin{bmatrix}
    \mathbf{v}_1 \\ \mathbf{v}_2
\end{bmatrix}
~,
\begin{bmatrix}
    \mathbf{u}_1 \\ \mathbf{u}_2
\end{bmatrix}
\in \mathcal{R}
$, then there exists $\mathbf{v}_1^\prime, \mathbf{u}_1^\prime \in \mathbb{F}^{Q_1}$, s.t.
\begin{align}
 \begin{bmatrix}
        \mH{1} & \mA  \\
        \mH{2} & \mB
    \end{bmatrix}
    \begin{bmatrix}
        a \mathbf{v}_1^\prime + b \mathbf{u}_1^\prime\\
        a \mathbf{v}_2 + b \mathbf{u}_2
    \end{bmatrix}
    = \mathbf{0}
\end{align}
for any $a, b \in \mathbb{F}$, i.e. 
$
    \begin{bmatrix}
        a \mathbf{v}_1^\prime + b \mathbf{u}_1^\prime\\
        a \mathbf{v}_2 + b \mathbf{u}_2
    \end{bmatrix} \in \mathcal{R}~.
$
Hence, $\mathcal{R}$ is a subspace of 
$
\text{Ker}\left( \begin{bmatrix}
    \mH{2} & \mathbf{B}
\end{bmatrix}
\right)
$. 

Furthermore, note that $\mathcal{R}$ is a proper subspace. This is because we have already shown that there exists a vector in the null space of $\begin{bmatrix} \mathbf{H}_{2}&\mathbf{B}\end{bmatrix}$ that satisfy (\ref{eqn:fcondition1}), that is, we have already shown the existence of a vector in the nullspace $\begin{bmatrix} \mathbf{H}_{2}&\mathbf{B}\end{bmatrix}$  that does not lie in $\mathcal{R}$. 
Therefore, by picking $\mathbf{v}$ uniformly at random from  
$
\text{Ker}\left( \begin{bmatrix}
    \mH{2} & \mathbf{B}
\end{bmatrix}
\right)
$, the probability that both condition \eqref{eqn:fcondition1} and \eqref{eqn:fcondition2} are satisfied is
\begin{align}
    1 - P \left( \mathbf{v} \in \mathcal{R} \right)
    = 1 -  \frac{1}{|\mathbb{F}|^t}~,  t \in \mathbb{Z}^+
\end{align}
where $|\mathbb{F}|$ is the size of the underlying field and $t$ is the difference between the dimension of subspace
$\mathcal{R}$ and the dimension of 
$\text{Ker}\left( \begin{bmatrix}
    \mH{2} & \mathbf{B}
\end{bmatrix}
\right)$. Hence, as the field size increases arbitrarily, the probability of both \eqref{eqn:fcondition1} and
\eqref{eqn:fcondition2} hold approaches to $1$. This completes the proof.
}

\section{Proof of Lemma~\ref{lma:emptycols}}
\label{app:emtycolsproof}
    Clearly, at the initial stage, by the construction of the algorithm, every column of the matrix corresponds to a
    source edge and is assigned a unit vector. As a result, none of the columns in the matrix 
    $
    \begin{bmatrix}
        \HT{1}{N} \\  \HT{2}{N}
    \end{bmatrix}
    $
    is all-zero. 
    
    Suppose that there is no all-zero column in the matrix 
    $
    \begin{bmatrix}
        \HT{1}{i+1} \\  \HT{2}{i+1}
    \end{bmatrix}
    $. Clearly, in $\HTall$, no all-zero columns can be generated from random coding of columns of $
   \begin{bmatrix}
        \HT{1}{i+1} \\  \HT{2}{i+1}
    \end{bmatrix}
    $
    with a probability that tends to $1$ as the field size increases.  It remains to consider the case when alignment
    happens. But by
    Lemma~\ref{lem:alignment}, the alignment step will generate a column that does not belong to the column span of the
    matrix $\HUall$. Hence, any column generated by alignment cannot be all zero either. This completes the proof.

\section{Proof of Lemma~\ref{lma:emptyG2}}
\label{app:emtyG2proof}
Suppose that from stage $k+1$ to stage $k$, the algorithm performs some alignment step.  Since $A_2^{(k)} \subset
O_1^{(k)}$, we have $Q^{(k)} = T_2^{(k)} \backslash T_1^{(k)} = U_2^{(k)} \cup B_2^{(k)} \backslash U_1^{(k)}$.  By the
construction of the algorithm, if at some step $j$ in phase $2$, an alignment step takes place, then we have,
\begin{align*}
\Span\left(\mathbf{H}_2^{I^{(i)}_1}\right) &\not \subset
\Span\left(
\begin{bmatrix}
    \mathbf{H}_2^{U^{(i)}_1} & \mathbf{H}_2^{O_{1,j-1}^{(i)}}
\end{bmatrix}
\right),  \\
\Span\left(\mathbf{H}_2^{I^{(k)}_1}\right) &\subset
\Span\left(
\begin{bmatrix}
    \mathbf{H}_2^{U^{(k)}_1} & \mathbf{H}_2^{O_{1,j-1}^{(k)}} & \mathbf{G}_2^{U^{(k)}_2} &
    \mathbf{G}_2^{A^{(k)}_{2,j-1}} & \mathbf{G}_2^{B^{(k)}_2}
\end{bmatrix}
\right)~.
\end{align*}
Since $A_{2,j-1}^{(k)}  \subset O_{1,j-1}^{(k)}$, $\mathbf{G}_2^{A^{(k)}_{2,j-1}}$ is a submatrix of $
\mathbf{H}_2^{O_{1,j-1}^{(k)}} $. For the conditions to hold, we must have 
\begin{align*}
    \mathbf{G}_2^{U_2^{(k)} \cup B_2^{(k)} \backslash U_1^{(k)}} = \mathbf{G}_2^{Q^{(k)}} \neq \mathbf{0} ~.
\end{align*}
Now consider $Q^{(k-1)} = T_2^{(k-1)} \backslash T_1^{(k-1)}$. By the construction of the algorithm, $Q^{(k-1)}$ does
not communicate with $T_1$. Hence, all the columns corresponding to edges in $Q^{(k-1)}$ are either source edge columns
or columns generated by random coding in phase $1$ at some stage of the algorithm. On the other hand, since all edges in
$Q^{(k)}$ communicate with $T_2$ but not $T_1$, for each edge $e_i^{(k)} \in Q^{(k)}$ there exists an edge $e_i^{(k-1)}
\in Q^{(k-1)}$ such that either $e_i^{(k)} = e_i^{(k-1)}$ or $e_i^{(k)}$ is a parent edge of $e_i^{(k-1)}$. Thus, each
column in $\mathbf{G}_2^{Q^{(k)}}$ participates in random coding for at least one edge in $Q^{(k-1)}$. Now since
$\mathbf{G}_2^{Q^{(k)}} \neq \mathbf{0}$, we conclude that at least one column of $\mathbf{G}_2^{Q^{(k-1)}}$ is not
all-zero and thus $\mathbf{G}_2^{Q^{(k-1)}}$ is not a zero matrix. Subsequently we have, for all $0 \leq i \leq k$,
$\mathbf{G}_2^{Q^{(i)}} \neq \mathbf{0}$, which completes the proof.

\end{document}